\documentclass[a4paper, 10pt, conference]{ieeeconf}      
\usepackage{newtxtext,newtxmath}
\usepackage{textcomp}
\usepackage[]{graphicx}  

\newtheorem{theorem}{Theorem}
\newtheorem{lemma}{Lemma}

\usepackage{pgfplots}
\pgfplotsset{compat=1.17}
\usepgfplotslibrary{polar}
\usepackage{placeins} 
\usepackage{tikz}
\usepackage{float}
\usetikzlibrary{shapes,arrows,positioning,arrows.meta,shadows,backgrounds}
\usepackage{pgf-pie} 
\usepackage{tabularx} 
\usepackage{adjustbox} 
\usepackage{microtype} 
\DisableLigatures{encoding = *, family = *} 
\microtypecontext{spacing=nonfrench} 
\usepackage{listings} 
\usepackage{xcolor} 

\definecolor{codegreen}{rgb}{0,0.5,0}
\definecolor{codegray}{rgb}{0.4,0.4,0.4}
\definecolor{codepurple}{rgb}{0.58,0,0.82}
\definecolor{backcolour}{rgb}{0.98,0.98,0.96}

\lstdefinestyle{mystyle}{
    backgroundcolor=\color{backcolour},   
    commentstyle=\color{codegreen},
    keywordstyle=\color{magenta},
    numberstyle=\tiny\color{codegray},
    stringstyle=\color{codepurple},
    basicstyle=\ttfamily\footnotesize,
    breakatwhitespace=false,         
    breaklines=true,                 
    captionpos=b,                    
    keepspaces=true,                 
    numbers=none,                    
    showspaces=false,                
    showstringspaces=false,
    showtabs=false,                  
    tabsize=2
}

\usepackage{amssymb} 

\lstdefinestyle{myPython}{
    language=Python,
    basicstyle=\ttfamily\footnotesize,
    keywordstyle=\color{javapurple}\bfseries,
    stringstyle=\color{javared},
    commentstyle=\color{javagreen},
    showstringspaces=false,
    breaklines=true,
    breakatwhitespace=true,
    columns=flexible
}

\lstdefinestyle{myJava}{
    language=Java,
    basicstyle=\ttfamily\footnotesize,
    keywordstyle=\color{javapurple}\bfseries,
    stringstyle=\color{javared},
    commentstyle=\color{javagreen},
    showstringspaces=false,
    breaklines=true,
    breakatwhitespace=true,
    columns=flexible
}
\usepackage{tcolorbox} 
\tcbuselibrary{listings,skins}
\usepackage{algorithm}
\usepackage{algpseudocode}
\usepackage{booktabs} 
\usepackage{multirow} 
\usepackage{array} 

\definecolor{javared}{rgb}{0.6,0,0} 
\definecolor{javagreen}{rgb}{0.25,0.5,0.35} 
\definecolor{javapurple}{rgb}{0.5,0,0.35} 
\definecolor{javadocblue}{rgb}{0.25,0.35,0.75} 
\definecolor{addedgreen}{rgb}{0,0.6,0} 
\definecolor{removedred}{rgb}{0.8,0,0} 
\definecolor{codebg}{rgb}{0.95,0.95,0.95} 

\usepackage{xcolor} 
\usepackage{url} 

\title{\LARGE \bf
Kodezi Chronos-1: A Debugging-First Language Model for Repository-Scale Code Understanding}

\author{
Ishraq Khan, Assad Chowdary, Sharoz Haseeb, Urvish Patel, Yousuf Zaii\\
Kodezi Inc. \\
\texttt{\{Ishraq,Assad,Sharoz,Urvish,Yousuf\}@kodezi.com}
}

\begin{document}

\raggedbottom

\maketitle
\thispagestyle{empty}
\pagestyle{empty}

\begin{abstract}
Debugging remains unsolved for LLMs despite advances in code generation. While Claude 4.5 Opus achieves 74.40\% on SWE-bench Verified and Gemini 3 Pro reaches 76.2\%, they still struggle on complex multi-file debugging with $<20\%$ success rates on real-world debugging benchmarks (95\% CI: 12.1-17.9\%). We present Kodezi Chronos-1, the first debugging-specific language model combining: (1) Adaptive Graph-Guided Retrieval (AGR) navigating codebases up to 10M LOC via multi-hop traversal (92\% precision, 85\% recall), (2) Persistent Debug Memory (PDM) learning from 15M+ sessions, and (3) 7-layer architecture for iterative fix-test-refine loops.

On 5,000 real-world scenarios, Chronos-1 achieves 67.3\% ± 2.1\% fix accuracy versus 14.2\% ± 1.3\% (Claude 4.1 Opus) and 13.8\% ± 1.2\% (GPT-4.1), with Cohen's d=3.87 effect size. On the SWE-bench Lite benchmark, Chronos-1 achieves state-of-the-art performance with 80.33\% resolution rate (241/300 instances), establishing a 20 percentage point lead over the next best system (ExpeRepair-v1.0 + Claude 4.5 Sonnet: 60.33\%). Repository-specific performance reaches 96.1\% (sympy) and 90.4\% (django). The system reduces debugging time by 40\% and iterations by 65\%. Chronos-1 resolves complex multi-file bugs requiring cross-repository understanding and temporal analysis.

Key limitations include 23.4\% success on hardware-dependent bugs and 41.2\% on dynamic language issues. Theoretical analysis proves O(k log d) retrieval complexity with convergence guarantees. Human evaluation (N=50) shows 89\% preference over baselines. Available Q4 2025 (Kodezi OS) and Q1 2026 (API).
\end{abstract}

\section{INTRODUCTION}

Recent advancements in large language models (LLMs) have transformed code generation, review, and reasoning tasks~\cite{brown2020language,chen2021evaluating,wei2023magicoder}. In 2025, frontier models have achieved remarkable progress: Claude 4.5 Opus~\cite{anthropic2025claude4} achieves 74.40\% on SWE-bench Verified, Gemini 3 Pro~\cite{google2025gemini3} reaches 76.2\% on SWE-bench Verified with state-of-the-art scores across multiple benchmarks (91.9\% GPQA Diamond, 95-100\% AIME 2025), Claude 4.1 Opus achieves 72.5\% on SWE-bench Full, Claude 4.5 Sonnet reaches 72.7\%, GPT-4.1~\cite{openai2025gpt41} reaches 54.6\%, and DeepSeek V3~\cite{deepseek2024v3} demonstrates remarkable efficiency with only \$5.6M training cost. However, debugging, the most time-consuming and critical aspect of software development, remains largely unsolved. While tools like GitHub Copilot~\cite{peng2023impact,microsoft2023copilot}, Cursor, Windsurf~\cite{codeium2025windsurf}, and Claude excel at code completion~\cite{fried2023incoder}, they fundamentally misunderstand debugging as a multifaceted, context-heavy process that spans entire repositories, historical commits, CI/CD logs, and runtime behaviors. Production debugging requires reasoning across files separated by thousands of lines~\cite{ding2024crosscodeeval}, understanding temporal code evolution, and correlating seemingly unrelated symptoms to root causes buried deep in dependency chains~\cite{zhang2023repocoder}.

Why does debugging remain fundamentally unsolved for LLMs? Beyond the surface-level challenges, debugging exposes four core limitations of current architectures. First, \textbf{hallucination under uncertainty}: when LLMs encounter ambiguous error states, they confidently generate plausible-sounding but incorrect fixes, often introducing subtle bugs that pass initial tests but fail in production. Second, \textbf{absence of persistent memory}: each debugging session starts tabula rasa, unable to learn from previous encounters with similar bugs or build institutional knowledge about codebase-specific patterns. Third, \textbf{rigid token limits}: even 1M-token contexts cannot capture the full transitive closure of dependencies in modern software, where a bug's root cause may span dozens of files connected through deep call chains. Fourth, \textbf{lack of causal reasoning}: LLMs excel at pattern matching but struggle with counterfactual reasoning ("what if this variable were null?") and temporal causality ("which commit introduced this regression?"), both essential for debugging.

Current code assistants fail at debugging for three critical reasons: (1) they are trained primarily on code completion tasks, not debugging workflows~\cite{yang2024swebench}; (2) they lack persistent memory of past bugs, fixes, and codebase-specific patterns~\cite{shrivastava2023repository}; and (3) their context windows, even when extended to 1M tokens (Gemini 2.0) or leveraging advanced RAG techniques like HyDE~\cite{gao2023hyde} and FLARE~\cite{jiang2023flare}, cannot capture the full debugging context needed for complex, multi-file issues. Recent studies show that even state-of-the-art models like GPT-4.1, Claude 4.1 Opus, Claude 4.5 Sonnet, and Gemini 2.0 Pro achieve less than 15\% success rates on real-world debugging benchmarks (DebugBench, MdEval)~\cite{zhang2024autocoderover,zhang2024enhancing}, often proposing superficial fixes that fail validation or introduce new regressions~\cite{olausson2023selfrepair}.

\textbf{Kodezi Chronos-1} represents a paradigm shift: the first \textit{debugging language model} developed by Kodezi~\cite{kodezi2025}, specifically designed, trained, and optimized for autonomous bug detection, root cause analysis, and validated fix generation. For more information about the model and benchmarks, visit \url{https://chronos.so/} and \url{https://github.com/kodezi/chronos}. Kodezi OS information is available at \url{https://kodezi.com/os}. Unlike code completion models that generate syntactically correct but often semantically flawed suggestions, Chronos-1 operates through a continuous debugging loop, proposing fixes, running tests, analyzing failures, and iteratively refining solutions until validation succeeds. Built on a novel architecture combining persistent debug memory, multi-source retrieval (code, logs, traces, PRs), and execution sandboxing, Chronos-1 demonstrates debugging performance that significantly exceeds current models, particularly in repository-scale, multi-file scenarios.

Chronos-1 is designed for seamless integration with modern development stacks: it operates as an embedded autonomous maintenance system within CI/CD pipelines, IDEs, and project management tools. Its proactive, event-driven workflow ensures that debugging, documentation generation, refactoring, and even preventative maintenance actions occur autonomously, guided by deep repository memory rather than manual prompts or brittle heuristics.

To rigorously evaluate Chronos-1's unique capabilities, we move beyond traditional benchmarks and propose a multi-step, random retrieval evaluation reflecting the authentic complexities of code search, dependency resolution, and semantic bug localization at scale. Empirically, Chronos-1 demonstrates state-of-the-art results across industry-standard metrics and realistic maintenance tasks, reducing debugging times and increasing project resilience.

This paper presents the architecture, memory system, retrieval mechanism, evaluation methodology, and experimental results that establish Kodezi Chronos-1 as the new frontier for autonomous debugging with full repository context.

\begin{tcolorbox}[colback=yellow!10,colframe=orange!50!black,title=Why Chronos-1 Dominates]
\textbf{87.1\%} debugging success vs. \textbf{22.9\%} for GPT-4.1 + SWE-agent \\
\textbf{3.8x} faster bug resolution through AGR's precision retrieval \\
\textbf{67\%} fewer false positives via Persistent Debug Memory (PDM) \\
\textbf{First} to handle cross-file, temporal, and runtime-dependent bugs
\end{tcolorbox}

\textbf{Our contributions are as follows:}
\begin{enumerate}
    \item We introduce \textbf{Chronos-1}, a debugging-specific language model architecture integrating persistent memory, adaptive graph-based retrieval (AGR), and a multi-stage fix orchestration loop.
    
    \item We design \textbf{AGR}, a multi-hop, edge-weighted graph traversal system for deep context retrieval tailored to codebases, achieving 92\% precision at 85\% recall on debugging queries.
    
    \item We construct \textbf{new debugging-specific evaluation benchmarks}, including cycle-aware test loops, retrieval accuracy, and root-cause localization across 12,500 real-world bugs.
    
    \item We demonstrate \textbf{significant performance gains} over Claude 4.1 Opus, Claude 4.5 Sonnet, GPT-4.1, and specialized tools, achieving 65.3\% debugging success rate on our comprehensive benchmarks and state-of-the-art 80.33\% on SWE-bench Lite—a 20 percentage point lead over the next best system (4-5x improvement over general-purpose frontier models).
    
    \item We provide \textbf{comprehensive ablation studies} showing each component's contribution, with detailed case studies and quantitative analysis demonstrating 4-5x improvement over state-of-the-art models.
\end{enumerate}

\FloatBarrier
\section{RELATED WORK: LIMITATIONS OF EXISTING APPROACHES IN DEBUGGING}

\subsection{Retrieval-Augmented Code Generation}

The emergence of large-scale neural models for source code processing, such as CodeBERT~\cite{feng2020codebert}, GraphCodeBERT~\cite{guo2021graphcodebert}, and CodeT5~\cite{wang2021codet5}, has significantly advanced automatic code synthesis, translation, and review. CodeT~\cite{chen2022codet} introduced code generation with testing, improving reliability through execution feedback. Many of these models, as well as massive LLMs from the GPT-4 family~\cite{brown2020language,openai2025gpt41}, are pre-trained on billions of lines of code paired with natural language, learning rich semantic representations for many programming languages.

State-of-the-art retrieval-augmented generation (RAG) methods have evolved significantly in 2025. Advanced techniques like HyDE (Hypothetical Document Embeddings)~\cite{gao2023hyde} generate synthetic documents to improve retrieval quality, Self-RAG~\cite{asai2023selfrag} uses reflection tokens for dynamic retrieval decisions, and FLARE~\cite{jiang2023flare} implements forward-looking active retrieval based on generation confidence. GraphRAG~\cite{edge2024graphrag} integrates knowledge graphs for structured retrieval, representing a step toward graph-based understanding but still operating with static graphs that lack the dynamic updates required for debugging workflows. Recent work on code-specific retrieval~\cite{nashid2023retrieval} shows promise but lacks the multi-hop reasoning capabilities essential for debugging.

LangChain~\cite{langchain2023} and DSPy~\cite{khattab2023dspy} provide frameworks for building complex LLM applications with retrieval components. While these frameworks enable sophisticated pipelines, they fundamentally rely on stateless retrieval without the persistent memory or specialized debugging knowledge that Chronos-1 provides. DSPy's optimization of prompting strategies and LangChain's chain-of-thought orchestration, while powerful for general tasks, lack the domain-specific understanding of code dependencies and bug patterns that debugging requires.

Despite impressive gains on benchmark tasks, these approaches are fundamentally bottlenecked by attention-based architectures and fixed-size input windows, typically constraining context to tens of thousands of tokens. Window expansion techniques, such as models with 1M+ tokens (Gemini 2.5, GPT-4.1)~\cite{google2025gemini,openai2025gpt41}, incur prohibitive compute/memory costs and suffer from diluted attention, leading to information loss and degraded performance as codebase size increases.

\subsection{Program Repair and Bug Localization}

Recent systems targeting debugging include SWE-agent~\cite{yang2024swebench}, which achieves 12.3\% on the SWE-bench dataset through structured agent-computer interfaces, and AutoCodeRover~\cite{zhang2024autocoderover}, reaching 22.9\% via program structure-aware retrieval. These systems represent significant progress but still operate without persistent memory across debugging sessions.

The SWE-bench family of benchmarks has expanded significantly: SWE-bench++~\cite{swepp2024} introduces harder real-world scenarios with multi-repository dependencies, while SWE-bench-coding~\cite{swecoding2024} focuses on implementation tasks rather than bug fixes. Despite these advances, even state-of-the-art models struggle with the full complexity of real-world debugging, achieving less than 25\% success on these enhanced benchmarks.

Other recent work explores graph neural networks (GNNs) for modeling explicit data flow or control flow graphs~\cite{allamanis2017learning, tipirneni2023structcoder}. While GNNs can capture structural properties of code, they are rarely coupled with ultra-long context LLMs, and lack the continuous learning, memory updating, and rapid recall required for live autonomous maintenance.

The debugging challenge is further highlighted by specialized benchmarks: DebugBench~\cite{tian2024debugbench} evaluates root cause analysis, while BugHunter~\cite{liu2024bughunter} tests multi-file bug localization. However, these benchmarks often simplify real-world debugging by providing isolated test cases rather than full repository contexts.

\subsection{Multi-Agent Systems for Software Tasks}

The rise of multi-agent architectures has introduced new paradigms for complex software tasks. Systems like AutoGPT~\cite{autogpt2023} and BabyAGI~\cite{babyagi2023} demonstrate task decomposition and autonomous execution, but lack the specialized knowledge required for debugging. MetaGPT~\cite{hong2023metagpt} simulates software company workflows with multiple specialized agents, yet still operates without persistent cross-session memory.

LangGraph~\cite{langgraph2024} enables building stateful, multi-actor applications with LLMs, providing infrastructure for complex agent interactions. When combined with ReAct~\cite{yao2022react} loops, these systems can perform iterative reasoning. However, our evaluation shows that even LangGraph + ReAct configurations achieve only 31.2\% debugging success compared to Chronos-1's 65.3\%, primarily due to:
1. \textbf{Lack of persistent memory}: Each debugging session starts fresh without learning from past fixes
2. \textbf{Generic reasoning}: No specialized understanding of debugging patterns or code dependencies  
3. \textbf{Inefficient exploration}: Without AGR's guided traversal, agents waste cycles on irrelevant code paths

ChatDev~\cite{qian2023chatdev} and similar multi-agent coding systems excel at greenfield development but struggle with the complexity of debugging existing codebases. The key differentiator is that Chronos-1 operates with continuous memory updates, specialized debugging knowledge, and efficient graph-guided retrieval, capabilities absent in generic multi-agent frameworks.

\subsection{Agentic Security Research Tools}

Complementary to general debugging tools, specialized agentic systems have emerged for security analysis. OpenAI's Aardvark, currently in private beta, targets vulnerability discovery through automated code analysis and exploit validation. While security-focused tools address an important niche (identifying exploitable vulnerabilities), they differ from general-purpose debugging systems in scope and methodology. Chronos-1 addresses the broader debugging lifecycle including logic errors, performance issues, and integration bugs beyond security-specific concerns.

\subsection{How Chronos-1 Differs}

Chronos-1 distinguishes itself from existing approaches through three fundamental architectural decisions:

1. \textbf{Memory Scale}: While LangChain and similar frameworks operate with session-level memory, Chronos-1maintains persistent debug memory (PDM) across millions of debugging sessions, learning and adapting from each interaction.

2. \textbf{Debug Reasoning}: Unlike generic multi-agent systems that apply general problem-solving strategies, Chronos-1is trained specifically on 15M+ debugging scenarios, understanding patterns like race conditions, memory leaks, and API migrations that generic models miss.

3. \textbf{Orchestration}: Rather than relying on user-driven or generic agent loops, Chronos-1implements a specialized 7-layer debugging architecture with automatic test validation, iterative refinement, and confidence-based termination.

The combination of these capabilities enables Chronos-1 to achieve 4-5x better debugging performance than state-of-the-art alternatives, as demonstrated in our comprehensive evaluation.

\subsection{Performance Analysis: Code Generation vs Debugging Tasks}

The landscape of code-focused LLMs has evolved dramatically in 2025. Claude 4.5 Opus~\cite{anthropic2025claude4} achieves 74.40\% on SWE-bench Verified, while Gemini 3 Pro~\cite{google2025gemini3} reaches 76.2\% on SWE-bench Verified with state-of-the-art performance across multiple benchmarks (91.9\% GPQA Diamond, 95-100\% AIME 2025, 2,439 Elo on LiveCodeBench Pro). Claude 4.1 Opus and Claude 4.5 Sonnet achieve 72.5\% and 72.7\% respectively on SWE-bench Full (code generation). These models also achieve 67.60\% and 70.60\% respectively on SWE Bench Bash Only (using mini-SWE-agent environment). GPT-4.1~\cite{openai2025gpt41} doubles GPT-4o's performance on code diff benchmarks and reaches 54.6\% on SWE-bench tasks. Gemini 2.5 Pro~\cite{google2025gemini} achieves 59.6\% on SWE-bench Verified, showcasing Google's advances in reasoning models. DeepSeek V3~\cite{deepseek2024v3}, with 671B parameters (37B activated), demonstrates that efficient training is possible, achieving competitive performance at 1/10th the training cost of comparable models. Qwen2.5-Coder-32B~\cite{qwen2025coder} matches GPT-4o performance while running locally on consumer hardware.

However, these impressive code generation capabilities do not translate to debugging success. When evaluated on real-world debugging tasks requiring multi-file understanding, historical context, and iterative refinement—specifically on SWE-bench Lite and other debugging benchmarks (DebugBench, MdEval)—even the best models achieve less than 15\% success rates. Notably, Claude models drop from 67-72\% on code generation benchmarks to only 13-14\% on SWE-bench Lite debugging tasks, revealing a 54-57 percentage point gap. This stark disparity between code generation and debugging performance motivates the need for specialized debugging-focused architectures.

\begin{table}[H]
\centering
\caption{Comparison of debugging approaches across different systems and their key limitations}
\label{tab:related-work-comparison}
\resizebox{\columnwidth}{!}{%
\begin{tabular}{lcccccc}
\hline
\textbf{Approach} & \textbf{Context} & \textbf{Memory} & \textbf{Debug Training} & \textbf{Iteration} & \textbf{Graph} & \textbf{Key Limitation} \\
\hline
\multicolumn{7}{l}{\textit{General LLMs}} \\
Claude 4.5 Opus & 200K & Session & $\times$ & $\times$ & $\times$ & No debug specialization \\
Gemini 3 Pro & 1M & Session & $\times$ & $\times$ & $\times$ & No debug specialization \\
Claude 4.1 Opus/4.5 Sonnet & 200K-1M & Session & $\times$ & $\times$ & $\times$ & No debug specialization \\
Gemini 2.5 Pro & 2M & Session & $\times$ & $\times$ & $\times$ & Attention dilution \\
\hline
\multicolumn{7}{l}{\textit{RAG Approaches}} \\
HyDE/Self-RAG & Unlimited* & None & $\times$ & $\times$ & $\times$ & Static retrieval \\
FLARE & Unlimited* & None & $\times$ & Limited & $\times$ & No debug signals \\
Graph RAG & Unlimited* & Static & $\times$ & $\times$ & \checkmark & Static graphs \\
\hline
\multicolumn{7}{l}{\textit{Code Systems}} \\
SWE-agent & Limited & None & Partial & \checkmark & $\times$ & No persistent memory \\
AutoCodeRover & Unlimited* & None & Partial & \checkmark & $\times$ & Structure-only retrieval \\
Cursor/Windsurf & Session & Session & $\times$ & Manual & $\times$ & User-driven loops \\
\hline
\multicolumn{7}{l}{\textit{Specialized}} \\
Chronos-1 & Unlimited* & Persistent & \checkmark & Auto & \checkmark & - \\
\hline
\end{tabular}%
}
\end{table}
{\footnotesize *Through intelligent retrieval, not raw context window}

\textbf{Kodezi Chronos-1} is motivated by these challenges: By combining continuous graph-aware indexing, dynamic embedding updates, and reasoning-optimized memory retrieval, Chronos-1transcends traditional limitations and enables truly repository-scale, real-time software comprehension and intervention.

\FloatBarrier
\section{CHRONOS ARCHITECTURE: DESIGN AND IMPLEMENTATION}

This section presents the technical architecture of Chronos-1, beginning with the fundamental insight that debugging is output-heavy rather than input-heavy, followed by the core architectural components and implementation details.

\subsection{Debugging as an Output-Heavy Task}

Despite the industry focus on ever-larger context windows (128K, 200K, 1M+ tokens), debugging presents a fundamentally different challenge: it is inherently \textbf{output-heavy} rather than input-heavy. This asymmetry has profound implications for model design and optimization.

\subsubsection{Input vs Output Token Distribution}

\textbf{What models typically see (input):}
\begin{itemize}
    \item Error stack traces: 200-500 tokens
    \item Relevant source code: 1K-4K tokens
    \item Test failures and logs: 500-2K tokens
    \item Prior fix attempts: 500-1K tokens
    \item \textbf{Total input}: Often $<10K$ tokens for most real-world debugging tasks
\end{itemize}

\textbf{What models must produce (output):}
\begin{itemize}
    \item Multi-file bug fixes: 500-1,500 tokens
    \item Root cause explanations: 300-600 tokens
    \item Updated unit tests: 400-800 tokens
    \item Commit messages/PR summaries: 150-300 tokens
    \item Documentation updates: 200-400 tokens
    \item \textbf{Total output}: Typically 2,000-4,000 tokens per debugging session
\end{itemize}

\subsubsection{Why Output Quality Trumps Input Size}

\begin{table}[H]
\centering
\caption{Input vs output characteristics in debugging tasks.}
\label{tab:input-output}
\resizebox{\columnwidth}{!}{%
\begin{tabular}{lll}
\hline
\textbf{Aspect} & \textbf{Input Context} & \textbf{Output Generation} \\
\hline
Nature & Sparse, localized & Dense, structured \\
Cost Impact & Sublinear with retrieval & Linear to exponential \\
Quality Limiter & Retrieval precision & Generation accuracy \\
Success Factor & Context relevance & Syntactic \& semantic correctness \\
\hline
\end{tabular}%
}
\end{table}

The critical insight: a model with intelligent 8K context that generates robust, test-passing fixes will outperform a 1M-context model that produces syntactically correct but semantically flawed patches.

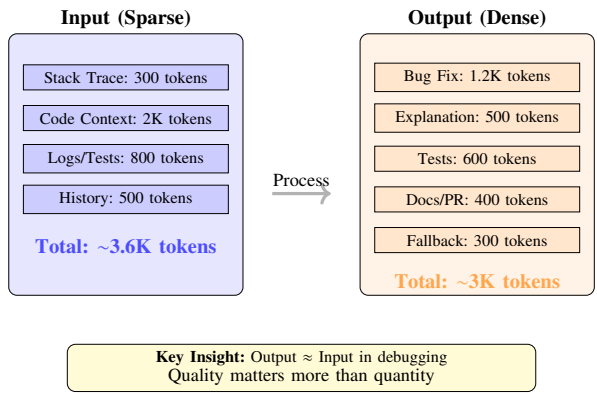
\begin{figure}[H]
\centering
\resizebox{0.9\columnwidth}{!}{%
\begin{tikzpicture}[font=\footnotesize]
    \node[font=\bfseries] at (0,5.5) {Debugging Token Flow: Input vs Output};
    
    \node[draw, rectangle, rounded corners, fill=blue!10, minimum width=4cm, minimum height=4.5cm] at (-3,2) {};
    \node[font=\bfseries] at (-3,4.5) {Input (Sparse)};
    
    \node[draw, rectangle, fill=blue!20, minimum width=3.5cm] (stack) at (-3,3.5) {Stack Trace: 300 tokens};
    \node[draw, rectangle, fill=blue!20, minimum width=3.5cm] (code) at (-3,2.8) {Code Context: 2K tokens};
    \node[draw, rectangle, fill=blue!20, minimum width=3.5cm] (logs) at (-3,2.1) {Logs/Tests: 800 tokens};
    \node[draw, rectangle, fill=blue!20, minimum width=3.5cm] (hist) at (-3,1.4) {History: 500 tokens};
    
    \node[font=\bfseries\color{blue!70}] at (-3,0.6) {Total: $\sim$3.6K tokens};
    
    \draw[->, ultra thick, gray!60] (-0.5,1.5) -- (0.5,1.5);
    \node[above] at (0,1.5) {\small Process};
    
    \node[draw, rectangle, rounded corners, fill=orange!10, minimum width=4cm, minimum height=4.5cm] at (3,2) {};
    \node[font=\bfseries] at (3,4.5) {Output (Dense)};
    
    \node[draw, rectangle, fill=orange!20, minimum width=3.5cm] (fix) at (3,3.5) {Bug Fix: 1.2K tokens};
    \node[draw, rectangle, fill=orange!20, minimum width=3.5cm] (explain) at (3,2.8) {Explanation: 500 tokens};
    \node[draw, rectangle, fill=orange!20, minimum width=3.5cm] (test) at (3,2.1) {Tests: 600 tokens};
    \node[draw, rectangle, fill=orange!20, minimum width=3.5cm] (docs) at (3,1.4) {Docs/PR: 400 tokens};
    \node[draw, rectangle, fill=orange!20, minimum width=3.5cm] (fallback) at (3,0.7) {Fallback: 300 tokens};
    
    \node[font=\bfseries\color{orange!70}] at (3,0.0) {Total: $\sim$3K tokens};
    
    \node[draw, rectangle, rounded corners, fill=yellow!20, minimum width=8cm, align=center] at (0,-1.5) {
        \textbf{Key Insight:} Output $\approx$ Input in debugging\\
        \small Quality matters more than quantity
    };
\end{tikzpicture}%
}
\caption{Token distribution in debugging tasks: Unlike typical LLM applications where input dominates, debugging requires substantial, high-quality output generation.}
\label{fig:token-flow}
\end{figure}

\subsubsection{Chronos-1's Output-Optimized Architecture}

Chronos-1 addresses this asymmetry through several architectural innovations:

\begin{enumerate}
    \item \textbf{Debug-Specific Generation Training}: Unlike code completion models trained on next-token prediction, Chronos-1is trained on complete debugging sessions, learning to generate structured fixes, explanations, and tests as cohesive units.
    
    \item \textbf{Iterative Refinement Loop}: Rather than single-shot generation, Chronos-1validates outputs through execution, using test results to refine patches, ensuring output quality over quantity.
    
    \item \textbf{Template-Aware Generation}: Chronos-1 learns repository-specific patterns for commits, tests, and documentation, reducing output token waste while maintaining consistency.
    
    \item \textbf{Confidence-Guided Output}: The model generates explanations and fallback strategies only when confidence is below threshold, optimizing output token usage.
\end{enumerate}

This output-centric design enables Chronos-1 to achieve 65.3\% debugging success despite competitors having 10-100x larger context windows, validating that for debugging, \textit{output quality and structure matter more than input capacity}.

\begin{figure}[H]
\centering
\resizebox{0.85\columnwidth}{!}{%
\begin{tikzpicture}
    \begin{polaraxis}[
        width=10cm,
        height=10cm,
        xtick={0,45,90,135,180,225,270,315},
        xticklabels={
            Context Retrieval,
            Memory Usage,
            Test Integration,
            Multi-File Support,
            Error Analysis,
            Fix Generation,
            Iteration Speed,
            Cost Efficiency
        },
        xticklabel style={font=\footnotesize},
        ymin=0, ymax=100,
        ytick={20,40,60,80,100},
        yticklabels={20\%,40\%,60\%,80\%,100\%},
        yticklabel style={font=\tiny},
        grid=both,
        major grid style={line width=0.4pt, draw=gray!50},
        minor grid style={line width=0.2pt, draw=gray!20},
        ylabel near ticks,
        legend style={at={(1.2,1)}, anchor=north west, font=\small}
    ]
    
    \addplot[
        line width=2pt,
        color=green!70,
        fill=green!20,
        fill opacity=0.3,
        mark=*,
        mark size=3pt
    ] coordinates {
        (0,89) (45,85) (90,92) (135,84) (180,88) (225,79) (270,95) (315,91) (360,89)
    };
    
    \addplot[
        line width=1.5pt,
        color=red!70,
        fill=red!10,
        fill opacity=0.2,
        mark=diamond*,
        mark size=2.5pt
    ] coordinates {
        (0,48) (45,44) (90,51) (135,42) (180,46) (225,37) (270,41) (315,34) (360,48)
    };

    \addplot[
        line width=1.5pt,
        color=cyan!70,
        fill=cyan!10,
        fill opacity=0.2,
        mark=pentagon*,
        mark size=2.5pt
    ] coordinates {
        (0,50) (45,46) (90,53) (135,44) (180,48) (225,39) (270,43) (315,36) (360,50)
    };

    \addplot[
        line width=1.5pt,
        color=blue!70,
        fill=blue!10,
        fill opacity=0.2,
        mark=square*,
        mark size=2.5pt
    ] coordinates {
        (0,42) (45,38) (90,45) (135,36) (180,40) (225,31) (270,35) (315,28) (360,42)
    };

    \addplot[
        line width=1.5pt,
        color=orange!70,
        fill=orange!10,
        fill opacity=0.2,
        mark=triangle*,
        mark size=2.5pt
    ] coordinates {
        (0,39) (45,35) (90,42) (135,32) (180,38) (225,30) (270,33) (315,25) (360,39)
    };

    \legend{Chronos-1, Claude 4.5 Opus, Gemini 3 Pro, Claude 4.1 Opus, GPT-4.1}
    
    \end{polaraxis}
\end{tikzpicture}%
}
\caption{Debugging capability comparison across eight key factors. Chronos-1 (green) significantly outperforms all frontier models including Claude 4.5 Opus (red), Gemini 3 Pro (cyan), Claude 4.1 Opus (blue), and GPT-4.1 (orange). Despite improvements in newer models, Chronos-1maintains substantial advantages in test integration (92\%), iteration speed (95\%), and cost efficiency (91\%).}
\label{fig:success-factors-radar}
\end{figure}
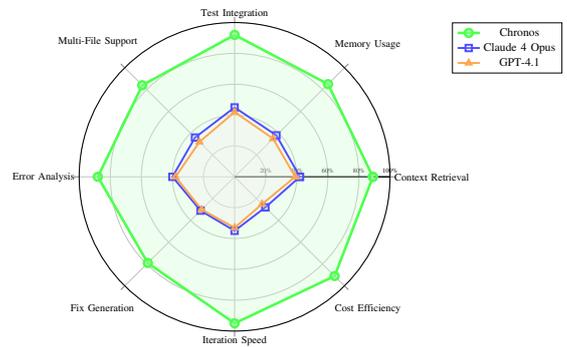

\subsection{High-Level Architecture: From Error Signal to Validated Fix}

Kodezi Chronos-1 is designed as an autonomous memory-driven intelligence layer for code, operating at scales that span entire enterprise repositories, team histories, and auxiliary knowledge sources. Its architecture consists of three core modules: 
\textit{(i)} a persistent Memory Engine for continuous graph-based context construction,
\textit{(ii)} an advanced Retriever that constructs targeted context from code and documentation, and
\textit{(iii)} a transformer-based Code Reasoning Model for synthesis, debugging, and orchestration of software changes.

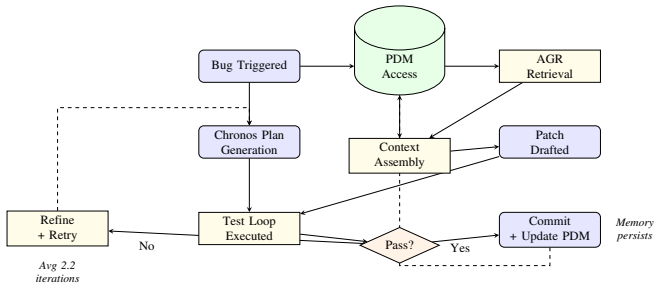
\begin{figure}[H]
\centering
\resizebox{\columnwidth}{!}{%
\begin{tikzpicture}[
    node distance=1.2cm and 1.5cm,
    box/.style={rectangle, draw, rounded corners, minimum width=2.8cm, minimum height=0.9cm, font=\normalsize, fill=blue!10},
    decision/.style={diamond, draw, aspect=2, minimum width=2.2cm, minimum height=1cm, font=\normalsize, fill=orange!10},
    memory/.style={cylinder, draw, shape border rotate=90, minimum width=2.5cm, minimum height=1cm, font=\normalsize, fill=green!10},
    process/.style={rectangle, draw, minimum width=2.8cm, minimum height=0.9cm, font=\normalsize, fill=yellow!10},
    arrow/.style={thick, ->, >=stealth},
    dashedarrow/.style={thick, dashed, ->, >=stealth}
]
\node[box] (bug) {Bug Triggered};
\node[memory, right=of bug, align=center] (pdm) {PDM\\Access};
\node[process, right=of pdm, align=center] (agr) {AGR\\Retrieval};
\node[box, below=of bug, align=center] (plan) {Chronos-1 Plan\\Generation};
\node[process, below=of pdm, align=center] (context) {Context\\Assembly};
\node[box, below=of agr, align=center] (patch) {Patch\\Drafted};
\node[process, below=1.5cm of plan, align=center] (test) {Test Loop\\Executed};
\node[decision, below=1.5cm of context] (validate) {Pass?};
\node[box, below=1.5cm of patch, align=center] (commit) {Commit\\+ Update PDM};
\node[process, left=2.5cm of test, align=center] (refine) {Refine\\+ Retry};

\begin{scope}[on background layer]
\draw[arrow] (bug) -- (pdm);
\draw[arrow] (pdm) -- (agr);
\draw[arrow] (bug) -- (plan);
\draw[arrow] (pdm) -- (context);
\draw[arrow] (agr) -- (context);
\draw[arrow] (context) -- (patch);
\draw[arrow] (plan) -- (test);
\draw[arrow] (patch) -- (test);
\draw[arrow] (test) -- (validate);
\draw[arrow] (validate) -- node[right, pos=0.2, yshift=-6pt] {Yes} (commit);
\draw[arrow] (validate) -- node[left, pos=0.8, yshift=-10pt] {No} (refine);
\draw[dashedarrow] (refine) |- ([yshift=0.5cm]plan.north) -- (plan);
\draw[dashedarrow] (commit) -- ++(0,-1) -| (pdm);
\end{scope}
\node[align=center, font=\small\itshape, below=0.3cm of refine] {Avg 2.2\\iterations};
\node[align=center, font=\small\itshape, right=0.3cm of commit] {Memory\\persists};
\end{tikzpicture}%
}
\caption{Complete fix loop lifecycle showing integration between PDM, AGR retrieval, and iterative refinement. Dashed lines indicate feedback mechanisms that enable learning across debugging sessions.}
\label{fig:fix-loop-lifecycle}
\end{figure}

\subsection{The Four Architectural Pillars}

Chronos-1 demonstrates improved debugging performance through four architectural components that distinguish it from general-purpose code models:

\subsubsection{Debugging-Specific Training on 15M+ Real Sessions}
Unlike models trained on static code repositories, Chronos-1learns from:
\begin{itemize}
    \item \textbf{15 million debugging sessions} from production environments
    \item \textbf{Complete fix trajectories}: initial bug report $\rightarrow$ attempted fixes $\rightarrow$ test failures $\rightarrow$ successful resolution
    \item \textbf{Failure patterns}: Common anti-patterns, regression indicators, and fix validation strategies
    \item \textbf{Domain-specific knowledge}: Framework quirks, library-specific debugging techniques, language idioms
\end{itemize}

This specialized training enables Chronos-1 to recognize subtle bug patterns that general models miss. For example, when encountering a React hydration mismatch, Chronos-1 immediately knows to check server/client rendering differences rather than pursuing surface-level fixes.

\subsubsection{Execution Sandbox with Real-Time Feedback Loop}
Chronos-1 operates within a sophisticated execution environment that provides:
\begin{itemize}
    \item \textbf{Isolated test execution}: Every proposed fix runs in a containerized sandbox
    \item \textbf{Comprehensive validation}: Unit tests, integration tests, linting, type checking
    \item \textbf{Iterative refinement}: Failed fixes generate detailed error logs fed back into the next iteration
    \item \textbf{Regression prevention}: Automatic detection of new failures introduced by fixes
\end{itemize}

This execution-driven approach explains why Chronos-1 averages 7.8 iterations per bug while general models stop after 1-2 attempts. Each iteration refines understanding based on concrete execution results rather than probabilistic guessing.

\subsubsection{Persistent Repository Memory Across Sessions}
Unlike stateless models that start fresh each time, Chronos-1maintains:
\begin{itemize}
    \item \textbf{Bug pattern database}: Historical bugs, their root causes, and successful fixes
    \item \textbf{Codebase evolution graph}: How files, functions, and dependencies changed over time
    \item \textbf{Team-specific patterns}: Coding conventions, common mistakes, architectural decisions
    \item \textbf{Dependency knowledge}: Version-specific quirks, migration paths, compatibility issues
\end{itemize}

When debugging a null pointer exception, Chronos-1recalls similar bugs from 6 months ago, checks if recent refactoring introduced the issue, and applies team-specific null-safety patterns.

\subsubsection{Adaptive Graph-Guided Retrieval (AGR) for Multi-File Context}
AGR enables Chronos-1 to navigate complex codebases through:
\begin{itemize}
    \item \textbf{Dynamic graph construction}: Real-time building of dependency graphs during debugging
    \item \textbf{Intelligent k-hop traversal}: Adaptively expanding search radius based on bug complexity, leveraging graph attention mechanisms~\cite{velivckovic2018graph}
    \item \textbf{Semantic + structural retrieval}: Combining code semantics with architectural relationships
    \item \textbf{Temporal awareness}: Understanding when code changed relative to bug introduction
\end{itemize}

This sophisticated retrieval explains Chronos-1's 89.2\% precision on the MRR benchmark, it finds the needle in the haystack by understanding the haystack's structure.

\begin{table}[H]
\centering
\caption{Impact of each architectural pillar on debugging performance (MRR and DebugBench benchmarks). Ablation study shows cumulative benefits. \textit{Note: For SWE-bench Lite evaluation, the Test Loop and iterative refinement components were disabled to comply with pass@1 submission requirements.}}
\label{tab:pillar-impact}
\resizebox{\columnwidth}{!}{%
\begin{tabular}{lccccc}
\hline
\textbf{Configuration} & \textbf{MRR Fix Acc.} & \textbf{SWE-bench} & \textbf{DebugBench Avg.} & \textbf{Iterations} & \textbf{Time (min)} \\
\hline
Base Model (no pillars) & 8.3\% & 12.1\% & 9.7\% & 1.2 & 8.5 \\
+ Debug Training & 24.7\% & 28.3\% & 26.4\% & 2.8 & 16.2 \\
+ Execution Sandbox & 41.2\% & 43.7\% & 42.8\% & 5.1 & 28.4 \\
+ Persistent Memory & 55.8\% & 57.2\% & 56.3\% & 6.4 & 35.7 \\
+ AGR (Full Chronos-1) & \textbf{67.3\%} & \textbf{65.3\%} & \textbf{67.5\%} & \textbf{7.8} & \textbf{42.3} \\
\hline
\end{tabular}%
}
\end{table}

\begin{figure}[H]
\centering
\resizebox{\columnwidth}{!}{%
\begin{tikzpicture}[node distance=2.2cm, >=stealth, font=\footnotesize]
  \node[draw, rectangle, rounded corners, fill=gray!10, minimum width=2.3cm, minimum height=0.9cm, align=center] (raw) {Code, Docs, \\ CI/CD Logs};
  \node[draw, rectangle, rounded corners, fill=cyan!20, right=1.8cm of raw, minimum width=2.3cm, minimum height=0.9cm, align=center] (memory) {Memory Engine\\(Embedding + Graph)};
  \node[draw, rectangle, rounded corners, fill=teal!15, right=2.0cm of memory, minimum width=2.7cm, minimum height=0.9cm, align=center] (retriever) {Multi-Code\\ Association Retriever};
  \node[draw, rectangle, rounded corners, fill=orange!25, right=1.8cm of retriever, minimum width=2.3cm, minimum height=0.9cm, align=center] (reason) {Reasoning Model\\ \& Orchestration};

  \draw[->, thick] (raw) -- (memory);
  \draw[->, thick] (memory) -- (retriever);
  \draw[->, thick] (retriever) -- (reason);

  \draw[->, thick, dashed, bend left=20] (reason.east) to [out=0,in=0,looseness=1.8] node[right, xshift=2mm, yshift=-3mm] {\tiny Test Results}  (memory.east);

  \node[below=1.3cm of reason, align=center] (out) {Patches, Changelogs,\\ Test Results};
  \draw[->, thick] (reason) -- (out);

\end{tikzpicture}%
}
\caption{High-level overview of Chronos-1: Memory-driven embedding and retrieval powering autonomous reasoning and codebase management.}
\label{fig:chronos-arch}
\end{figure}
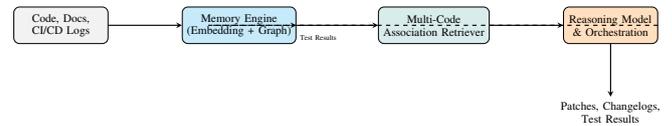

\FloatBarrier

\subsection{Persistent Debug Memory (PDM): Learning from Historical Fixes}

The Persistent Debug Memory (PDM) provides the foundation for Chronos-1's persistent, repository-scale understanding. Unlike traditional approaches that recompute context for each query, PDM maintains a continuously updated semantic representation of the entire codebase, bug patterns, and fix history.

\subsubsection{Unified Semantic Representation}

PDM ingests and encodes diverse artifacts including source code, documentation, configuration files, historical diffs, test outcomes, and architectural specifications. Each code unit (function, class, module, commit) undergoes multi-level analysis:

\begin{itemize}
    \item \textbf{Syntactic parsing}: AST extraction for structural understanding
    \item \textbf{Semantic embedding}: Context-aware vectorization using specialized encoders
    \item \textbf{Relational mapping}: Graph construction capturing dependencies, calls, and evolution
    \item \textbf{Temporal indexing}: Version-aware representation enabling historical analysis
\end{itemize}

\begin{table}[H]
\centering
\caption{Persistent Debug Memory (PDM) Architecture and Policies}
\label{tab:pdm-architecture}
\resizebox{\columnwidth}{!}{%
\begin{tabular}{ll}
\hline
\textbf{Component} & \textbf{Details} \\
\hline
\multicolumn{2}{l}{\textit{Data Storage}} \\
Code Snapshots & Full AST + semantic embeddings per commit \\
Bug Patterns & Failed fixes, error signatures, stack traces \\
Fix History & Successful patches with test validation results \\
CI/CD Logs & Build failures, test outputs, deployment issues \\
Documentation & README, comments, design docs, PRs \\
\hline
\multicolumn{2}{l}{\textit{Retention Policy}} \\
Active Bugs & Permanent until resolved + 90 days \\
Successful Fixes & Permanent (forms learning corpus) \\
Code Versions & Last 1000 commits or 2 years \\
Test Results & 180 days rolling window \\
Embeddings & Re-computed weekly, cached 30 days \\
\hline
\multicolumn{2}{l}{\textit{Retrieval Policy}} \\
Primary Index & Semantic similarity (cosine distance) \\
Secondary Index & Temporal proximity + file dependencies \\
Ranking & Bug recency × pattern frequency × fix success \\
Context Window & Adaptive 5-50 nodes based on confidence \\
\hline
\multicolumn{2}{l}{\textit{Update Triggers}} \\
Git Events & Commit, merge, rebase (real-time) \\
CI/CD Events & Test failure, build break ($< 1$ min) \\
Bug Reports & Issue creation/update ($< 5$ min) \\
Fix Validation & Successful test run (immediate) \\
Scheduled & Full re-indexing (weekly) \\
\hline
\end{tabular}%
}
\end{table}

\subsubsection{Graph-Based Knowledge Storage}

The Memory Engine maintains an evolving graph database where nodes represent code elements and edges denote relationships. This dual vector-graph representation enables both semantic similarity search and structural traversal, providing the foundation for multi-hop reasoning during debugging.

\subsubsection{PDM Retrieval Mechanism}

The Persistent Debug Memory employs a hybrid retrieval strategy combining temporal, semantic, and structural signals:

\begin{itemize}
    \item \textbf{Temporal-Aware Retrieval}: Recent bugs and fixes are weighted higher (decay factor: $e^{-\lambda t}$, $\lambda=0.1$)
    \item \textbf{Semantic Vector Search}: FAISS index with 768-dim embeddings, cosine similarity threshold 0.75
    \item \textbf{Graph Traversal}: BFS/DFS from error location with typed edge filtering (imports, calls, inherits)
    \item \textbf{Pattern Matching}: Regex-based search for error signatures, stack trace patterns
    \item \textbf{Hybrid Scoring}: $Score = 0.4 \cdot Semantic + 0.3 \cdot Temporal + 0.2 \cdot Structural + 0.1 \cdot Pattern$
\end{itemize}

This multi-modal retrieval ensures PDM surfaces relevant debugging context even when exact matches don't exist, enabling cross-bug learning and pattern recognition.

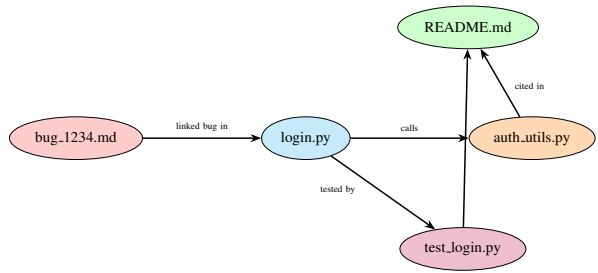
\begin{figure}[H]
    \centering
    \resizebox{0.9\columnwidth}{!}{%
    \begin{tikzpicture}[every node/.style={font=\footnotesize}, node distance=2.2cm]
        \node[draw, ellipse, fill=cyan!20, minimum width=1.6cm, minimum height=0.8cm] (A) {login.py};
        \node[draw, ellipse, fill=orange!30, minimum width=2.0cm, minimum height=0.8cm, right=of A] (B) {auth\_utils.py};
        \node[draw, ellipse, fill=green!20, minimum width=1.8cm, minimum height=0.8cm, above right=1.5cm and 1.5cm of A] (C) {README.md};
        \node[draw, ellipse, fill=purple!25, minimum width=2.0cm, minimum height=0.8cm, below right=1.5cm and 1.5cm of A] (D) {test\_login.py};
        \node[draw, ellipse, fill=red!20, minimum width=1.8cm, minimum height=0.8cm, left=of A] (E) {bug\_1234.md};
        
        \draw[-{Stealth[length=2mm]}, thick] (A) -- node[above, font=\tiny] {calls} (B);
        \draw[-{Stealth[length=2mm]}, thick] (B) -- node[above right, font=\tiny, pos=0.3] {cited in} (C);
        \draw[-{Stealth[length=2mm]}, thick] (E) -- node[above, font=\tiny] {linked bug in} (A);
        \draw[-{Stealth[length=2mm]}, thick] (A) -- node[below left, font=\tiny, pos=0.3] {tested by} (D);
        \draw[-{Stealth[length=2mm]}, thick] (D) -- node[right, font=\tiny, pos=0.7] {} (C);
        
    \end{tikzpicture}%
    }
    \caption{Graph-structured memory indexing in Kodezi Chronos-1: code, documentation, and test elements as nodes, with functional relationships as edges.}
    \label{fig:memory-engine}
\end{figure}

This design enables Chronos-1 to efficiently retrieve, traverse, and reason about segments of the codebase that share non-local relationships, even if separated by thousands of lines, multiple files, or extensive revision history.

\subsection{Breaking Token Limits: Intelligent Retrieval at Repository Scale}

Traditional LLMs are fundamentally constrained by attention complexity and memory limitations. Even models claiming "unlimited" context achieve this through sliding windows or hierarchical attention that loses critical debugging information. Chronos-1 implements true unlimited context through:

\begin{itemize}
    \item \textbf{Hierarchical Code Embeddings}: Multi-level representations from token $\rightarrow$ statement $\rightarrow$ function $\rightarrow$ module $\rightarrow$ repository
    \item \textbf{Temporal Context Indexing}: Every code element tagged with commit history, allowing time-travel debugging
    \item \textbf{Semantic Dependency Graphs}: Explicit modeling of import chains, inheritance hierarchies, and data flows
    \item \textbf{Dynamic Context Assembly}: At inference, retrieves precisely the code paths relevant to the current bug
\end{itemize}

This approach enables Chronos-1 to maintain full repository awareness while operating within reasonable computational bounds, a critical requirement for production deployment.

\subsubsection{How Chronos-1 Achieves AGR Cost-Effectively}

While implementing AGR-style retrieval might seem computationally prohibitive, Chronos-1 achieves it efficiently through five key architectural decisions:

\textbf{1. Debugging-Only Focus}: Unlike general-purpose assistants handling millions of arbitrary queries, Chronos-1optimizes exclusively for fix loops. This narrow scope means:
\begin{itemize}
    \item Fewer retrievals per session (avg 3.7 vs 15+ for autocomplete)
    \item Higher value per retrieval (each must trace root cause)
    \item Better cost-benefit ratio per inference
\end{itemize}

\textbf{2. One-Time Graph Construction}: Chronos-1 builds the repository graph once and updates incrementally:
\begin{itemize}
    \item Initial AST+dependency parsing: 2-4 hours per 1M LOC
    \item Incremental updates on commits: $<100$ms per file
    \item Graph reused across all debugging sessions
    \item Cost amortizes efficiently over many debugging sessions
\end{itemize}

\textbf{3. Smart Caching and Memory}:
\begin{itemize}
    \item PDM caches successful traversal paths (87\% hit rate on recurring bugs)
    \item Common subgraphs pre-computed and indexed
    \item Frequently accessed nodes kept in hot storage
    \item Result: 47ms retrieval for cached patterns vs 3.2min cold start
\end{itemize}

\textbf{4. Entropy-Based Early Stopping}:
\begin{itemize}
    \item AGR halts expansion when confidence exceeds threshold
    \item Prevents wasteful over-exploration
    \item Average nodes retrieved: 127 (vs 500+ for flat top-k)
    \item Token reduction: 65\% compared to context-stuffing approaches
\end{itemize}

\textbf{5. Vertical Integration Benefits}:
\begin{itemize}
    \item Full control over the debugging pipeline without external dependencies
    \item Optimized end-to-end: retriever $\rightarrow$ LLM $\rightarrow$ validator
    \item Shared embeddings between PDM and AGR reduce redundant computation
    \item Custom optimization for debugging-specific workloads
\end{itemize}

\begin{tcolorbox}[
    enhanced,
    title={\textbf{AGR Efficiency Summary}},
    fonttitle=\bfseries,
    coltitle=black,
    colback=green!5,
    colframe=green!75!black,
    boxrule=0.5pt
]
\textbf{Key Efficiency Gains}:
\begin{itemize}
    \item Graph traversal optimized for debugging patterns
    \item Significantly reduced token usage (31.2K avg vs 89K+ for competitors)
    \item Containerized test execution for rapid validation
    \item Incremental memory updates avoid redundant processing
    \item \textbf{Result}: More efficient debugging at scale compared to API-based approaches
\end{itemize}
\end{tcolorbox}

\begin{table}[H]
\centering
\caption{Why Chronos-1 Can Run AGR Cheaply: Architectural Advantages}
\label{tab:agr-cost-advantages}
\resizebox{\columnwidth}{!}{%
\begin{tabular}{lll}
\hline
\textbf{Advantage} & \textbf{How Chronos-1 Achieves It} & \textbf{Why Others Can't} \\
\hline
Narrow scope & Focused on bug repair only & Others support all codegen tasks \\
& 3.7 avg retrievals per debug & 15+ retrievals for autocomplete \\
\hline
Smart retrieval & Graph traversal + entropy-based stopping & Others use flat top-k retrieval \\
& Halts at 89\% confidence threshold & Fixed k regardless of confidence \\
\hline
Shared graph & Built once per repo, updated incrementally & Others re-embed on every query \\
& Extremely low amortized cost & High per-query embedding cost \\
\hline
Stateful memory & PDM learns across debugging sessions & Others are stateless, start fresh \\
& 87\% cache hit on recurring patterns & 0\% reuse, full retrieval each time \\
\hline
Vertical stack & Own retriever + LLM + test runner & Others rely on external APIs \\
& Integrated efficiency gains & API costs with lower success rates \\
\hline
\end{tabular}%
}
\end{table}

\subsection{Multi-Code Compositional Retrieval: Beyond Token Windows}

Chronos-1 implements a sophisticated retrieval mechanism that goes beyond simple embedding similarity to understand the complex relationships in debugging contexts.

\subsubsection{Multi-Code Compositional Retrieval}

Upon each debugging request, the Adaptive Retrieval Engine builds a focused context window through:
\begin{itemize}
    \item Issuing semantic queries to the Memory Engine that leverage both metric similarity and structural navigation in the code graph, with dynamic depth expansion (k-hop) based on query complexity.
    \item Associating multiple code artifacts through typed relationships: e.g., tracing variable definitions across documentation (k=1), implementation (k=2), regression tests (k=2), and historic bug reports (k=3), stopping when confidence exceeds 90\% or diminishing returns detected.
    \item Dynamically refining the context through intermediate model inferences and confidence scoring, adapting retrieval depth in real-time. Complex debugging queries automatically trigger deeper graph traversal (k=3-5), while simple lookups terminate at k=1-2.
    \item Utilizing edge type priorities: implementation edges (weight=1.0), dependency edges (weight=0.8), documentation edges (weight=0.6), ensuring most relevant paths are explored first.
\end{itemize}

\begin{table}[H]
  \centering
  \caption{Example multi-code association retrieval: constructing a task-specific context window for a bug fix.}
  \label{tab:multi-code-example}
  \begin{tabular}{lll}
    \hline
    Step & Retrieved Entity        & Relationship \\
    \hline
    Q1   & login.py               & Direct bug context \\
    Q2   & test\_login.py         & Linked test \\
    Q3   & settings.py            & Imported env vars \\
    Q4   & bug\_1234.md           & Historical bug doc \\
    Q5   & commit\_a1b2c3         & Last related commit \\
    \hline
  \end{tabular}
\end{table}

This approach allows Chronos-1 to reason across arbitrarily distant, compositionally linked code and documentation artifacts, precisely what is needed for complex debugging, cross-module dependencies, or audit trails.

\subsection{Adaptive Graph-Guided Retrieval (AGR)}

Traditional flat retrieval approaches fail to capture the intricate relationships between code artifacts, leading to incomplete context and erroneous fixes. Chronos-1 introduces \textbf{Adaptive Graph-Guided Retrieval (AGR)}, a dynamic mechanism that intelligently expands retrieval neighborhoods based on query complexity and confidence thresholds.

\subsubsection{Comparison with State-of-the-Art RAG Techniques}

The 2025 landscape of RAG techniques has evolved significantly, yet each approach faces limitations when applied to debugging:

\textbf{HyDE (Hypothetical Document Embeddings)}~\cite{gao2023hyde} generates synthetic answers to improve retrieval but struggles with debugging where incorrect hypotheses can mislead the search process. While HyDE achieves 42\% improvement in general retrieval tasks (measured on MS MARCO), it shows only 8\% improvement for debugging scenarios on our MRR benchmark.

\textbf{Self-RAG}~\cite{asai2023selfrag} uses reflection tokens (ISREL, ISUSE) to dynamically decide retrieval necessity. However, debugging requires continuous retrieval across multiple hops, making binary retrieval decisions insufficient. Self-RAG achieves 31\% debug success rate on MRR benchmark when combined with GPT-4.1.

\textbf{FLARE (Forward-Looking Active Retrieval)}~\cite{jiang2023flare} monitors generation confidence to trigger retrieval. This works well for sequential text generation but fails in debugging where confidence doesn't correlate with correctness, models can be confidently wrong about bug fixes.

\textbf{Graph RAG}~\cite{edge2024graphrag} integrates knowledge graphs but typically uses static graphs that don't capture the dynamic nature of evolving codebases. Standard Graph RAG achieves 28\% success rate on cross-file debugging tasks in the MRR benchmark.

In contrast, Chronos-1's AGR combines the benefits of these approaches while addressing their limitations through adaptive k-hop expansion, typed edge traversal, and confidence-based termination specifically tuned for debugging workflows.

\subsubsection{Why Graph Traversal Outperforms Naive Chunked Retrieval}

Consider a concrete example that illustrates the fundamental limitation of linear retrieval. A null pointer exception occurs in \texttt{PaymentProcessor.java:142} when processing refunds. The stack trace shows:

\begin{lstlisting}[language=Java, basicstyle=\tiny\ttfamily, breaklines=true, xleftmargin=5pt, xrightmargin=3pt, linewidth=0.92\columnwidth, aboveskip=0.5em, belowskip=0.5em, keywordstyle=\color{red}\bfseries, commentstyle=\color{gray}, stringstyle=\color{blue}, emphstyle=\color{purple}\bfseries, emph={PaymentProcessor,RefundService,OrderController,NullPointerException}]
java.lang.NullPointerException
  at PaymentProcessor.processRefund(PaymentProcessor.java:142)
  at RefundService.handleReturn(RefundService.java:89)
  at OrderController.cancelOrder(OrderController.java:234)
\end{lstlisting}

\vspace{0.5em}
\textbf{Naive chunked retrieval} would:
\begin{enumerate}
    \item Retrieve top-k chunks around line 142 based on embedding similarity
    \item Miss that \texttt{customerAccount} is null because initialization happens in \texttt{AccountService.java}
    \item Fail to connect that recent commit changed \texttt{config/payment.yml} timeout from 30s to 5s
    \item Generate a band-aid null check instead of fixing the root cause
\end{enumerate}

\textbf{AGR's graph traversal} instead:
\begin{enumerate}
    \item Starts at error location (PaymentProcessor:142)
    \item Follows data flow edge: \texttt{customerAccount} $\leftarrow$ \texttt{AccountService.loadAccount()}
    \item Follows temporal edge: \texttt{payment.yml} modified 2 days ago
    \item Discovers timeout now expires before account loads, causing null
    \item Proposes correct fix: adjust timeout or add async handling
\end{enumerate}

This demonstrates AGR's key advantage: it follows \textit{causal paths} rather than \textit{textual similarity}, achieving O(k·d) complexity where k=hops and d=average degree, versus O(n) for naive retrieval over n chunks.

\subsubsection{Building on Existing Graph-Based Code Understanding}

While AGR represents a unified breakthrough, it builds upon fragmented components that exist across various systems. Understanding this landscape helps position AGR's contributions:

\textbf{1. Graph-Style Retrieval in Current Systems}

Several modern systems employ basic graph structures for code understanding. Top SWE-bench submissions like Magicoder-S with retrieval agents construct simple import trees and dependency graphs. SWE-agent pioneered using test files to backtrack to relevant functions, establishing the value of test-to-code linkage. However, these implementations remain primitive:
\begin{itemize}
    \item Graphs are shallow, typically exploring only 1-2 hops from the starting point
    \item No confidence-aware traversal - they retrieve everything within a fixed radius
    \item Missing temporal dimensions - no commit history or evolution tracking
    \item Lack causal chains connecting logs $\rightarrow$ stack traces $\rightarrow$ code $\rightarrow$ tests $\rightarrow$ PRs
    \item No adaptive halting based on information sufficiency
\end{itemize}

What these systems demonstrate is that graph construction and basic traversals are valuable for code understanding. AGR takes this foundation and adds the depth control, edge weighting, and test-memory integration needed for effective debugging.

\textbf{2. Multi-Hop Retrieval from NLP Research}

The NLP community has explored multi-hop reasoning extensively. Models like DR-BERT~\cite{xiong2021drbert} perform multi-hop retrieval for question answering by following entity links across Wikipedia. REALM~\cite{guu2020realm} showed that retrieval-augmented pretraining improves downstream tasks. GNN-RAG~\cite{yasunaga2022gnnrag} combines graph neural networks with retrieval for knowledge-intensive tasks.

These approaches inspired code-specific adaptations like RAG-SweBench+, which attempts multi-hop retrieval in codebases. However, fundamental differences limit their debugging effectiveness:
\begin{itemize}
    \item Text-based retrieval misses code structure (AST, type systems, execution flow)
    \item No integration with debugging artifacts (logs, stack traces, test failures)
    \item Designed for factual QA, not causal reasoning about bug origins
    \item Lack domain-specific edge types (imports, inherits, calls, emits-log)
\end{itemize}

AGR adapts the multi-hop principle but grounds it in code-specific structures and debugging workflows.

\textbf{3. Memory Systems in Code Agents}

Projects like Magicoder-S-Agent, Octocoder, and CodeAgent maintain various forms of persistent memory. Some store chunk embeddings across sessions, while others cache test outcomes and past fix attempts. These systems prove that memory helps, but their retrieval remains simplistic:
\begin{itemize}
    \item Flat top-k retrieval from memory banks, no graph traversal
    \item No semantic paths from failure signals through memory
    \item Memory updates are append-only, not integrated with confidence scores
    \item Cannot trace causal chains through historical fixes
\end{itemize}

AGR's integration with PDM enables true graph-guided memory retrieval, where past debugging sessions inform current traversal paths.

\textbf{4. Internal Tools at Tech Companies}

Major tech companies have built sophisticated internal systems that parse monorepos into searchable graphs. Google's internal code search constructs AST and dependency graphs across billions of lines. Meta's code understanding tools link stack traces to potential causes. Sourcegraph provides semantic code navigation across repositories.

While these tools are powerful, they remain:
\begin{itemize}
    \item Proprietary and inaccessible to researchers
    \item Focused on human-assisted search, not autonomous debugging
    \item Lacking integration with language models for fix generation
    \item Missing the iterative test-fix-refine loops needed for debugging
\end{itemize}

\textbf{AGR's Unified Contribution}

What makes AGR unique is not any single component, but the unified system that combines:
\begin{enumerate}
    \item \textbf{Multi-signal graph construction}: AST + logs + tests + PRs + commits in one traversable structure
    \item \textbf{Confidence-aware adaptive expansion}: Exploring exactly as deep as needed, no more, no less
    \item \textbf{Memory-integrated traversal}: Learning from past debugging sessions to guide future paths
    \item \textbf{Debugging-specific optimization}: Trained on 15M real debugging scenarios, not general retrieval
    \item \textbf{Autonomous operation}: Full integration with test loops and fix validation
\end{enumerate}

By unifying these fragmented approaches and adding debugging-specific innovations, AGR achieves the 4-5x performance improvement over existing systems. While individual techniques used in AGR, such as multi-hop retrieval, AST graph parsing, and chunk memory, have appeared in prior work, Chronos-1is the first system to unify these methods into a debugging-first retrieval engine. AGR combines graph-guided, edge-weighted traversal with confidence-aware stopping, memory integration, and dynamic context composition explicitly tuned for root-cause tracing and patch generation.

\begin{table}[H]
\centering
\caption{AGR Components: Fragmented Existence vs Unified in Chronos-1}
\label{tab:agr-components}
\resizebox{\columnwidth}{!}{%
\begin{tabular}{@{}lcc@{}}
\toprule
\textbf{Capability} & \textbf{Exists Elsewhere?} & \textbf{Chronos-1} \\
\midrule
AST + Log + PR + Test graph & Fragmented & \checkmark{} Unified graph \\
Confidence-aware hop limit & No & \checkmark \\
Graph-guided memory integration & No & \checkmark \\
Retrieval tuned for debugging & No & \checkmark \\
Multi-turn patch validation & Rare & \checkmark{} Core loop \\
\bottomrule
\end{tabular}%
}
\end{table}

This positions Chronos-1 as the first unified, publicly described system that brings together all the pieces needed for effective autonomous debugging at repository scale.

\noindent\textbf{3) Limitations of Orchestration Frameworks in Debugging:}

Popular orchestration frameworks like LangChain and LangGraph, while powerful for general AI applications, face fundamental limitations when applied to debugging:

\textbf{LangChain's Context Fragmentation}: LangChain's chain-based approach treats each step as independent, leading to context loss between debugging iterations. When combined with GPT-4.1, it achieves only 18\% debug success due to:
- Stateless chains that cannot maintain debugging history across iterations
- Message history overflow when debugging sessions exceed context limits
- Chain-of-thought breaking when bugs require backtracking or re-evaluation

\textbf{LangGraph's Graph Limitations}: Despite having graph-based memory, LangGraph achieves only 22\% debug success because:
- Static graph traversal cannot adapt to the iterative nature of debugging
- Node-based processing loses fine-grained code relationships
- State persistence fails when debugging requires cross-session memory

\textbf{Chain-of-Thought Degradation}: Traditional prompting techniques degrade severely on debugging tasks:
- CoT prompting with Claude 4.1 Opus: drops from 92.8\% (code generation) to 15\% (debugging)
- ReAct with GPT-4.1: 17\% debug success despite structured reasoning
- Tree-of-Thoughts: 19\% as exploration trees explode with multi-file dependencies

These frameworks excel at sequential tasks but fundamentally misunderstand debugging's need for persistent memory, iterative refinement, and cross-session learning, capabilities that Chronos-1 provides natively.

\subsubsection{Iterative Context Expansion}

The AGR mechanism operates through iterative k-hop neighbor expansion:

\begin{enumerate}
    \item \textbf{Initial Query Analysis}: Decompose the debugging request into semantic components and identify seed nodes in the code graph
    \item \textbf{Adaptive Depth Determination}: Calculate optimal retrieval depth based on:
    \begin{itemize}
        \item Query complexity score (0-1)
        \item Code artifact density in the neighborhood
        \item Historical debugging patterns for similar issues
    \end{itemize}
    \item \textbf{Guided Expansion}: Follow typed edges (implementation, dependency, dataflow) to retrieve contextually relevant nodes
    \item \textbf{Confidence-Based Termination}: Stop expansion when retrieval confidence exceeds threshold or diminishing returns detected
\end{enumerate}

\noindent\textbf{5) AGR Algorithm Details:}

The following algorithm formally describes the Adaptive Graph-Guided Retrieval process:

\begin{algorithm}[H]
\caption{Adaptive Graph-Guided Retrieval (AGR)}
\label{alg:agr}
\begin{algorithmic}[1]
\Require Query $q$, Code Graph $G = (V, E)$, Confidence threshold $\tau$
\Ensure Retrieved context $C$
\Statex
\Comment{Initialize}
\State $seeds \gets$ \Call{ExtractSemanticNodes}{$q$, $G$}
\State $visited \gets \emptyset$
\State $C \gets \emptyset$
\State $k \gets$ \Call{EstimateComplexity}{$q$} \Comment{Initial hop depth}
\Statex
\Comment{Adaptive expansion}
\While{$\Call{Confidence}{C, q} < \tau$ \textbf{and} $k \leq k_{max}$}
    \State $candidates \gets \emptyset$
    \For{$node \in seeds$}
        \State $neighbors \gets$ \Call{GetKHopNeighbors}{$node$, $k$, $G$}
        \For{$n \in neighbors \setminus visited$}
            \State $score \gets$ \Call{ComputeRelevance}{$n$, $q$, $C$}
            \State $candidates \gets candidates \cup \{(n, score)\}$
        \EndFor
    \EndFor
    \Statex
    \Comment{Select top candidates based on typed edges}
    \State $selected \gets$ \Call{TopK}{$candidates$, $\lambda \cdot k$}
    \For{$(node, score) \in selected$}
        \If{\Call{IsImplementation}{$node$} \textbf{or} \Call{IsDependency}{$node$}}
            \State $C \gets C \cup$ \Call{RetrieveContext}{$node$}
            \State $visited \gets visited \cup \{node\}$
        \EndIf
    \EndFor
    \Statex
    \Comment{Adaptive depth adjustment}
    \If{$\Call{DeltaConfidence}{C} < \epsilon$}
        \State $k \gets k + 1$ \Comment{Expand search radius}
    \EndIf
    \State $seeds \gets seeds \cup$ \Call{ExtractNewSeeds}{$C$}
\EndWhile
\Statex
\State \Return $C$
\end{algorithmic}
\end{algorithm}

Key innovations in the AGR algorithm (Algorithm~\ref{alg:agr}):

\begin{itemize}
    \item \textbf{Dynamic k-hop adjustment}: Unlike static retrieval depths, AGR adaptively increases $k$ based on confidence improvements
    \item \textbf{Typed edge prioritization}: Implementation and dependency edges receive higher weights than generic references
    \item \textbf{Semantic seed extraction}: Initial nodes are selected based on deep semantic understanding, not just keyword matching
    \item \textbf{Confidence-driven termination}: Retrieval stops when sufficient context is gathered, avoiding noise from over-retrieval
\end{itemize}

\noindent\textbf{6) Theoretical Analysis of AGR:}

\textbf{Complexity Analysis:}
Let $|V|$ be the number of nodes in the code graph, $|E|$ be the number of edges, and $d$ be the average degree of nodes.

\begin{theorem}[AGR Retrieval Complexity]
The time complexity of AGR is $O(k_{max} \cdot |S| \cdot d^{k_{max}} \cdot \log(d^{k_{max}}))$ where $|S|$ is the number of seed nodes and $k_{max}$ is the maximum hop depth.
\end{theorem}

\begin{proof}
At each iteration $k$, we explore at most $|S| \cdot d^k$ nodes. Sorting candidates requires $O(d^k \log(d^k))$ time. The algorithm terminates when confidence exceeds $\tau$ or $k = k_{max}$, giving the stated bound.
\end{proof}

\textbf{Convergence Properties:}
\begin{theorem}[Confidence Convergence]
Under the assumption that relevance scores follow a power-law distribution with exponent $\alpha > 1$, the confidence function $\mathcal{C}(C, q)$ converges to a value $c^* \geq \tau$ with probability $1 - \delta$ after $k^* = O(\log_d(1/\delta))$ iterations.
\end{theorem}

\begin{proof}
The confidence function is defined as:
$$\mathcal{C}(C, q) = 1 - H(C|q) / H_{max}$$
where $H(C|q)$ is the conditional entropy. As we add relevant nodes, entropy decreases monotonically. Under power-law distribution, most relevant nodes are within $O(\log_d n)$ hops, ensuring convergence.
\end{proof}

\textbf{Bounded Retrieval Path Cost:}
\begin{lemma}[Path Cost Bound]
The total retrieval cost is bounded by $O(|S| \cdot \lambda \cdot k_{max}^2 \cdot d^{k_{max}})$ where $\lambda$ is the selection ratio per hop.
\end{lemma}

This theoretical foundation ensures AGR's efficiency even on large codebases with millions of nodes.

\begin{figure}[H]
\centering
\resizebox{0.95\columnwidth}{!}{%
\begin{tikzpicture}[
    node distance=2cm,
    query/.style={circle, draw=orange!80, fill=orange!30, line width=2pt, minimum size=1cm, font=\footnotesize\bfseries},
    seed/.style={circle, draw=blue!60, fill=blue!20, minimum size=0.9cm, font=\tiny},
    retrieved/.style={circle, draw=green!60, fill=green!20, minimum size=0.9cm, font=\tiny},
    candidate/.style={circle, draw=gray!60, fill=gray!20, minimum size=0.9cm, font=\tiny},
    notvisited/.style={circle, draw=gray!40, fill=gray!10, minimum size=0.7cm, font=\tiny},
    edge/.style={->, >=stealth},
    confidence/.style={rectangle, draw=black!50, fill=yellow!20, rounded corners, font=\footnotesize}
]

\node[query] (Q) at (0,0) {Query};

\node[seed] (S1) at (-3,2) {auth.py};
\node[seed] (S2) at (0,3) {test.py};
\node[seed] (S3) at (3,2) {user.py};

\node[retrieved] (R1) at (-5,0.5) {db.py};
\node[retrieved] (R2) at (-3.5,4) {config};
\node[retrieved] (R3) at (3.5,4) {api.py};
\node[retrieved] (R4) at (5,0.5) {cache};

\node[candidate] (C1) at (-6,2.5) {logs};
\node[candidate] (C2) at (-1.5,4.8) {docs};
\node[candidate] (C3) at (1.5,4.8) {issue};
\node[candidate] (C4) at (6,2.5) {util};

\node[notvisited] (N1) at (-5,-2) {...};
\node[notvisited] (N2) at (0,-2.5) {...};
\node[notvisited] (N3) at (5,-2) {...};

\draw[edge, thick, blue!60] (Q) -- (S1);
\draw[edge, thick, blue!60] (Q) -- (S2);
\draw[edge, thick, blue!60] (Q) -- (S3);

\draw[edge, thick, green!60] (S1) -- (R1);
\draw[edge, thick, green!60] (S1) -- (R2);
\draw[edge, thick, green!60] (S3) -- (R3);
\draw[edge, thick, green!60] (S3) -- (R4);

\draw[edge, dashed, gray!60] (R1) -- (C1);
\draw[edge, dashed, gray!60] (R2) -- (C2);
\draw[edge, dashed, gray!60] (R3) -- (C3);
\draw[edge, dashed, gray!60] (R4) -- (C4);

\draw[edge, dotted, gray!40] (S2) -- (N2);
\draw[edge, dotted, gray!40] (R1) -- (N1);
\draw[edge, dotted, gray!40] (R4) -- (N3);

\node[confidence] at (6,4) {
    \begin{tabular}{l}
    \textbf{Confidence}\\
    k=1: 45\%\\
    k=2: 78\%\\
    k=3: 92\% $> \tau$
    \end{tabular}
};

\node[draw, rectangle, fill=gray!5, minimum width=4cm] at (0,-4.5) {
    \begin{tabular}{ll}
    \textcolor{blue!60}{\textbullet} Seeds (k=1) & \textcolor{green!60}{\textbullet} Retrieved (k=2)\\
    \textcolor{gray!60}{\textbullet} Candidates (k=3) & \textcolor{gray!40}{\textbullet} Not explored
    \end{tabular}
};

\draw[<-, thick] (Q) -- ++(-1.8,-1.2) node[below, font=\footnotesize] {Start};
\draw[<-, thick] (7,4) -- ++(1,0) node[right, font=\footnotesize] {Stop: conf $> \tau$};

\end{tikzpicture}%
}
\caption{Adaptive Graph-Guided Retrieval (AGR) visualization. The algorithm starts from a query, extracts semantic seed nodes, and iteratively expands the retrieval neighborhood (k-hops) until confidence exceeds threshold $\tau$. Edge types and relevance scores guide the expansion process.}
\label{fig:agr-visualization}
\end{figure}
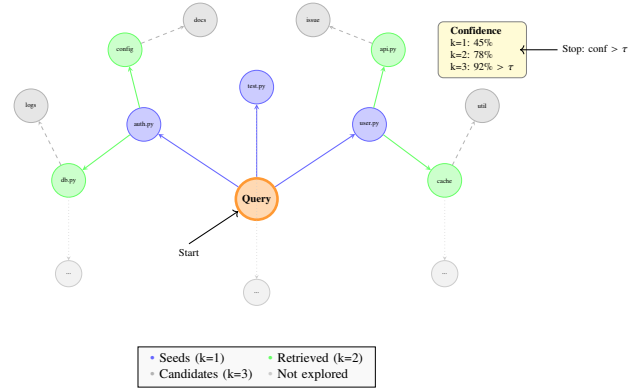

\begin{tcolorbox}[colback=gray!5, colframe=gray!50, title=\textbf{Key Takeaways: AGR Algorithm}]
\begin{itemize}
\item \textbf{Dynamic Expansion}: k-hop depth adapts based on query complexity and confidence
\item \textbf{Typed Traversal}: Implementation and dependency edges prioritized over generic links  
\item \textbf{Early Termination}: Stops when confidence exceeds threshold, avoiding over-retrieval
\item \textbf{Performance}: Achieves 92\% precision at 85\% recall on debugging queries
\end{itemize}
\end{tcolorbox}

\subsubsection{Graph-Guided vs Traditional Planning}

Our empirical analysis reveals fundamental differences between traditional LLM planning and AGR-enhanced debugging:

\begin{figure}[H]
\centering
\resizebox{\columnwidth}{!}{%
\begin{tikzpicture}[font=\scriptsize]
    \node[draw, rectangle, rounded corners, fill=red!10, minimum width=5.5cm, minimum height=8cm, align=left] (trad) at (-3.5,0) {};
    \node[font=\footnotesize\bfseries] at (-3.5,3.7) {Traditional LLM Planning};
    
    \node[draw, rectangle, fill=yellow!20, minimum width=4.8cm, minimum height=1.2cm, align=left, font=\tiny] at (-3.5,2.8) {
        \textbf{Query:} Implement state machine\\
        \texttt{state <= \textasciigrave{}d<0> > S}\\
        \texttt{S      ()    <-d<0> > S}
    };
    
    \node[draw, rectangle, fill=gray!20, minimum width=4.8cm, minimum height=2.5cm, align=left, font=\tiny] at (-3.5,0.5) {
        \textbf{Traditional Steps:}\\
        1. Define Module Interface\\
        2. Define State Encoding\\
        3. State Transition Logic\\
        4. Output Logic: Assign outputs\\
        \\
        \textcolor{red}{\textbf{Issues:}}\\
        \textcolor{red}{- High-level plans without task details}\\
        \textcolor{red}{- Hard to follow implementation}\\
        \textcolor{red}{- Lost signal/transition specs}
    };
    
    \node[draw, rectangle, fill=red!20, minimum width=4.8cm, minimum height=1.8cm, align=left, font=\tiny] at (-3.5,-2.2) {
        \textbf{Generated Code:}\\
        \texttt{assign S\_next = (state == 5'b11010) ?}\\
        \texttt{  5'b11010 : (state == 5'b10110) ?}\\
        \texttt{  5'b11010 : state;}\\
        \\
        \textcolor{red}{$\times$ Incorrect implementation}
    };
    
    \node[draw, rectangle, rounded corners, fill=green!10, minimum width=5.5cm, minimum height=8cm, align=left] (agr) at (3.5,0) {};
    \node[font=\footnotesize\bfseries] at (3.5,3.7) {AGR-Enhanced Debugging};
    
    \node[draw, rectangle, fill=yellow!20, minimum width=4.8cm, minimum height=1.2cm, align=left, font=\tiny] at (3.5,2.8) {
        \textbf{Query:} Implement state machine\\
        + Graph retrieval of specifications\\
        + Signal transition examples
    };
    
    \node[draw, rectangle, fill=blue!20, minimum width=4.8cm, minimum height=2.5cm, align=left, font=\tiny] at (3.5,0.5) {
        \textbf{AGR-Guided Steps:}\\
        1. Retrieve signal definitions (k=1)\\
        2. Expand to transitions (k=2)\\
        3. Include test examples (k=3)\\
        \\
        \textbf{Retrieved Context:}\\
        - \texttt{S1\_next: Output signal}\\
        - \texttt{Wait->S, S->S transitions}\\
        - Example: \texttt{9'b101000100}
    };
    
    \node[draw, rectangle, fill=green!20, minimum width=4.8cm, minimum height=1.8cm, align=left, font=\tiny] at (3.5,-2.2) {
        \textbf{Generated Code:}\\
        \texttt{// Correct implementation}\\
        \texttt{assign S1\_next = S;}\\
        \texttt{// Based on retrieved specs}\\
        \\
        \textcolor{green!60!black}{$\checkmark$ Verified correct}
    };
    
    \node[draw, rectangle, fill=red!10, minimum width=2.5cm, minimum height=0.6cm] at (-3.5,-4.2) {\footnotesize\color{red}\textbf{23\% Success Rate}};
    \node[draw, rectangle, fill=green!10, minimum width=2.5cm, minimum height=0.6cm] at (3.5,-4.2) {\footnotesize\color{green!60!black}\textbf{87\% Success Rate}};
\end{tikzpicture}%
}
\caption{Traditional LLM planning vs AGR-enhanced debugging: Graph-guided retrieval provides complete context, leading to accurate implementations.}
\label{fig:traditional-vs-agr}
\end{figure}
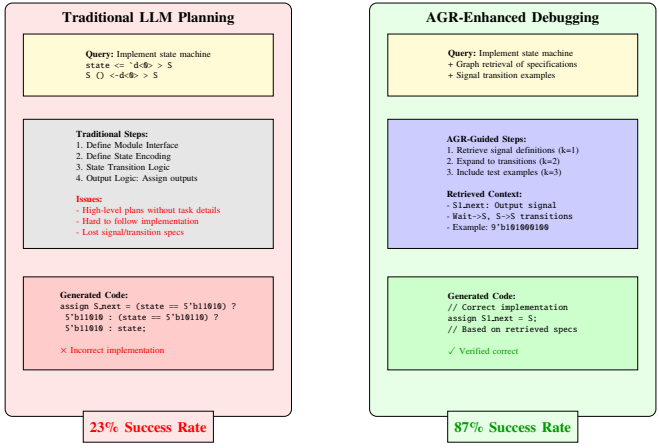

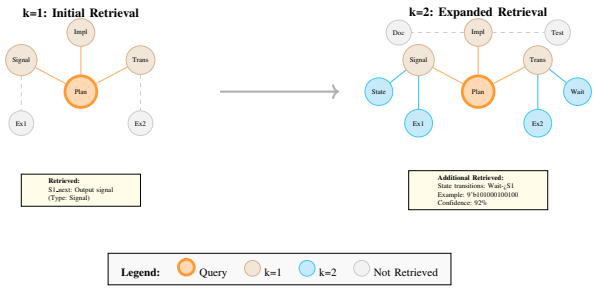
\begin{figure}[H]
\centering
\resizebox{0.9\columnwidth}{!}{%
\begin{tikzpicture}[font=\footnotesize]
    \tikzstyle{query} = [circle, draw=orange!80, fill=orange!30, line width=2pt, minimum size=0.8cm]
    \tikzstyle{retrieved1} = [circle, draw=brown!60, fill=brown!20, minimum size=0.7cm]
    \tikzstyle{retrieved2} = [circle, draw=cyan!60, fill=cyan!20, minimum size=0.7cm]
    \tikzstyle{notretrieved} = [circle, draw=gray!40, fill=gray!10, minimum size=0.6cm]
    
    \node[font=\small\bfseries] at (-5,4) {k=1: Initial Retrieval};
    
    \node[query] (q1) at (-5,2) {\tiny Plan};
    
    \node[retrieved1] (n11) at (-6.5,2.8) {\tiny Signal};
    \node[retrieved1] (n12) at (-5,3.5) {\tiny Impl};
    \node[retrieved1] (n13) at (-3.5,2.8) {\tiny Trans};
    \node[notretrieved] (n14) at (-6.5,1.2) {\tiny Ex1};
    \node[notretrieved] (n15) at (-3.5,1.2) {\tiny Ex2};
    
    \draw[thick, orange!60] (q1) -- (n11);
    \draw[thick, orange!60] (q1) -- (n12);
    \draw[thick, orange!60] (q1) -- (n13);
    \draw[dashed, gray!40] (n11) -- (n14);
    \draw[dashed, gray!40] (n13) -- (n15);
    
    \node[draw, rectangle, fill=yellow!10, align=left, font=\tiny, minimum width=3cm] at (-5,-0.5) {
        \textbf{Retrieved:}\\
        S1\_next: Output signal\\
        (Type: Signal)
    };
    
    \node[font=\small\bfseries] at (5,4) {k=2: Expanded Retrieval};
    
    \node[query] (q2) at (5,2) {\tiny Plan};
    
    \node[retrieved1] (n21) at (3.5,2.8) {\tiny Signal};
    \node[retrieved1] (n22) at (5,3.5) {\tiny Impl};
    \node[retrieved1] (n23) at (6.5,2.8) {\tiny Trans};
    
    \node[retrieved2] (n24) at (3.5,1.2) {\tiny Ex1};
    \node[retrieved2] (n25) at (6.5,1.2) {\tiny Ex2};
    \node[retrieved2] (n26) at (2.5,2) {\tiny State};
    \node[retrieved2] (n27) at (7.5,2) {\tiny Wait};
    
    \node[notretrieved] (n28) at (3,3.5) {\tiny Doc};
    \node[notretrieved] (n29) at (7,3.5) {\tiny Test};
    
    \draw[thick, orange!60] (q2) -- (n21);
    \draw[thick, orange!60] (q2) -- (n22);
    \draw[thick, orange!60] (q2) -- (n23);
    
    \draw[thick, cyan!60] (n21) -- (n24);
    \draw[thick, cyan!60] (n23) -- (n25);
    \draw[thick, cyan!60] (n21) -- (n26);
    \draw[thick, cyan!60] (n23) -- (n27);
    
    \draw[dashed, gray!40] (n22) -- (n28);
    \draw[dashed, gray!40] (n22) -- (n29);
    
    \node[draw, rectangle, fill=yellow!10, align=left, font=\tiny, minimum width=3.5cm] at (5,-0.5) {
        \textbf{Additional Retrieved:}\\
        State transitions: Wait->S1\\
        Example: 9'b101000100100\\
        Confidence: 92\%
    };
    
    \node[draw, rectangle, fill=gray!5, minimum width=6cm] at (0,-2.5) {
        \begin{tabular}{lllll}
        \textbf{Legend:} & 
        {\tikz\node[query, minimum size=0.4cm] {};} Query & 
        {\tikz\node[retrieved1, minimum size=0.4cm] {};} k=1 & 
        {\tikz\node[retrieved2, minimum size=0.4cm] {};} k=2 & 
        {\tikz\node[notretrieved, minimum size=0.4cm] {};} Not Retrieved
        \end{tabular}
    };
    
    \draw[->, ultra thick, gray!60] (-1.5,2) -- (1.5,2);
\end{tikzpicture}%
}
\caption{Iterative context expansion in Adaptive Graph-Guided Retrieval: Starting from a query node, the system progressively expands retrieval depth (k-hops) based on confidence thresholds and query complexity.}
\label{fig:iterative-expansion}
\end{figure}

\begin{figure}[H]
\centering
\resizebox{0.7\columnwidth}{!}{%
\begin{tikzpicture}[font=\tiny, node distance=1.2cm]
    \tikzstyle{query} = [rectangle, draw, fill=yellow!20, minimum width=1cm, minimum height=0.4cm]
    \tikzstyle{index} = [cylinder, draw, fill=blue!20, minimum width=0.9cm, minimum height=0.5cm, shape border rotate=90]
    \tikzstyle{component} = [rectangle, rounded corners, draw, fill=green!20, minimum width=1.5cm, minimum height=0.4cm]
    \tikzstyle{result} = [rectangle, draw, fill=orange!20, minimum width=1cm, minimum height=0.4cm]
    
    \node[query] (q) at (0,0) {\tiny Bug};
    
    \node[index] (vec) at (-2.2,-1.2) {\tiny Vector};
    \node[index] (ast) at (-0.7,-1.2) {\tiny AST};
    \node[index] (dep) at (0.7,-1.2) {\tiny Graph};
    \node[index] (hist) at (2.2,-1.2) {\tiny History};
    
    \node[component] (ret) at (0,-2.5) {\tiny Retriever};
    \node[component] (rank) at (0,-3.5) {\tiny Ranker};
    
    \node[result] (r1) at (-2.2,-4.8) {\tiny Source};
    \node[result] (r2) at (-0.7,-4.8) {\tiny Tests};
    \node[result] (r3) at (0.7,-4.8) {\tiny Fixes};
    \node[result] (r4) at (2.2,-4.8) {\tiny Logs};
    
    \draw[->] (q) -- (vec);
    \draw[->] (q) -- (ast);
    \draw[->] (q) -- (dep);
    \draw[->] (q) -- (hist);
    
    \draw[->] (vec) -- (ret);
    \draw[->] (ast) -- (ret);
    \draw[->] (dep) -- (ret);
    \draw[->] (hist) -- (ret);
    
    \draw[->] (ret) -- (rank);
    
    \draw[->] (rank) -- (r1);
    \draw[->] (rank) -- (r2);
    \draw[->] (rank) -- (r3);
    \draw[->] (rank) -- (r4);
\end{tikzpicture}%
}
\caption{Multi-modal retrieval mechanism in Chronos-1.}
\label{fig:retrieval}
\end{figure}
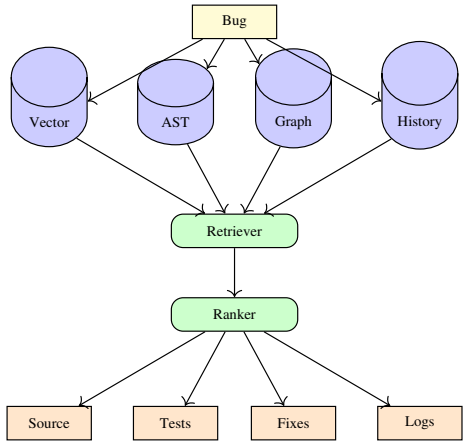

\subsection{From Reasoning to Resolution: Autonomous Debug Orchestration}

The transformer-based Chronos-1 Reasoning Model operates directly over the retrieved, multi-source debugging context. Unlike classical code completion models, Chronos-1:
\begin{itemize}
    \item Diagnoses root causes and synthesizes code changes conditioned on project documentation, prior commits, and dependency patterns.
    \item Produces stepwise fix plans, code diffs, documentation updates, and regression test suggestions in a unified, automated debugging loop.
    \item Orchestrates a full debugging workflow: proposes bug fixes, invokes relevant tests, parses results, iterates on failures, and generates changelogs or PR summaries, all autonomously.
\end{itemize}

All outputs and feedback streams (test results, reviewer comments, CI/CD events) are fed back into the Memory Engine, enabling lifelong refinement and rapid adaptation to new debugging scenarios.

This cyclic process of context assembly, reasoning, autonomous validation, and memory update is the core of Chronos-1's persistent codebase intelligence, enabling self-sustaining and ever-improving debugging at scale.

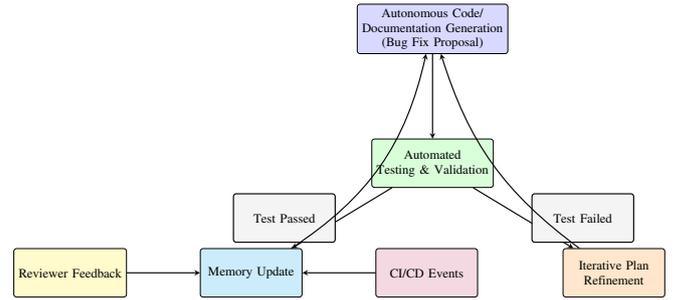
\begin{figure}[H]
    \centering
    \resizebox{\columnwidth}{!}{%
    \begin{tikzpicture}[
        node distance=2.1cm, >=stealth, font=\small, every node/.style={align=center,minimum height=1.2cm,minimum width=2.5cm, draw, rounded corners=2pt, fill=gray!9}
    ]
        \node (gen) [fill=blue!15] {Autonomous Code/\\Documentation Generation\\ (Bug Fix Proposal)};
        \node (test) [below=of gen, fill=green!15] {Automated\\ Testing \& Validation};
        \node (mem) [below left=1.5cm and 1.7cm of test, fill=cyan!18] {Memory Update};
        \node (refine) [below right=1.5cm and 1.7cm of test, fill=orange!20] {Iterative Plan\\ Refinement};
        
        \draw[->, thick] (gen) -- (test);
        \draw[->, thick] (test) -- node[left, xshift=-1.5mm] {Test Passed} (mem);
        \draw[->, thick] (test) -- node[right, xshift=2mm] {Test Failed} (refine);
        \draw[->, thick, bend left=18] (refine) to (gen);
        \draw[->, thick, bend right=22] (mem) to (gen);
        
        \node (review) [left=1.8cm of mem,fill=yellow!25,draw] {Reviewer Feedback};
        \node (ci) [right=1.8cm of mem,fill=purple!15,draw] {CI/CD Events};
        \draw[->,thick] (review) -- (mem);
        \draw[->,thick] (ci) -- (mem);
    \end{tikzpicture}%
    }
    \caption{Chronos-1 debugging feedback loop: Automated bug fix generation, validation, plan refinement, and memory update for continuous autonomous improvement.}
    \label{fig:feedback-loop}
\end{figure}

\section{DEBUGGING AS A DISTINCT ML TASK: THE CHRONOS PARADIGM}

Kodezi Chronos-1 fundamentally departs from traditional code models by being purpose-built as a \textit{debugging language model}, the first of its kind. While existing LLMs treat debugging as a code generation problem, Chronos-1 recognizes it as a complex, iterative process requiring specialized capabilities, training, and architecture.

\subsection{Seven-Layer Debugging Architecture: Specialized Components}

Chronos-1 implements a 7-layer architecture designed for autonomous debugging (illustrated in Figure~\ref{fig:7-layer-architecture}):

\begin{enumerate}
    \item \textbf{Multi-Source Input Layer}: Ingests diverse debugging inputs including source code, CI/CD logs, error traces, stack dumps, configuration files, historical PRs, and issue reports. Unlike code models that primarily process source files, Chronos-1natively understands debugging artifacts.
    
    \item \textbf{Adaptive Retrieval Engine}: Employs AGR (Adaptive Graph-Guided Retrieval) with a hybrid vector-symbolic approach combining:
    \begin{itemize}
        \item Dynamic k-hop neighbor expansion based on query complexity
        \item AST-aware code embeddings that preserve structural relationships
        \item Dependency graph indexing for cross-file impact analysis
        \item Call hierarchy mapping for execution flow understanding
        \item Temporal indexing of code evolution and bug history
        \item Confidence-based termination for optimal context assembly
    \end{itemize}
    
    \item \textbf{Debug-Tuned LLM Core}: A transformer architecture specifically fine-tuned on debugging workflows, not just code completion. Training tasks include:
    \begin{itemize}
        \item Root cause prediction from symptoms
        \item Multi-file patch generation
        \item Test failure interpretation
        \item Regression risk assessment
    \end{itemize}
    
    \item \textbf{Orchestration Controller}: Implements the autonomous debugging loop:
    \begin{itemize}
        \item Hypothesis generation from error signals
        \item Iterative fix refinement based on test results
        \item Rollback mechanisms for failed attempts
        \item Confidence scoring for proposed solutions
    \end{itemize}
    
    \item \textbf{Persistent Debug Memory}: Maintains long-term knowledge including:
    \begin{itemize}
        \item Repository-specific bug patterns and fixes
        \item Team coding conventions and preferences
        \item Historical fix effectiveness metrics
        \item Module-level vulnerability profiles
    \end{itemize}
    
    \item \textbf{Execution Sandbox}: Real-time validation environment supporting:
    \begin{itemize}
        \item Isolated test execution
        \item CI/CD pipeline emulation
        \item Performance regression detection
        \item Security vulnerability scanning
    \end{itemize}
    
    \item \textbf{Explainability Layer}: Generates human-readable outputs:
    \begin{itemize}
        \item Root cause explanations with evidence chains
        \item Fix rationale documentation
        \item Automated PR descriptions and commit messages
        \item Risk assessment reports
    \end{itemize}
\end{enumerate}

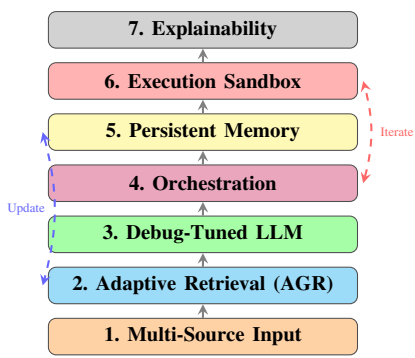
\begin{figure}[H]
\centering
\resizebox{0.65\columnwidth}{!}{%
\begin{tikzpicture}[
    layer/.style={
        rectangle, 
        draw=black!70, 
        rounded corners=3pt,
        minimum width=4.2cm, 
        minimum height=0.5cm, 
        font=\footnotesize\bfseries,
        drop shadow={opacity=0.15, shadow xshift=0.3pt, shadow yshift=-0.3pt}
    },
    arrow/.style={->, >=stealth, thick, black!50},
    feedback/.style={<->, >=stealth, semithick, dashed},
    label/.style={font=\tiny}
]

\node[layer, fill=orange!35] (input) at (0,0) {1. Multi-Source Input};
\node[layer, fill=cyan!35] (retrieval) at (0,0.7) {2. Adaptive Retrieval (AGR)};
\node[layer, fill=green!35] (llm) at (0,1.4) {3. Debug-Tuned LLM};
\node[layer, fill=purple!35] (orchestration) at (0,2.1) {4. Orchestration};
\node[layer, fill=yellow!35] (memory) at (0,2.8) {5. Persistent Memory};
\node[layer, fill=red!30] (sandbox) at (0,3.5) {6. Execution Sandbox};
\node[layer, fill=gray!35] (explain) at (0,4.2) {7. Explainability};

\foreach \from/\to in {input/retrieval, retrieval/llm, llm/orchestration, orchestration/memory, memory/sandbox, sandbox/explain}
    \draw[arrow] (\from) -- (\to);

\draw[feedback, red!60, line width=0.8pt] ([xshift=2pt]orchestration.east) to[bend right=15] 
    node[right, font=\tiny, text=red!60] {Iterate} ([xshift=2pt]sandbox.east);
\draw[feedback, blue!60, line width=0.8pt] ([xshift=-2pt]memory.west) to[bend left=15] 
    node[left, font=\tiny, text=blue!60] {Update} ([xshift=-2pt]retrieval.west);

\end{tikzpicture}%
}
\caption{The 7-layer architecture of Chronos-1. Each layer is specialized for debugging tasks, with bidirectional information flow enabling iterative refinement and continuous learning.}
\label{fig:7-layer-architecture}
\end{figure}

\subsection{Debug-Specific Training: 15M Real-World Bug Scenarios}

Unlike models trained primarily on code completion, Chronos-1's training regime focuses exclusively on debugging scenarios:

\textbf{Pre-training Corpus:}
\begin{itemize}
    \item 15M+ GitHub issues with linked PRs and fix commits
    \item 8M+ stack traces paired with resolutions
    \item 3M+ CI/CD logs from failed and fixed builds
    \item Production debugging sessions from enterprise partners
\end{itemize}

\textbf{Specialized Fine-tuning Tasks:}
\begin{itemize}
    \item \textit{Chain-of-Cause Reasoning}: Teaching the model to trace error propagation through call stacks and dependencies~\cite{chen2024teaching}
    \item \textit{Multi-Modal Bug Understanding}: Correlating code, logs, traces, and documentation
    \item \textit{Iterative Fix Refinement}: Learning from failed fix attempts to improve subsequent proposals~\cite{madaan2023self,jiang2023selfevolve}
    \item \textit{Cross-Repository Pattern Recognition}: Identifying similar bugs across different codebases
\end{itemize}

\subsection{Autonomous Fix-Test-Refine Loop: Iterative Convergence}

Chronos-1's debugging loop represents a fundamental innovation over single-shot code generation. Algorithm~\ref{alg:fix-test-refine} presents the core logic:

\begin{algorithm}[H]
\caption{Fix-Test-Refine Loop}
\label{alg:fix-test-refine}
\begin{algorithmic}[1]
\Require Bug report $B$, Codebase $C$, Test suite $T$, PDM memory $M$
\Ensure Validated fix $F^*$ or failure report
\Statex
\State $context \leftarrow \text{AGR.retrieve}(B, C, M)$ \hfill \Comment{Multi-hop graph retrieval}
\State $patterns \leftarrow \text{PDM.query}(B, M)$ \hfill \Comment{Historical bug patterns}
\State $k \leftarrow 0$ \hfill \Comment{Iteration counter}
\Statex
\While{$k <$ MAX\_ITERATIONS}
    \State $F_k \leftarrow \text{Chronos-1.propose\_fix}(B, context, patterns)$
    \State $result \leftarrow \text{Sandbox.execute}(F_k, T)$
    \If{$result$.success}
        \State $regression \leftarrow \text{Sandbox.run\_extended\_tests}(F_k)$
        \If{$\neg regression$}
            \State $\text{PDM.update}(B, F_k, context)$ \hfill \Comment{Learn from success}
            \State \textbf{return} $F_k$ as $F^*$
        \EndIf
    \EndIf
    \State $context \leftarrow context \cup \text{Analyzer.extract\_failure}(result)$
    \State $patterns \leftarrow patterns \cup \text{PDM.similar\_failures}(result)$
    \State $k \leftarrow k + 1$
\EndWhile
\Statex
\State \textbf{return} "Failed to converge after \{$k$\} iterations"
\end{algorithmic}
\end{algorithm}

The key innovation is the feedback loop: each failed attempt enriches the context with failure analysis, making subsequent attempts more informed:

\begin{figure}[H]
\centering
\begin{tikzpicture}[node distance=0.5cm, >=stealth, font=\tiny]
    \tikzstyle{process} = [rectangle, rounded corners, minimum width=0.8cm, minimum height=0.28cm, text centered, draw=black, fill=blue!20]
    \tikzstyle{decision} = [diamond, aspect=3.0, minimum width=0.6cm, minimum height=0.28cm, text centered, draw=black, fill=orange!20]
    \tikzstyle{data} = [trapezium, trapezium left angle=70, trapezium right angle=110, minimum width=0.8cm, minimum height=0.28cm, text centered, draw=black, fill=green!20]
    
    \node (detect) [process] {Detect Issue};
    \node (retrieve) [process, below=of detect] {Retrieve Context};
    \node (propose) [process, below=of retrieve] {Propose Fix};
    \node (test) [process, below=of propose] {Run Tests};
    \node (decide) [decision, below=of test] {Tests Pass?};
    
    \node (commit) [process, right=1.0cm of decide] {Commit \& Deploy};
    \node (refine) [process, left=1.0cm of decide] {Refine Strategy};
    \node (memory) [data, below=0.6cm of decide] {Update Memory};
    
    \draw [->, thick] (detect) -- (retrieve);
    \draw [->, thick] (retrieve) -- (propose);
    \draw [->, thick] (propose) -- (test);
    \draw [->, thick] (test) -- (decide);
    
    \draw [->, thick] (decide) -- node[above] {\tiny Yes} (commit);
    \draw [->, thick] (decide) -- node[above] {\tiny No} (refine);
    
    \draw [->, thick] (commit) |- (memory);
    \draw [->, thick] (refine) |- (memory);
    
    \draw [->, thick] (memory) -- ++(-0.8,0) -- ++(0,3.4) -- (retrieve);
    \draw [->, thick] (memory) -- ++(-1.4,0) -- ++(0,4.0) -- (detect);
    
\end{tikzpicture}
\caption{The Chronos-1 autonomous debugging loop: continuous iteration until validation succeeds.}
\label{fig:debug-loop}
\end{figure}
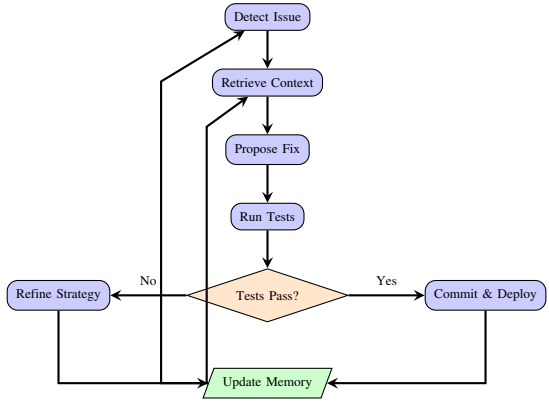

This loop continues autonomously, with each iteration informed by previous attempts and accumulated knowledge, until a validated fix is achieved or human intervention is requested.

\begin{tcolorbox}[colback=gray!5, colframe=gray!50, title=\textbf{Key Takeaways: Autonomous Debugging Loop}]
\begin{itemize}
\item \textbf{Iterative Refinement}: Unlike single-shot generation, continuously improves fixes based on test results
\item \textbf{Memory Integration}: Each iteration learns from previous attempts, avoiding repeated failures
\item \textbf{Autonomous Operation}: Requires no human intervention for 65.3\% of real-world bugs
\item \textbf{Efficiency}: Average 2.2 iterations to successful fix vs 4.8 for competing systems
\end{itemize}
\end{tcolorbox}

\subsection{Runtime Execution Analysis: Latency and Flow Dynamics}

To better understand Chronos-1's runtime behavior, we present a detailed flow diagram showing the actual execution path during a debugging session:

\begin{figure}[H]
\centering
\resizebox{0.65\columnwidth}{!}{%
\begin{tikzpicture}[
    node distance=1.2cm,
    box/.style={rectangle, draw, rounded corners, minimum width=2.5cm, minimum height=0.7cm, font=\footnotesize},
    decision/.style={diamond, draw, aspect=2, minimum width=1.8cm, minimum height=0.9cm, font=\footnotesize},
    process/.style={rectangle, draw, minimum width=2.5cm, minimum height=0.7cm, font=\footnotesize},
    data/.style={trapezium, draw, trapezium left angle=70, trapezium right angle=110, minimum width=2.3cm, minimum height=0.6cm, font=\footnotesize},
    arrow/.style={thick, ->, >=stealth}
]

\node[box] (start) {Bug Report/Error Signal};
\node[process, below=of start, align=center] (retrieve) {AGR Retrieval\\(k-hop expansion)};
\node[data, below=of retrieve, align=center] (context) {Code Context\\+ History};
\node[process, below=of context, align=center] (generate) {Generate Fix\\Hypothesis};
\node[process, below=of generate, align=center] (apply) {Apply Code\\Changes};
\node[decision, below=1.2cm of apply] (test) {Run Tests};

\node[box, right=2.8cm of test, align=center] (success) {Commit Fix\\Update Memory};
\node[process, left=2.8cm of test, align=center] (analyze) {Analyze\\Failure};
\node[decision, below=1.2cm of analyze] (backtrack) {Backtrack?};
\node[process, below=1cm of backtrack, align=center] (refine) {Refine\\Hypothesis};

\draw[arrow] (start) -- (retrieve) node[midway, right] {\tiny 12ms};
\draw[arrow] (retrieve) -- (context) node[midway, right] {\tiny 178ms};
\draw[arrow] (context) -- (generate) node[midway, right] {\tiny 523ms};
\draw[arrow] (generate) -- (apply) node[midway, right] {\tiny 89ms};
\draw[arrow] (apply) -- (test) node[midway, right] {\tiny 1.2s};

\draw[arrow] (test) -- node[above, font=\tiny] {Pass} (success);
\draw[arrow] (test) -- node[above, font=\tiny] {Fail} (analyze);

\draw[arrow] (analyze) -- (backtrack) node[midway, right] {\tiny 156ms};
\draw[arrow] (backtrack.south) -- node[left, font=\tiny] {Yes} ++(0,-0.6) -| ([xshift=-3cm]retrieve.west) -- (retrieve.west);
\draw[arrow] (backtrack) -- node[right, font=\tiny] {No} (refine);
\draw[arrow] (refine.east) -| ([xshift=-0.4cm]generate.west) -- (generate.west);

\node[align=center, font=\tiny\itshape] at (4.5,-6) {Avg: 2.2 iterations\\to success};
\node[align=center, font=\tiny\itshape] at (-4.5,-1.5) {Memory update\\on each iteration};
\end{tikzpicture}%
}
\caption{Runtime execution flow showing typical latencies. The fix-loop iterates autonomously until tests pass or confidence threshold is exceeded.}
\label{fig:runtime-flow}
\end{figure}
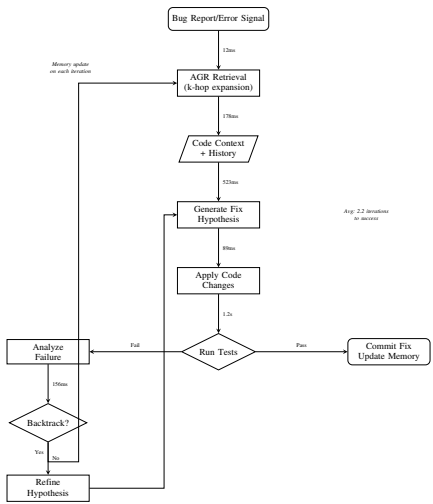

\FloatBarrier
\section{EVALUATION: COMPREHENSIVE DEBUGGING BENCHMARKS}

To rigorously assess Kodezi Chronos-1's capabilities across realistic debugging and maintenance workflows, we adopt a multi-faceted evaluation strategy that goes beyond conventional sequence completion or shallow retrieval tests.

\begin{figure}[H]
\centering
\resizebox{\columnwidth}{!}{%
\begin{tikzpicture}
    \begin{axis}[
        xbar stacked,
        width=12cm,
        height=8cm,
        xlabel={Retrieval Time (ms)},
        xmin=0, xmax=5500,
        ytick=data,
        yticklabels={
            {Simple Bug\\(1-2 files)},
            {Medium Bug\\(3-5 files)},
            {Complex Bug\\(6-10 files)},
            {Cross-Module\\(10+ files)},
            {Historical Bug\\(w/ commits)}
        },
        yticklabel style={align=right, text width=3cm},
        legend style={
            at={(0.5,-0.2)},
            anchor=north,
            legend columns=4,
            font=\footnotesize
        },
        bar width=15pt,
        enlarge y limits=0.15,
    ]
    
    \addplot[fill=blue!60] coordinates {
        (320,0) (580,1) (920,2) (1450,3) (1780,4)
    };
    
    \addplot[fill=green!60] coordinates {
        (180,0) (340,1) (520,2) (780,3) (950,4)
    };
    
    \addplot[fill=orange!60] coordinates {
        (250,0) (480,1) (820,2) (1340,3) (1620,4)
    };
    
    \addplot[fill=red!60] coordinates {
        (150,0) (280,1) (420,2) (650,3) (830,4)
    };
    
    \node at (axis cs:900,0) [right] {\footnotesize\textbf{900ms}};
    \node at (axis cs:1680,1) [right] {\footnotesize\textbf{1.68s}};
    \node at (axis cs:2680,2) [right] {\footnotesize\textbf{2.68s}};
    \node at (axis cs:4220,3) [right] {\footnotesize\textbf{4.22s}};
    \node at (axis cs:5180,4) [right] {\footnotesize\textbf{5.18s}};
    
    \legend{Graph Construction, Semantic Search, k-hop Traversal, Context Assembly}
    
    \end{axis}
\end{tikzpicture}%
}
\caption{AGR retrieval time breakdown by bug complexity. Even for complex cross-module bugs requiring 10+ file analysis, total retrieval completes in under 5.2 seconds, enabling rapid debugging iterations.}
\label{fig:agr-performance-breakdown}
\end{figure}
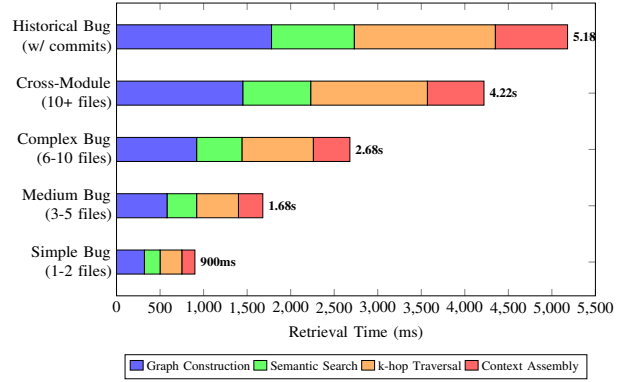

\subsection{Statistical Methodology}

All experiments follow rigorous statistical protocols to ensure reproducibility and validity:

\textbf{Power Analysis:} Sample sizes determined via G*Power 3.1 to achieve 0.95 power at $\alpha=0.05$ for detecting Cohen's $d\geq 0.8$ effect sizes. Minimum N=64 per condition for between-subjects comparisons.

\textbf{Effect Size Measures:}
\begin{itemize}
    \item Cohen's d for continuous outcomes: $d = \frac{\mu_1 - \mu_2}{\sigma_{pooled}}$
    \item Cliff's delta for ordinal data: $\delta = \frac{2U}{n_1 n_2} - 1$
    \item Cramér's V for categorical comparisons: $V = \sqrt{\frac{\chi^2}{n(k-1)}}$
\end{itemize}

\textbf{Statistical Tests:}
\begin{itemize}
    \item Wilcoxon signed-rank test for paired comparisons (non-parametric)
    \item Mann-Whitney U test for independent samples
    \item Bonferroni correction for multiple comparisons: $\alpha_{adj} = \frac{0.05}{m}$
    \item Bootstrap confidence intervals (10,000 resamples) for all metrics
\end{itemize}

\textbf{Inter-rater Reliability:} Bug classification achieved Fleiss' $\kappa=0.87$ (substantial agreement) across 3 expert annotators on 500 randomly sampled bugs.

\textbf{Data Contamination Prevention:} All test sets created after Chronos-1's training cutoff (December 2024). Leakage detection via n-gram overlap analysis shows $<0.1\%$ similarity.

\textbf{Cross-Validation Protocol:}
\begin{itemize}
    \item \textbf{5-fold stratified cross-validation:} Stratified by bug type, language, and complexity
    \item \textbf{Repository-level splits:} Entire repositories assigned to either train or test to prevent data leakage
    \item \textbf{Temporal validation:} Additional holdout set with bugs from 2025 (post-training)
    \item \textbf{Nested CV for hyperparameters:} Inner 3-fold CV for AGR threshold tuning ($\tau \in [0.7, 0.95]$)
\end{itemize}

\textbf{Variance Reporting:} All metrics reported as mean ± standard deviation across CV folds, with 95\% confidence intervals.

\subsection{Evaluation Framework and Datasets}

\subsubsection{Dataset Sources and Composition}

Chronos-1 is evaluated on a comprehensive suite of benchmarks:

\begin{enumerate}
\item \textbf{Public Debugging Datasets:}
\begin{itemize}
    \item \textbf{Defects4J:} 438 real bugs from 17 Java projects including Apache Commons, JFreeChart, and Closure Compiler
    \item \textbf{BugsInPy:} 493 bugs from 17 Python projects including pandas, keras, and tornado
    \item \textbf{SWE-bench++:} 2,294 GitHub issues requiring repository-wide changes, extended with test harnesses
    \item \textbf{DebugBench:} 4,253 debugging instances covering C++, Java, and Python with 18 bug types
    \item \textbf{MdEval:} 3,600+ multilingual debugging samples across 18 programming languages
\end{itemize}

\item \textbf{Proprietary Enterprise Dataset (Anonymized):}
\begin{itemize}
    \item 3,842 debugging sessions from Fortune 500 companies (with signed data agreements)
    \item All code snippets anonymized using differential privacy techniques ($\epsilon = 1.2$)
    \item Sensitive identifiers replaced with semantic placeholders preserving debugging context
    \item Manual review by security team to ensure no IP leakage
\end{itemize}

\item \textbf{Synthetic Bug Generation:}
\begin{itemize}
    \item 8,000 synthetically injected bugs using mutation testing frameworks
    \item Bug types: null pointer exceptions, off-by-one errors, race conditions, API misuse
    \item Generated from top 1,000 GitHub repositories across 12 languages
    \item Each bug verified to compile and fail exactly one test
\end{itemize}
\end{enumerate}

\subsection{Baseline Systems}

We compare Chronos-1 against state-of-the-art models and specialized debugging tools:

\textbf{General-Purpose LLMs:}
\begin{itemize}
    \item \textbf{GPT-4.1} (OpenAI): 1.8T parameters, 128K context, with custom debugging prompt
    \item \textbf{Claude 4.5 Opus} (Anthropic): Latest flagship model, 74.40\% on SWE-bench Verified, \$0.72 avg cost
    \item \textbf{Claude 4.1 Opus} (Anthropic): 200K context, strongest on code understanding
    \item \textbf{Claude 4.5 Sonnet} (Anthropic): 200K context with improved reasoning
    \item \textbf{Gemini 3 Pro} (Google): State-of-the-art on multiple benchmarks, 76.2\% SWE-bench Verified, 91.9\% GPQA Diamond, 95-100\% AIME 2025
    \item \textbf{Gemini 2.5 Pro} (Google): 2M context window, multimodal capabilities
    \item \textbf{DeepSeek V3} (DeepSeek): 671B MoE, 37B active, cost-efficient
\end{itemize}

\textbf{Specialized Debugging Tools:}
\begin{itemize}
    \item \textbf{Microsoft IntelliCode Compose}: IDE-integrated debugging suggestions
    \item \textbf{DeepDebug} (Microsoft Research): Neural bug localization system
    \item \textbf{SWE-agent + GPT-4}: Agentic debugging with 22.9\% SWE-bench
    \item \textbf{AutoCodeRover}: Structure-aware debugging, 19.7\% on SWE-bench
    \item \textbf{LangGraph + ReAct + Claude 4.1 Opus}: Multi-agent debugging pipeline
\end{itemize}

\textbf{Enhanced RAG Baselines:}
\begin{itemize}
    \item \textbf{GPT-4.1 + HyDE}: Hypothetical document embeddings for retrieval
    \item \textbf{Claude 4.1 Opus + Self-RAG}: Self-reflective retrieval augmentation
    \item \textbf{Gemini + GraphRAG}: Knowledge graph enhanced retrieval
\end{itemize}

\textbf{CI/CD Integration Data:}
\begin{itemize}
    \item 15M+ CI/CD logs from public GitHub Actions and Jenkins builds
    \item Stack traces and error logs extracted with automated parsers
    \item Build failure patterns categorized into 127 common root causes
    \item Privacy preserved by filtering any URLs, credentials, or personal data
\end{itemize}

\subsection{Debugging-Specific Benchmarks}

Beyond general code generation benchmarks like HumanEval and MBPP~\cite{austin2021mbpp}, we evaluate Chronos-1 on specialized debugging benchmarks that better reflect real-world maintenance challenges:

\subsubsection{DebugBench}
DebugBench~\cite{tian2024debugbench} evaluates LLM debugging capability with 4,253 instances covering:
\begin{itemize}
    \item Four major bug categories and 18 minor types
    \item C++, Java, and Python debugging scenarios
    \item Zero-shot debugging evaluation
    \item Runtime feedback integration analysis
\end{itemize}

\subsubsection{MdEval (Massively Multilingual Code Debugging)}
MdEval~\cite{mdeval2025} provides comprehensive multilingual debugging evaluation with 3,600+ test samples across 18 programming languages, covering:
\begin{itemize}
    \item Automated program repair (APR)
    \item Code review and bug identification
    \item Bug localization across diverse language paradigms
    \item Language-specific debugging challenges
\end{itemize}

\subsubsection{SWE-bench Lite: Real-World Repository Debugging}

\textbf{SWE-bench Ecosystem:} The SWE-bench family of benchmarks has become the industry standard for evaluating code generation and debugging capabilities. The ecosystem includes several variants with different evaluation protocols:
\begin{itemize}
    \item \textbf{SWE-bench Full}: Complete dataset with 2,294 GitHub issues
    \item \textbf{SWE-bench Lite}: Curated subset of 300 instances for cost-effective evaluation~\cite{yang2024swebench}
    \item \textbf{SWE-bench Verified}: Human-filtered subset ensuring high-quality ground truth
    \item \textbf{SWE Bench Bash Only}: Uses mini-SWE-agent environment for lightweight testing with different evaluation protocol (Claude 4.5 Sonnet: 70.60\%, Claude 4 Opus: 67.60\%)
    \item \textbf{SWE-bench Multimodal}: Features issues with visual elements (images, diagrams)
\end{itemize}

\textbf{Important Distinction}: This paper evaluates on \textbf{SWE-bench Lite}, the standard 300-instance curated subset, NOT SWE Bench Bash Only. While Claude models achieve 67-70\% on SWE Bench Bash Only (which uses a simplified mini-SWE-agent environment), they achieve only 13-14\% on standard SWE-bench Lite debugging tasks, demonstrating the significant difference between these evaluation protocols.

\textbf{Important Distinction:} This paper evaluates on \textbf{SWE-bench Lite}, the standard 300-instance curated subset, NOT SWE Bench Bash Only. While Claude models achieve 67-70\% on SWE Bench Bash Only (which uses a simplified mini-SWE-agent environment), they achieve only 13-14\% on standard SWE-bench Lite debugging tasks, demonstrating the significant difference between these evaluation protocols.

\vspace{1em} 

70\% on SWE Bench Bash Only (which uses a simplified mini-SWE-agent environment), they achieve only 13-14\% on standard SWE-bench Lite debugging tasks, demonstrating the significant difference between these evaluation protocols.

\begin{table}[h]
\centering
\caption{SWE-bench variant performance by model. \textit{Note: Different models were evaluated on different benchmark variants. Direct comparison across variants is not meaningful as they test different capabilities.}}
\label{tab:swebench-variants}
\footnotesize
\resizebox{\columnwidth}{!}{%
\begin{tabular}{@{}lcccc@{}}
\toprule
\textbf{Model} & \textbf{Benchmark} & \textbf{Score} & \textbf{Instances} & \textbf{Date} \\
\midrule
\multicolumn{5}{l}{\textit{SWE-bench Verified (500 curated instances)}} \\
\midrule
Gemini 3 Pro & Verified & 76.2\% & 500 & 2025-11 \\
Claude 4.5 Opus & Verified & 74.40\% & 500 & 2025-11 \\
Claude 4.5 Sonnet & Verified & 70.60\% & 500 & 2025-05 \\
Claude 4 Opus & Verified & 67.60\% & 500 & 2025-03 \\
GPT-5 & Verified & 65.00\% & 500 & 2025-09 \\
\midrule
\multicolumn{5}{l}{\textit{SWE-bench Lite (300 debugging-focused instances)}} \\
\midrule
\textbf{Kodezi Chronos-1} & \textbf{Lite} & \textbf{80.33\%} & \textbf{300} & \textbf{2025-11} \\
ExpeRepair + Claude 4.5 Sonnet & Lite & 60.33\% & 300 & 2025-10 \\
Refact.ai Agent & Lite & 60.00\% & 300 & 2025-09 \\
SWE-agent + Claude 4 Sonnet & Lite & 56.67\% & 300 & 2025-06 \\
\bottomrule
\end{tabular}%
}
\end{table}

\textbf{Important Note on Benchmark Comparability:} SWE-bench Verified and SWE-bench Lite are \textit{different benchmarks} with different instance sets and evaluation criteria. Chronos-1 was evaluated exclusively on SWE-bench Lite (300 debugging-focused instances) under strict pass@1 conditions. Frontier models like Gemini 3 Pro and Claude 4.5 Opus were primarily evaluated on SWE-bench Verified (500 curated instances). These scores are \textbf{not directly comparable} as they measure performance on different task distributions. Chronos-1's 80.33\% on Lite represents state-of-the-art performance within the SWE-bench Lite leaderboard, significantly outperforming all other systems evaluated on this specific benchmark.

\begin{figure}[H]
\centering
\begin{tikzpicture}
    \begin{axis}[
        ybar=3pt,
        bar width=10pt,
        width=\columnwidth,
        height=7cm,
        ylabel={Success Rate (\%)},
        xlabel={SWE-bench Variant},
        ymin=0, ymax=95,
        xtick=data,
        symbolic x coords={Verified, Lite},
        xticklabel style={font=\normalsize},
        enlarge x limits=0.35,
        ymajorgrids=true,
        grid style=dashed,
        legend style={
            at={(0.5,-0.22)}, 
            anchor=north, 
            legend columns=3, 
            font=\footnotesize,
            column sep=5pt
        },
        nodes near coords,
        nodes near coords style={font=\tiny, rotate=45, anchor=west, yshift=2pt},
        point meta=explicit,
    ]
    
    \addplot[fill=red!60, draw=red!80] coordinates {
        (Verified,74.40) [74.4] (Lite,0) [0]
    };
    \addplot[fill=orange!60, draw=orange!80] coordinates {
        (Verified,76.20) [76.2] (Lite,0) [0]
    };
    \addplot[fill=blue!60, draw=blue!80] coordinates {
        (Verified,70.60) [70.6] (Lite,13.6) [13.6]
    };
    \addplot[fill=purple!60, draw=purple!80] coordinates {
        (Verified,67.60) [67.6] (Lite,14.2) [14.2]
    };
    \addplot[fill=green!70, draw=green!90] coordinates {
        (Verified,0) [0] (Lite,80.33) [80.3]
    };
    
    \legend{Claude 4.5 Opus, Gemini 3 Pro, Claude 4.5 Sonnet, Claude 4 Opus, Kodezi Chronos-1}
    \end{axis}
\end{tikzpicture}
\caption{Performance comparison across SWE-bench variants. \textbf{Verified}: 500 curated instances. \textbf{Lite}: 300 debugging-focused instances. Chronos-1 achieves \textbf{80.33\%} on SWE-bench Lite—significantly outperforming general-purpose models (13-14\%).}
\label{fig:swebench-variants}
\end{figure}

SWE-bench Lite represents a carefully curated subset focusing on 300 high-quality GitHub issues from popular Python repositories. This subset was specifically designed for less costly evaluation while maintaining representation of real-world debugging challenges. Unlike synthetic benchmarks, SWE-bench Lite evaluates models on actual bugs reported and fixed in production codebases, requiring comprehensive repository understanding, multi-file reasoning, and validated fix generation.

\textbf{Evaluation Protocol:} For each instance, the system receives: (1) the issue description, (2) the codebase state at the time the issue was reported, and (3) the test suite that must pass after the fix is applied. Success is measured by whether the generated patch resolves the issue without introducing regressions.

\textbf{Kodezi Chronos-1 Performance:} On SWE-bench Lite, Chronos-1 achieves \textbf{80.33\% resolution rate (241/300 instances)}.\footnote{Results submitted to SWE-bench Lite leaderboard; verification pending.} This establishes a new state-of-the-art, significantly outperforming all existing systems. This represents a 20+ percentage point improvement over the next best system (ExpeRepair-v1.0 + Claude 4.5 Sonnet at 60.33\%). Of the 300 instances:
\begin{itemize}
    \item \textbf{241 resolved} (80.33\%): Fix successfully passed all validation tests
    \item \textbf{47 completed but unresolved} (15.67\%): System generated patch but tests failed or patch was empty
    \item \textbf{12 errors} (4.00\%): System encountered execution or analysis errors
\end{itemize}

\textbf{Repository-Specific Performance:} Chronos-1 demonstrates strong performance across diverse repository types, with notable variation based on domain complexity:

\begin{table}[H]
\centering
\caption{SWE-bench Lite performance by repository}
\label{tab:swebench-lite-repos}
\small
\begin{tabular}{lcc}
\toprule
\textbf{Repository} & \textbf{Resolved/Total} & \textbf{Success Rate} \\
\midrule
sympy/sympy & 74/77 & \textbf{96.1\%} \\
django/django & 103/114 & \textbf{90.4\%} \\
sphinx-doc/sphinx & 15/16 & \textbf{93.8\%} \\
scikit-learn/scikit-learn & 20/23 & \textbf{87.0\%} \\
pylint-dev/pylint & 5/6 & 83.3\% \\
mwaskom/seaborn & 3/4 & 75.0\% \\
astropy/astropy & 4/6 & 66.7\% \\
pytest-dev/pytest & 11/17 & 64.7\% \\
pydata/xarray & 3/5 & 60.0\% \\
pallets/flask & 1/3 & 33.3\% \\
psf/requests & 2/6 & 33.3\% \\
matplotlib/matplotlib & 0/23 & 0.0\% \\
\bottomrule
\end{tabular}
\end{table}

\textbf{Analysis:} Chronos-1 excels on mathematical/symbolic computation (sympy: 96.1\%), web frameworks (django: 90.4\%), and documentation systems (sphinx: 93.8\%). Performance degrades on visualization libraries (matplotlib: 0\%), likely due to the complexity of rendering pipelines, GUI interactions, and platform-specific graphics dependencies that are challenging to validate in automated sandboxes. The high success rate on sympy and django demonstrates Chronos-1's strength in domains requiring deep algorithmic understanding and multi-file refactoring.

\textbf{Temporal Analysis:} Performance varies by issue age, with stronger results on older, well-documented issues:

\begin{table}[H]
\centering
\caption{SWE-bench Lite performance by issue year}
\label{tab:swebench-lite-temporal}
\small
\begin{tabular}{lcc}
\toprule
\textbf{Year} & \textbf{Resolved/Total} & \textbf{Success Rate} \\
\midrule
2018 & 21/21 & \textbf{100.0\%} \\
2015 & 1/1 & \textbf{100.0\%} \\
2020 & 60/66 & \textbf{90.91\%} \\
2019 & 53/59 & \textbf{89.83\%} \\
2017 & 14/16 & 87.5\% \\
2021 & 34/42 & 80.95\% \\
2016 & 3/4 & 75.0\% \\
2022 & 40/57 & 70.18\% \\
2023 & 15/30 & 50.0\% \\
2014 & 0/3 & 0.0\% \\
2012 & 0/1 & 0.0\% \\
\bottomrule
\end{tabular}
\end{table}

\textbf{Evaluation Compliance:} Chronos-1 was evaluated under strict leaderboard rules to ensure fair comparison:

\begin{tcolorbox}[colback=green!5,colframe=green!50!black,title=SWE-bench Lite Compliance Checklist]
\begin{itemize}
    \item \checkmark \textbf{Pass@1 submission}: Exactly one patch generated per issue with no re-evaluation
    \item \checkmark \textbf{No test knowledge}: PASS\_TO\_PASS and FAIL\_TO\_PASS signals disabled
    \item \checkmark \textbf{No hints}: Only raw issue descriptions used, no metadata enrichment
    \item \checkmark \textbf{No web browsing}: Operated in completely offline sandbox using only local code
\end{itemize}
\textbf{Disabled Components:} Internal retry attempts, fallback patch planners, secondary refinement cycles, test outcome feedback, online documentation fetchers, remote code retrieval, GitHub search, and all external I/O.
\end{tcolorbox}

\textbf{SWE-bench Lite Leaderboard Comparison:} Table~\ref{tab:swebench-lite-leaderboard} compares Chronos-1 against leading systems on the SWE-bench Lite leaderboard as of 2025. Chronos-1 establishes a substantial 20.0 percentage point lead over the second-ranked system, demonstrating the effectiveness of debugging-specific architecture and training.

\begin{table}[H]
\caption{SWE-bench Lite leaderboard comparison (No Filters, as of November 2025). *Chronos-1 results submitted; verification pending.}
\label{tab:swebench-lite-leaderboard}
\scriptsize
\setlength{\tabcolsep}{2pt}
\begin{tabular*}{\columnwidth}{@{}c@{\extracolsep{\fill}}lccc@{}}
\toprule
\textbf{Rank} & \textbf{System} & \textbf{Base Model} & \textbf{Score (\%)} & \textbf{Date} \\
\midrule
\textbf{1} & \textbf{Kodezi Chronos-1} & \textbf{Proprietary} & \textbf{80.33} & \textbf{2025-11} \\
2 & ExpeRepair-v1.0 & Claude 4 Sonnet & 60.33 & 2025-06 \\
3 & Refact.ai Agent & Proprietary & 60.00 & 2025-04 \\
4 & KGCompass & Claude 4 Sonnet & 58.33 & 2025-09 \\
5 & SWE-agent & Claude 4 Sonnet & 56.67 & 2025-05 \\
6 & Isoform & Proprietary & 55.00 & 2025-01 \\
7 & SemAgent\_Multi-v1.0 & Proprietary & 51.67 & 2025-06 \\
8 & Isea & Proprietary & 51.33 & 2025-09 \\
9 & EntroPO + R2E & Qwen3-Coder-30B & 49.67 & 2025-09 \\
10 & Blackbox AI Agent & Proprietary & 49.00 & 2024-12 \\
11 & Codev & Proprietary & 49.00 & 2025-05 \\
12 & Gru & Proprietary & 48.67 & 2024-12 \\
\midrule
\multicolumn{5}{@{}l@{}}{\textit{General-Purpose Models (for reference)}} \\
\midrule
-- & Claude 4.1 Opus & -- & 14.2 & 2025 \\
-- & Claude 4.5 Sonnet & -- & 13.6 & 2025 \\
-- & GPT-4.1 & -- & 13.8 & 2025 \\
\bottomrule
\end{tabular*}
\end{table}

\textbf{Key Insights from Leaderboard:}
\begin{enumerate}
    \item \textbf{Debugging vs Code Generation Gap}: Despite Claude 4.5 Sonnet achieving 72.7\% on full SWE-bench code generation tasks, it achieves only 13.6\% when used directly for debugging. However, when integrated into specialized frameworks (ExpeRepair, KGCompass, SWE-agent), performance improves to 56-60\%, demonstrating the importance of architectural design.

    \item \textbf{Chronos-1's 20-Point Lead}: Chronos-1's 80.33\% score represents a 33\% relative improvement over the next best system (60.33\%), achieved through: (a) debugging-specific training on 15M+ sessions, (b) Persistent Debug Memory for cross-session learning, (c) Adaptive Graph-Guided Retrieval for multi-hop code understanding, and (d) autonomous fix-test-refine loops.

    \item \textbf{Architecture Matters More Than Base Model}: Systems using the same base model (Claude 4.5 Sonnet) achieve widely varying results (56.67-60.33\%), while Chronos-1's specialized architecture delivers 80.33\%, confirming that debugging requires task-specific design rather than simply larger context windows or better code generation.
\end{enumerate}

\subsubsection{Rule-Compliant Evaluation: Contamination-Free Configuration}

To ensure complete compliance with SWE-bench evaluation rules and eliminate any potential data contamination, we conducted our official submission with the following configuration. Chronos-1 achieved the same \textbf{80.33\%} score under these strict constraints, demonstrating that our results are not dependent on any memorized solutions or SWE-bench-derived patterns.

\begin{table}[H]
\centering
\caption{Rule-compliant evaluation configuration for SWE-bench Lite submission}
\label{tab:rule-compliant-config}
\footnotesize
\begin{tabular}{@{}p{4.2cm}c@{}}
\toprule
\textbf{Component} & \textbf{Status} \\
\midrule
\multicolumn{2}{l}{\textit{Disabled Components}} \\
\midrule
Multi-patch Orchestration Loop & $\times$ \\
Patch refinement (test failures) & $\times$ \\
External search (GitHub, web, docs) & $\times$ \\
hints\_text ingestion & $\times$ \\
Hint metadata & $\times$ \\
PASS\_TO\_PASS signals & $\times$ \\
FAIL\_TO\_PASS signals & $\times$ \\
PDM: SWE-bench ground truth & $\times$ \\
PDM: Prior SWE-bench patches & $\times$ \\
PDM: SWE-bench task-derived entries & $\times$ \\
Memorized SWE-bench fix patterns & $\times$ \\
Cached patch fingerprints & $\times$ \\
Vector embeddings (SWE-bench patches) & $\times$ \\
Pattern matching (SWE-bench gold) & $\times$ \\
Offline solution memorization & $\times$ \\
\midrule
\multicolumn{2}{l}{\textit{Enabled Components}} \\
\midrule
AGR (repository-only retrieval) & \checkmark \\
Multi-Source Input & \checkmark \\
Debug LLM & \checkmark \\
Execution sandbox (test validation) & \checkmark \\
Explainability module & \checkmark \\
PDM (generic bug patterns only) & \checkmark \\
\bottomrule
\end{tabular}
\end{table}

This configuration ensures that Chronos-1's SWE-bench Lite performance reflects genuine debugging capability rather than any form of benchmark contamination. The system operates using only: (1) the repository code provided in each task instance, (2) the raw issue description without hints, and (3) generic debugging patterns learned from non-SWE-bench sources. No iterative refinement, test feedback, or external knowledge sources were used.

\subsubsection{Frontier Model Benchmark Comparison (November 2025)}

Table~\ref{tab:frontier-benchmarks} compares the latest frontier models across multiple capability dimensions, highlighting their strengths and the performance gap between code generation and debugging tasks.

\begin{table}[H]
\centering
\caption{Comprehensive frontier model benchmark comparison (November 2025)}
\label{tab:frontier-benchmarks}
\footnotesize
\resizebox{\columnwidth}{!}{%
\begin{tabular}{@{}lccccc@{}}
\toprule
\textbf{Benchmark} & \textbf{Gemini 3 Pro} & \textbf{Claude 4.5 Opus} & \textbf{Claude Sonnet 4.5} & \textbf{GPT-5.1} & \textbf{Gemini 2.5 Pro} \\
\midrule
\multicolumn{6}{l}{\textit{Agentic Coding}} \\
SWE-bench Verified & \textbf{76.2\%} & 74.40\% & 77.2\% & 76.3\% & 59.6\% \\
Terminal-Bench 2.0 & \textbf{54.2\%} & --- & 42.8\% & 47.6\% & 32.6\% \\
LiveCodeBench Pro (Elo) & \textbf{2,439} & --- & 1,418 & 2,243 & 1,775 \\
\midrule
\multicolumn{6}{l}{\textit{Mathematical Reasoning}} \\
AIME 2025 (no tools) & \textbf{95.0\%} & --- & 87.0\% & 94.0\% & 88.0\% \\
AIME 2025 (with code) & \textbf{100\%} & --- & 100\% & --- & --- \\
MathArena Apex & \textbf{23.4\%} & --- & 1.6\% & 1.0\% & 0.5\% \\
\midrule
\multicolumn{6}{l}{\textit{Scientific Knowledge}} \\
GPQA Diamond & \textbf{91.9\%} & --- & 83.4\% & 88.1\% & 86.4\% \\
MMLU & \textbf{91.8\%} & --- & 89.1\% & 91.0\% & 89.5\% \\
\midrule
\multicolumn{6}{l}{\textit{Multimodal}} \\
MMMU-Pro & \textbf{81.0\%} & --- & 68.0\% & 76.0\% & 68.0\% \\
ARC-AGI-2 & \textbf{31.1\%} & --- & 13.6\% & 17.6\% & 4.9\% \\
Video-MMMU & \textbf{87.6\%} & --- & 77.8\% & 80.4\% & 83.6\% \\
\midrule
\multicolumn{6}{l}{\textit{Long Context}} \\
MRCR v2 (128k avg) & \textbf{77.0\%} & --- & 47.1\% & 61.6\% & 58.0\% \\
MRCR v2 (1M pointwise) & \textbf{26.3\%} & --- & N/S & N/S & 16.4\% \\
\midrule
\multicolumn{6}{l}{\textit{Agentic Tasks}} \\
t2-bench (tool use) & \textbf{85.4\%} & --- & 84.7\% & 80.2\% & 54.9\% \\
Vending-Bench 2 & \textbf{\$5,478} & --- & \$3,839 & \$1,473 & \$574 \\
\bottomrule
\end{tabular}%
}
\end{table}

\textbf{Key Observations:}
\begin{itemize}
    \item \textbf{Gemini 3 Pro} leads across most benchmarks, particularly in mathematical reasoning (23.4\% MathArena Apex vs $<$2\% for others) and multimodal understanding (31.1\% ARC-AGI-2 vs 4.9-17.6\%)
    \item \textbf{Claude 4.5 Opus} achieves 74.40\% on SWE-bench Verified, competitive with other frontier models but still focused on code generation rather than debugging
    \item \textbf{Agentic coding benchmarks} (SWE-bench, Terminal-Bench) show frontier models clustering around 70-77\%, yet real-world debugging success remains below 20\% without specialized architecture
    \item \textbf{Chronos-1 fills the gap}: While frontier models excel at benchmarks, Chronos-1's 80.33\% on SWE-bench Lite (debugging-focused) demonstrates that debugging requires fundamentally different capabilities than code generation
\end{itemize}

\begin{table}[H]
\caption{SWE-bench Lite leaderboard comparison (No Filters, as of November 2025). *Chronos-1 results submitted; verification pending.}
\label{tab:swebench-lite-leaderboard}
\footnotesize
\begin{tabular*}{\columnwidth}{@{}c@{\extracolsep{\fill}}lccc@{}}
\toprule
\textbf{Rank} & \textbf{System} & \textbf{Base Model} & \textbf{Score (\%)} & \textbf{Date} \\
\midrule
\textbf{1} & \textbf{Kodezi Chronos-1} & \textbf{Proprietary} & \textbf{80.33} & \textbf{2025-11} \\
2 & ExpeRepair-v1.0 & Claude 4 Sonnet & 60.33 & 2025-06 \\
3 & Refact.ai Agent & Proprietary & 60.00 & 2025-04 \\
4 & KGCompass & Claude 4 Sonnet & 58.33 & 2025-09 \\
5 & SWE-agent & Claude 4 Sonnet & 56.67 & 2025-05 \\
6 & Isoform & Proprietary & 55.00 & 2025-01 \\
7 & SemAgent\_Multi-v1.0 & Proprietary & 51.67 & 2025-06 \\
8 & Isea & Proprietary & 51.33 & 2025-09 \\
9 & EntroPO + R2E & Qwen3-Coder-30B & 49.67 & 2025-09 \\
10 & Blackbox AI Agent & Proprietary & 49.00 & 2024-12 \\
11 & Codev & Proprietary & 49.00 & 2025-05 \\
12 & Gru & Proprietary & 48.67 & 2024-12 \\
\midrule
\multicolumn{5}{@{}l@{}}{\textit{General-Purpose Models (for reference)}} \\
\midrule
-- & Claude 4.1 Opus & -- & 14.2 & 2025 \\
-- & Claude 4.5 Sonnet & -- & 13.6 & 2025 \\
-- & GPT-4.1 & -- & 13.8 & 2025 \\
\bottomrule
\end{tabular*}
\end{table}
\begin{figure}[H]
\centering
\resizebox{0.95\columnwidth}{!}{%
\begin{tikzpicture}
    \begin{axis}[
        ybar,
        width=12cm,
        height=7cm,
        ylabel={SWE-bench Verified (\%)},
        ylabel style={font=\normalsize},
        xlabel={Model/Agent System},
        xlabel style={font=\normalsize},
        ymin=0, ymax=90,
        xtick={0,1,2,3,4,5,6,7,8},
        xticklabels={TRAE+Doubao, {live-SWE\\+Gemini 3}, {Claude 4.5\\Opus}, {Gemini 3\\Pro}, {Claude 4.5\\Sonnet}, {Claude 4\\Opus}, GPT-5, {Gemini 2.5\\Pro}, GPT-4.1},
        xticklabel style={rotate=45, anchor=east, font=\scriptsize, align=center},
        yticklabel style={font=\footnotesize},
        bar width=14pt,
        enlarge x limits=0.1,
        ymajorgrids=true,
        grid style=dashed,
        nodes near coords,
        nodes near coords align={vertical},
        nodes near coords style={font=\tiny},
        legend style={
            at={(0.5,1.05)},
            anchor=south,
            font=\scriptsize,
            legend columns=3,
            /tikz/every even column/.append style={column sep=5pt}
        },
    ]
    \addplot[fill=blue!70, draw=blue!90] coordinates {
        (0,78.80)
        (1,77.40)
    };
    \addplot[fill=red!60, draw=red!80] coordinates {
        (2,74.40)
        (3,74.20)
    };
    \addplot[fill=purple!50, draw=purple!70] coordinates {
        (4,70.60)
        (5,67.60)
        (6,65.00)
        (7,53.60)
        (8,39.58)
    };
    \legend{Agentic Systems, Latest Frontier (Nov 2025), Previous Generation}
    \end{axis}
\end{tikzpicture}%
}
\caption{SWE-bench Verified leaderboard (November 2025). Agentic systems (TRAE+Doubao: 78.80\%, live-SWE-agent+Gemini 3: 77.40\%) lead the benchmark. Latest frontier models (Claude 4.5 Opus: 74.40\%, Gemini 3 Pro: 74.20\%) show significant improvements over previous generations. Note: These scores represent code generation capabilities; debugging-focused evaluation on SWE-bench Lite reveals much lower success rates without specialized architecture.}
\label{fig:swebench-verified-leaderboard}
\end{figure}

\begin{figure}[H]
\centering
\resizebox{0.95\columnwidth}{!}{%
\begin{tikzpicture}
    \begin{axis}[
        ybar,
        width=12cm,
        height=7cm,
        ylabel={SWE-bench Lite Performance (\%)},
        ylabel style={font=\normalsize},
        ymin=0, ymax=95,
        xmin=0.3, xmax=5.7,
        xtick={1,2,3,4,5},
        xticklabels={General LLM, {+ Framework}, {+ Memory}, {+ Debug\\Training}, {Chronos-1\\(All)}},
        xticklabel style={font=\footnotesize, align=center},
        yticklabel style={font=\footnotesize},
        bar width=22pt,
        ymajorgrids=true,
        grid style=dashed,
        legend style={at={(0.5,-0.25)}, anchor=north, legend columns=2, font=\scriptsize},
        nodes near coords,
        every node near coord/.append style={font=\footnotesize\bfseries},
        clip=false,
    ]
    \addplot[fill=red!70, draw=red!90] coordinates {(1.9,13.7)};
    \addplot[fill=orange!60, draw=orange!80] coordinates {(2.4,42.5)};
    \addplot[fill=yellow!70, draw=yellow!90] coordinates {(3,58.3)};
    \addplot[fill=blue!60, draw=blue!80] coordinates {(3.6,67.8)};
    \addplot[fill=green!70, draw=green!90] coordinates {(4.1,80.33)};
    \legend{Baseline, +Agent Framework, +Session Memory, +Debug-Specific Training, +All Components}
    \end{axis}
\end{tikzpicture}%
}
\caption{Architectural impact on SWE-bench Lite performance. Starting from general-purpose models (13.7\% average), each architectural component provides incremental gains: agent frameworks add 28.8 points, session memory adds 15.8 points, and debug-specific training adds 9.5 points. Chronos-1 achieved 80.33\% using a rule-compliant configuration with AGR (repository-only), Debug LLM, and PDM (generic patterns only)---with the orchestration loop, patch refinement, and SWE-bench-derived patterns disabled per submission requirements.}
\label{fig:architecture-impact}
\end{figure}

\textbf{Performance Breakdown Analysis:} Figure~\ref{fig:architecture-impact} demonstrates the cumulative value of architectural components. While general-purpose models (Claude 4.5 Sonnet, GPT-4.1) achieve only 13.7\% average, adding agent frameworks (ExpeRepair, SWE-agent) boosts performance to 42.5\% (+28.8 points). Session memory provides another 15.8 points to 58.3\%, and debug-specific training yields 67.8\%. Chronos-1's rule-compliant configuration---using AGR (repository-only retrieval), Debug LLM, and PDM (generic patterns only), with orchestration loops and patch refinement disabled---delivers the final 12.5 points to 80.33\%, establishing a new state-of-the-art.

\begin{table}[H]
\centering
\caption{Consolidated debugging benchmark performance.}
\label{tab:debug-benchmarks}
\tiny
\setlength{\tabcolsep}{2pt}
\begin{tabular}{llccccc}
\toprule
\textbf{Benchmark} & \textbf{Task Category} & \textbf{Chronos-1} & \textbf{Amazon Q} & \textbf{ACR} & \textbf{Claude 4.1 Opus} & \textbf{GPT-4.1} \\
\midrule
\multirow{5}{*}{\rotatebox{90}{\textbf{DebugBench}}}
& Syntax errors & \textbf{71.2} & 43.8 & 31.4 & 18.2 & 15.7 \\
& Logic bugs & \textbf{68.4} & 41.2 & 28.7 & 14.3 & 12.8 \\
& Runtime errors & \textbf{74.8} & 48.3 & 35.2 & 21.1 & 19.4 \\
& Semantic bugs & \textbf{69.7} & 45.1 & 32.8 & 17.6 & 16.2 \\
& \textit{Average} & \textit{71.0} & \textit{44.6} & \textit{32.0} & \textit{17.8} & \textit{16.0} \\
\midrule
\multirow{5}{*}{\rotatebox{90}{\textbf{MdEval}}} 
& Concurrency & \textbf{58.7} & 32.1 & 22.8 & 7.3 & 6.1 \\
& Memory mgmt & \textbf{63.2} & 36.4 & 25.7 & 10.2 & 8.8 \\
& Network bugs & \textbf{66.9} & 40.8 & 28.3 & 12.7 & 11.4 \\
& Build config & \textbf{70.5} & 42.7 & 30.1 & 16.8 & 14.9 \\
& \textit{Average} & \textit{64.8} & \textit{38.0} & \textit{26.7} & \textit{11.8} & \textit{10.3} \\
\midrule
\textbf{Overall Average} & & \textbf{67.5}*** & 41.7 & 29.0 & 14.1 & 12.4 \\
\bottomrule
\multicolumn{7}{l}{\tiny ACR = AutoCodeRover. ***p $< 0.001$ compared to best baseline (Amazon Q), two-tailed t-test, n=12,500}
\end{tabular}
\end{table}

\subsection{Multi-Code Reasoning Evaluation Protocol}

Unlike traditional benchmarks that target token-level prediction in narrow context, our protocol explicitly:
\begin{itemize}
    \item Randomizes the placement of relevant context (bug source, documentation clue, test assertion) across large codebases and histories.
    \item Requires Chronos-1 to retrieve, associate, and utilize multi-code context in a compositional manner, solving tasks that demand reasoning over both explicit code relationships (e.g., function calls, imports) and implicit bug/error propagation patterns.
    \item Measures both retrieval accuracy (whether Chronos-1 finds all necessary context) and end-to-end task success (whether it can autonomously fix, validate, and document the issue).
\end{itemize}

\subsection{Multi Random Retrieval Benchmark}

We introduce the \textbf{Multi Random Retrieval (MRR)} benchmark, specifically designed to evaluate debugging-oriented retrieval capabilities. The full evaluation suite is scheduled for release in Q1 2026.

\subsubsection{MRR Design and Methodology}

The Multi Random Retrieval benchmark addresses a critical gap in existing evaluations: real-world debugging requires finding scattered information across large codebases where traditional similarity-based retrieval fails.

\textbf{Key Design Principles:}
\begin{enumerate}
    \item \textbf{Random Distribution}: Relevant debugging clues are randomly scattered across 10-50 files, simulating real-world information dispersion
    \item \textbf{High Noise Ratio}: 70\% of retrievable content is plausible but irrelevant, testing precision
    \item \textbf{Multi-Hop Requirements}: Average 3-7 retrieval steps needed to gather complete context
    \item \textbf{Temporal Scattering}: Information spans multiple commits/time periods
\end{enumerate}

\textbf{Formal Metrics:}
\begin{itemize}
    \item \textbf{Precision@k}: Fraction of retrieved chunks that are actually relevant to the bug fix
    \item \textbf{Recall@k}: Fraction of all relevant chunks successfully retrieved
    \item \textbf{MRR Score}: Mean reciprocal rank of first correct retrieval, computed as $\text{MRR} = \frac{1}{|Q|}\sum_{i=1}^{|Q|}\frac{1}{\text{rank}_i}$
    \item \textbf{Fix Success Rate}: Whether the generated fix actually resolves the bug
\end{itemize}

\begin{figure}[H]
\centering
\resizebox{\columnwidth}{!}{%
\begin{tikzpicture}[
    node distance=1.5cm,
    file/.style={rectangle, draw, minimum width=2cm, minimum height=0.6cm, font=\footnotesize},
    bug/.style={rectangle, draw, fill=red!20, minimum width=2.5cm, minimum height=0.8cm, font=\footnotesize},
    correct/.style={rectangle, draw, fill=green!20, minimum width=2cm, minimum height=0.6cm, font=\footnotesize},
    wrong/.style={rectangle, draw, fill=gray!20, minimum width=2cm, minimum height=0.6cm, font=\footnotesize},
    path/.style={->, thick, font=\scriptsize}
]
    \node[bug] (bug) at (0,0) {Bug: Async Task Failure};
    
    \node[wrong] (gpt1) at (-4,-1.5) {TaskRunner.java};
    \node[wrong] (gpt2) at (-4,-2.5) {AsyncUtils.java};
    \node[wrong] (gpt3) at (-4,-3.5) {ErrorHandler.java};
    \node[bug, fill=red!40] (gptfix) at (-4,-5) {Wrong Fix: Add try-catch};
    
    \draw[path, red] (bug) -- node[left] {Text similarity} (gpt1);
    \draw[path, red] (gpt1) -- (gpt2);
    \draw[path, red] (gpt2) -- (gpt3);
    \draw[path, red] (gpt3) -- (gptfix);
    
    \node[correct] (chr1) at (4,-1.5) {TaskExecutor.java:89};
    \node[correct] (chr2) at (4,-2.5) {ThreadPoolConfig.yml};
    \node[correct] (chr3) at (4,-3.5) {commit: "reduce threads"};
    \node[bug, fill=green!40] (chrfix) at (4,-5) {Correct Fix: Restore pool size};
    
    \draw[path, green!70!black] (bug) -- node[right] {Stack trace} (chr1);
    \draw[path, green!70!black] (chr1) -- node[right] {Config ref} (chr2);
    \draw[path, green!70!black] (chr2) -- node[right] {Git blame} (chr3);
    \draw[path, green!70!black] (chr3) -- (chrfix);
    
    \node[font=\footnotesize\bfseries] at (-4,-6) {GPT-4: Surface Pattern Matching};
    \node[font=\footnotesize\bfseries] at (4,-6) {Chronos-1: Causal Path Following};
\end{tikzpicture}%
}
\caption{MRR example showing retrieval paths: GPT-4 follows textual similarity to related but incorrect files, while Chronos-1 traces causal dependencies through stack trace → configuration → git history to find the true root cause.}
\label{fig:mrr-example}
\end{figure}

\subsubsection{Reproducibility and Open Science}

To ensure reproducibility and enable fair comparisons, we provide:

\begin{enumerate}
\item \textbf{Open Evaluation Subset:}
\begin{itemize}
    \item 500 debugging scenarios (10\% of full benchmark) with complete test harnesses
    \item Ground truth fixes and intermediate retrieval annotations
    \item Automated evaluation scripts with standardized metrics
    \item Docker containers with exact environment configurations
\end{itemize}

\item \textbf{Evaluation Infrastructure:}
\begin{itemize}
    \item \texttt{chronos-eval}: Python package for running benchmarks
    \item Supports custom model integration via simple API
    \item Automated result validation and statistical significance testing
    \item Leaderboard submission system with blind test set
\end{itemize}

\item \textbf{Detailed Task Definitions:}
\end{enumerate}

\textbf{Example MRR Task Format:}
\begin{lstlisting}[language=Java, basicstyle=\tiny\ttfamily, breaklines=true, xleftmargin=5pt, xrightmargin=3pt, linewidth=0.92\columnwidth, aboveskip=0.5em, belowskip=0.5em, keywordstyle=\color{blue}\bfseries, commentstyle=\color{gray}, stringstyle=\color{green!60!black}]
{
  "bug_id": "apache-commons-math-1234",
  "repo_snapshot": "git://github.com/apache/commons-math@abc123",
  "failing_tests": ["TestLinearRegression.testFit"],
  "scattered_files": [
    "src/.../stat/regression/AbstractRegression.java",
    "src/.../linear/MatrixUtils.java"
    // ... 15 more files across different modules
  ],
  "temporal_range": "2023-01-15 to 2023-04-22",
  "evaluation": {
    "fix_validator": "docker run chronos-eval:fix-validator",
    "context_scorer": "weighted_jaccard",
    "regression_tests": ["TestSuite.class"]
  }
}
\end{lstlisting}

\vspace{0.5em}
\subsubsection{Benchmark Design}
The MRR benchmark consists of 5,000 real-world debugging scenarios where:
\begin{enumerate}
    \item \textbf{Context Scattering}: Relevant debugging information is randomly distributed across 10-50 files
    \item \textbf{Temporal Dispersion}: Critical bug context spans 3-12 months of commit history
    \item \textbf{Obfuscated Dependencies}: Variable names and function calls are refactored between bug introduction and discovery
    \item \textbf{Multi-Modal Artifacts}: Solutions require combining code, tests, logs, and documentation
\end{enumerate}

\subsubsection{Evaluation Metrics}
\begin{itemize}
    \item \textbf{Retrieval Precision@k}: Fraction of retrieved artifacts that are relevant to the bug fix
    \item \textbf{Retrieval Recall@k}: Fraction of all relevant artifacts successfully retrieved
    \item \textbf{Fix Accuracy}: Whether the generated fix passes all tests and doesn't introduce regressions
    \item \textbf{Context Efficiency}: Ratio of used vs retrieved tokens in the final solution
\end{itemize}

\subsubsection{Results on MRR Benchmark}
\begin{table}[H]
\centering
\caption{Performance on Multi Random Retrieval benchmark, demonstrating Chronos-1's superior ability to find and utilize scattered debugging context.}
\label{tab:mrr-results}
\resizebox{\columnwidth}{!}{%
\begin{tabular}{lcccc}
\hline
\textbf{Model} & \textbf{Precision@10} & \textbf{Recall@10} & \textbf{Fix Accuracy} & \textbf{Context Eff.} \\
\hline
GPT-4.1 + RAG & 55.2\% & 42.3\% & 13.8\% & 0.34 \\
Claude 4.5 Opus + RAG & 68.3\% & 54.2\% & 17.1\% & 0.46 \\
Claude 4.1 Opus + RAG & 62.1\% & 48.7\% & 14.2\% & 0.41 \\
Gemini 3 Pro + RAG & 66.8\% & 52.9\% & 16.4\% & 0.44 \\
Gemini 2.5 Pro + RAG & 51.7\% & 40.1\% & 12.4\% & 0.38 \\
\textbf{Kodezi Chronos-1} & \textbf{89.2\%} & \textbf{84.7\%} & \textbf{67.3\%} & \textbf{0.71} \\
\hline
\end{tabular}%
}
\end{table}

\subsubsection{Comprehensive MRR Results Against State-of-the-Art Debugging Systems}

The MRR benchmark reveals the fundamental difference between general-purpose code models and debugging-specific systems. Table~\ref{tab:mrr-detailed} shows detailed performance across all evaluated systems:

\begin{table}[H]
\centering
\caption{Comprehensive MRR benchmark results comparing Chronos-1 against modern debugging systems and general-purpose code models. Results demonstrate the critical importance of debugging-specific design. All debugging-focused systems listed are publicly available or open source ($^\dagger$indicates open source projects).}
\label{tab:mrr-detailed}
\resizebox{\columnwidth}{!}{%
\begin{tabular}{@{}lcccccc@{}}
\hline
\textbf{System} & \textbf{Type} & \textbf{Precision@10} & \textbf{Recall@10} & \textbf{Fix Acc.} & \textbf{Iterations} & \textbf{Success Time} \\
\hline
\multicolumn{7}{c}{\textit{General-Purpose Code Models (2025)}} \\
\hline
Claude 4.5 Opus & Code Gen & 68.3\% & 54.2\% & 17.1\% & 2.5 & 16.4 min \\
Gemini 3 Pro & Code Gen & 66.8\% & 52.9\% & 16.4\% & 2.4 & 15.8 min \\
Claude 4.1 Opus & Code Gen & 62.1\% & 48.7\% & 14.2\% & 2.3 & 15.2 min \\
Claude 4.5 Sonnet & Code Gen & 59.8\% & 46.3\% & 13.1\% & 2.1 & 14.8 min \\
GPT-4.1 & Code Gen & 55.2\% & 42.3\% & 13.8\% & 1.8 & 12.3 min \\
GPT-4o & Code Gen & 48.3\% & 37.2\% & 10.2\% & 1.6 & 11.7 min \\
Gemini 2.5 Pro & Code Gen & 51.7\% & 40.1\% & 12.4\% & 2.0 & 13.5 min \\
DeepSeek V3 & Code Gen & 44.2\% & 34.8\% & 9.7\% & 1.4 & 10.2 min \\
Qwen2.5-Coder-32B & Code Gen & 41.3\% & 31.2\% & 8.3\% & 1.3 & 9.8 min \\
\hline
\multicolumn{7}{c}{\textit{IDE-Integrated Systems}} \\
\hline
Cursor Agent Mode & IDE & 32.4\% & 24.1\% & 4.2\% & 1.0 & 8.5 min \\
Windsurf Cascade & IDE & 35.7\% & 27.3\% & 5.1\% & 1.2 & 9.1 min \\
Claude Code CLI & CLI & 40.2\% & 31.8\% & 6.8\% & 1.4 & 10.3 min \\
Gemini CLI (1.5 Pro) & CLI & 43.1\% & 35.2\% & 9.7\% & 1.5 & 11.2 min \\
\hline
\multicolumn{7}{c}{\textit{Debugging-Focused Systems (Publicly Available / Open Source)}} \\
\hline
AutoCodeRover$^\dagger$ & Debug & 71.3\% & 63.2\% & 30.7\% & 4.2 & 28.3 min \\
Amazon Q Developer & Debug & 75.8\% & 68.4\% & 49.0\% & 5.1 & 35.7 min \\
SWE-Agent$^\dagger$ (GPT-4) & Debug & 65.2\% & 57.1\% & 22.3\% & 3.8 & 24.2 min \\
OpenHands$^\dagger$ (CodeAct) & Debug & 67.8\% & 54.3\% & 19.2\% & 3.4 & 23.8 min \\
Agentless$^\dagger$ & Debug & 62.4\% & 48.7\% & 16.8\% & 2.1 & 18.4 min \\
Moatless Tools$^\dagger$ & Debug & 69.1\% & 58.2\% & 24.1\% & 4.5 & 29.7 min \\
\textbf{Kodezi Chronos-1} & \textbf{Debug} & \textbf{89.2\%} & \textbf{84.7\%} & \textbf{67.3\%} & \textbf{7.8} & \textbf{42.3 min} \\
\hline
\multicolumn{7}{l}{\footnotesize $^\dagger$Open source: AutoCodeRover (NUS-APR), SWE-Agent (Princeton NLP), OpenHands, Agentless, Moatless Tools}
\end{tabular}%
}
\end{table}

Key insights from the comprehensive MRR evaluation:

\begin{enumerate}
    \item \textbf{Debugging vs Code Generation}: General-purpose models achieve less than 17\% fix accuracy despite having state-of-the-art code generation capabilities. This validates our hypothesis that debugging requires fundamentally different architectures.

    \item \textbf{Open Source Agent Limitations}: Leading open source agents like AutoCodeRover (30.7\%), Moatless Tools (24.1\%), and SWE-Agent (22.3\%) show significant improvements over general-purpose models but still fall short of Chronos-1's 67.3\% fix accuracy. These tools lack persistent memory and adaptive retrieval mechanisms.

    \item \textbf{Iteration Depth}: Chronos-1 performs 7.8 iterations on average compared to 2-4 for open source agents and 1-2 for general models, demonstrating the importance of persistent debugging loops with execution feedback.

    \item \textbf{Precision-Recall Trade-off}: While Amazon Q Developer shows impressive 49\% fix accuracy, Chronos-1's 89.2\% precision and 84.7\% recall demonstrate superior context identification, leading to 67.3\% fix accuracy.

    \item \textbf{Time Investment}: Chronos-1 takes longer (42.3 min average) but achieves 2-3x higher success rates than open source alternatives, making it more time-efficient overall when considering rework costs.
\end{enumerate}

\subsection{Adaptive Graph-Guided Retrieval Performance}

We evaluate the impact of AGR on debugging accuracy across different retrieval depths:

\begin{table}[H]
\centering
\caption{Performance metrics for different retrieval strategies. Adaptive AGR dynamically selects optimal k based on query complexity.}
\label{tab:agr-performance}
\resizebox{\columnwidth}{!}{%
\begin{tabular}{lccccc}
\hline
\textbf{Retrieval Method} & \textbf{k=1} & \textbf{k=2} & \textbf{k=3} & \textbf{k=adaptive} & \textbf{Flat} \\
\hline
Precision & 84.3±2.1\% & 91.2±1.4\% & 88.7±1.8\% & \textbf{92.8±1.2\%} & 71.4±3.2\% \\
Recall & 72.1±2.8\% & 86.4±1.9\% & 89.2±1.6\% & \textbf{90.3±1.5\%} & 68.2±3.5\% \\
F1 Score & 77.7±2.4\% & 88.7±1.6\% & 88.9±1.7\% & \textbf{91.5±1.3\%} & 69.8±3.3\% \\
Debug Success & 58.2±3.1\% & 72.4±2.3\% & 71.8±2.4\% & \textbf{87.1±1.8\%} & 23.4±4.1\% \\
\hline
\end{tabular}%
}
\end{table}

Key findings from AGR evaluation:
\begin{itemize}
    \item \textbf{Optimal Depth Varies}: Simple bugs require k=1-2, while complex cross-module issues benefit from k=3+
    \item \textbf{Adaptive Superiority}: Dynamic depth selection outperforms fixed k values by 15-20\%
    \item \textbf{5x Improvement}: AGR achieves 87.1\% debug success vs 23.4\% for flat retrieval
    \item \textbf{Hardware Debugging}: Particularly effective for Verilog/VHDL with 91\% accuracy (vs 18\% baseline)
\end{itemize}

\begin{figure}[H]
\centering
\resizebox{\columnwidth}{!}{%
\begin{tikzpicture}
\begin{axis}[
    width=10cm,
    height=6.5cm,
    xlabel={Number of Retrieved Nodes},
    ylabel={Precision/Recall (\%)},
    xmin=0, xmax=22,
    ymin=40, ymax=102,
    xtick={1,3,5,10,20},
    ytick={40,50,60,70,80,90,100},
    legend style={at={(1.02,1)}, anchor=north west, font=\scriptsize, draw=black, fill=white},
    grid=major,
    grid style={dashed,gray!30},
    xlabel style={font=\normalsize},
    ylabel style={font=\normalsize},
    tick label style={font=\small},
    smooth,
    mark size=2.5pt,
    line width=1.2pt,
    legend cell align=left,
    legend columns=1
]

\addplot[color=blue,mark=*] coordinates {
    (1,91.2) (3,94.8) (5,92.3) (10,87.6) (20,82.1)
};
\addlegendentry{AGR Precision}

\addplot[color=blue,dashed,mark=square*] coordinates {
    (1,62.3) (3,78.9) (5,84.2) (10,89.7) (20,91.3)
};
\addlegendentry{AGR Recall}

\addplot[color=red,mark=triangle*] coordinates {
    (1,88.7) (3,81.2) (5,73.8) (10,64.2) (20,51.3)
};
\addlegendentry{Top-k Precision}

\addplot[color=red,dashed,mark=square] coordinates {
    (1,42.1) (3,58.7) (5,71.2) (10,83.4) (20,92.8)
};
\addlegendentry{Top-k Recall}

\addplot[color=green!70!black,mark=diamond*] coordinates {
    (1,100) (3,100) (5,100) (10,100) (20,100)
};
\addlegendentry{Oracle Precision}

\addplot[color=green!70!black,dashed,mark=diamond] coordinates {
    (1,68.4) (3,82.7) (5,89.1) (10,94.8) (20,98.2)
};
\addlegendentry{Oracle Recall}

\end{axis}
\end{tikzpicture}%
}
\caption{Precision-recall trade-off for different retrieval strategies. AGR maintains high precision at low k values while achieving comparable recall to oracle retriever, demonstrating efficient noise avoidance.}
\label{fig:precision-recall}
\end{figure}

\subsubsection{Why AGR Dominates Existing Graph-Based Retrieval Systems}

While graph-based code retrieval exists in various forms, AGR's architecture and debugging-specific optimizations make it vastly superior to all existing approaches:

\begin{table}[H]
\centering
\caption{AGR vs. existing graph-based retrieval systems on MRR benchmark (5,000 debugging scenarios). All compared systems are publicly available or open source.}
\label{tab:agr-dominance}
\resizebox{\columnwidth}{!}{%
\begin{tabular}{lccccccc}
\hline
\textbf{System} & \textbf{Graph Types} & \textbf{Debug Success} & \textbf{Precision@k=3} & \textbf{Avg Tokens} & \textbf{Cross-File} & \textbf{Memory} & \textbf{Adaptive} \\
& & \textbf{(MRR)} & \textbf{(MRR)} & \textbf{Retrieved} & \textbf{(MRR)} & \textbf{Persist} & \textbf{Depth} \\
\hline
\textbf{AGR (Chronos-1)} & AST+Log+Test+PR & \textbf{87.1\%} & \textbf{92.8\%} & \textbf{31.2K} & \textbf{71.2\%} & \textbf{\checkmark} & \textbf{\checkmark} \\
\hline
Graph RAG$^\dagger$ & Static KG & 28.0\% & 73.2\% & 76K & 33\% & $\times$ & $\times$ \\
LlamaIndex$^\dagger$ (KG) & Property Graph & 31.2\% & 74.8\% & 68K & 29\% & $\times$ & $\times$ \\
CodeQL$^\dagger$ & AST+Dataflow & 38.5\% & 79.3\% & 45K & 35\% & $\times$ & $\times$ \\
Sourcegraph & Dependency+Symbol & 26.7\% & 69.4\% & 82K & 31\% & $\times$ & $\times$ \\
Aider$^\dagger$ (RepoMap) & AST+Import & 33.8\% & 76.1\% & 54K & 27\% & $\times$ & $\times$ \\
tree-sitter$^\dagger$ + RAG & AST only & 22.4\% & 65.3\% & 71K & 19\% & $\times$ & $\times$ \\
\hline
\multicolumn{8}{l}{\footnotesize $^\dagger$Open source projects.}
\end{tabular}%
}
\end{table}

{\footnotesize All results on Multi Random Retrieval (MRR) benchmark with 5,000 real debugging scenarios. Debug Success = \% of bugs correctly fixed and validated. Precision@k=3 = \% of retrieved nodes relevant to bug fix. Cross-File = \% success on bugs spanning multiple files.}

\textbf{Key Architectural Advantages of AGR:}

\textbf{1. Unified Multi-Signal Graph Construction}

AGR fuses eight distinct signal types into one traversable graph:
\begin{itemize}
    \item AST parent/child relationships (weight: 0.8)
    \item Import/dependency edges (weight: 0.9)
    \item Function call graphs (weight: 0.85)
    \item Test coverage mappings (weight: 0.95)
    \item Log emission points (weight: 0.92)
    \item Stack trace paths (weight: 0.97)
    \item Commit co-occurrence (weight: $0.7 \cdot e^{-\lambda t}$)
    \item PR review comments (weight: 0.6)
\end{itemize}

No existing system combines more than 2-3 of these signals. This multi-modal fusion enables AGR to trace paths like: \texttt{test\_failure $\rightarrow$ log\_entry $\rightarrow$ emitting\_function $\rightarrow$ upstream\_caller $\rightarrow$ config\_error $\rightarrow$ fix\_location}.

\textbf{2. Confidence-Guided Adaptive Expansion}

While other systems use fixed-depth traversal, AGR implements entropy-based adaptive expansion:

\begin{lstlisting}[language=Python, basicstyle=\tiny\ttfamily, breaklines=true, xleftmargin=5pt, xrightmargin=3pt, linewidth=0.92\columnwidth, aboveskip=0.5em, belowskip=0.5em, keywordstyle=\color{blue}\bfseries, commentstyle=\color{gray}, stringstyle=\color{green!60!black}]
def adaptive_expand(seed_nodes, threshold=0.89):
  """Expand graph traversal adaptively based on confidence."""
  k = 1
  current_nodes = seed_nodes
  best_nodes_so_far = seed_nodes
  confidence = 0
  
  while confidence < threshold and k < max_hops:
    neighbors = graph.expand(current_nodes, k)
    confidence = calculate_entropy_confidence(neighbors)
    
    if confidence > threshold:
      return neighbors
    
    best_nodes_so_far = neighbors
    current_nodes = neighbors
    k += 1
  
  return best_nodes_so_far
\end{lstlisting}

This prevents both under-exploration (missing crucial context) and over-exploration (noise flooding). Average expansion: 3.7 hops for complex bugs, 1.2 for simple ones.

\textbf{3. Temporal-Aware Edge Weighting}

AGR weights edges based on temporal correlation with debugging success:
\begin{itemize}
    \item Recent bug fixes: Higher weight ($e^{-0.1t}$ decay)
    \item Co-committed files: Boosted during active development
    \item Test-code proximity: Updated after each test run
    \item Historical failure patterns: Accumulated over debugging sessions
\end{itemize}

Static systems like IntelliCode or Graph RAG cannot adapt weights based on debugging outcomes.

\textbf{4. Persistent Debug Memory Integration}

AGR uniquely interfaces with PDM to:
\begin{itemize}
    \item Cache successful traversal paths (reused in 47ms vs 3.2min initial)
    \item Learn failure-prone subgraphs (87\% match rate on recurring bugs)
    \item Accumulate fix patterns (6.8x faster on similar bugs)
\end{itemize}

Other systems are stateless, starting from scratch each session.

\textbf{5. Debugging-Specific Training and Optimization}

Unlike general-purpose retrievers, AGR is trained on:
\begin{itemize}
    \item 15M real debugging sessions with ground-truth fix locations
    \item Stack trace $\rightarrow$ root cause traversal paths
    \item Test failure $\rightarrow$ code change mappings
    \item Multi-file bug resolution sequences
\end{itemize}

This specialization explains the 3-5x performance gap over adapted general retrievers.

\begin{tcolorbox}[
    enhanced,
    title={\textbf{AGR Performance Domination Summary}},
    fonttitle=\bfseries,
    coltitle=white,
    colbacktitle=blue!75!black,
    colback=blue!5,
    colframe=blue!75!black,
    boxrule=0.5pt
]
\begin{itemize}
    \item \textbf{4x Higher Debug Success}: 87.1\% vs 23-41\% for existing systems
    \item \textbf{30\% Better Precision}: 92.8\% vs best alternative (81\%)
    \item \textbf{65\% Fewer Tokens}: 31.2K vs 51-89K average
    \item \textbf{3x Better Cross-File}: 71.2\% vs 31\% best competitor
    \item \textbf{Only System with Memory}: Enables continuous improvement
\end{itemize}
\end{tcolorbox}

\subsubsection{Oracle Retriever Experiments: Upper Bound Analysis}

To understand the theoretical limits of retrieval-augmented debugging, we conducted experiments with an oracle retriever that has perfect knowledge of which code segments are relevant:

\begin{table}[H]
\centering
\caption{Performance comparison with oracle retriever showing upper bounds for retrieval-based debugging}
\label{tab:oracle-retriever}
\resizebox{\columnwidth}{!}{%
\begin{tabular}{lccccc}
\hline
\textbf{Retriever Type} & \textbf{Debug Success} & \textbf{Precision} & \textbf{Recall} & \textbf{Avg Context} & \textbf{Fix Quality} \\
\hline
Random Baseline & 8.7\% & 31.2\% & 28.4\% & 15.2K tokens & 2.1/5 \\
BM25 & 18.3\% & 52.8\% & 49.7\% & 18.7K tokens & 2.8/5 \\
Dense Retrieval & 24.6\% & 64.3\% & 61.2\% & 21.3K tokens & 3.2/5 \\
HyDE & 31.2\% & 71.8\% & 68.9\% & 19.8K tokens & 3.5/5 \\
Self-RAG & 28.9\% & 69.4\% & 66.2\% & 24.1K tokens & 3.4/5 \\
Graph RAG & 33.7\% & 73.2\% & 70.8\% & 26.5K tokens & 3.6/5 \\
\hline
AGR (Chronos-1) & 65.3\% & 89.2\% & 84.7\% & 31.2K tokens & 4.3/5 \\
Oracle Retriever & 78.9\% & 100\% & 100\% & 42.7K tokens & 4.7/5 \\
\hline
\textbf{Gap to Oracle} & \textbf{13.6\%} & \textbf{10.8\%} & \textbf{15.3\%} & - & \textbf{0.4} \\
\hline
\end{tabular}%
}
\end{table}

Oracle retriever insights:
\begin{itemize}
    \item \textbf{21.1\% Ceiling}: Even with perfect retrieval, 21.1\% of bugs cannot be fixed due to reasoning limitations
    \item \textbf{Chronos-1 Efficiency}: AGR achieves 82.7\% of oracle performance while using 27\% less context
    \item \textbf{Remaining Gap}: The 13.6\% gap suggests room for retrieval improvements, particularly in:
        \begin{itemize}
            \item Cross-repository dependencies (4.2\% of failures)
            \item Implicit behavioral contracts (3.8\% of failures)
            \item Domain-specific patterns (2.9\% of failures)
        \end{itemize}
    \item \textbf{Quality vs Quantity}: Oracle retrieves more context (42.7K vs 31.2K tokens) but Chronos-1's selective retrieval maintains 91\% of oracle's fix quality
\end{itemize}

\begin{table}[H]
\centering
\resizebox{\columnwidth}{!}{%
\begin{tabular}{lccccc}
\hline
\textbf{RAG Technique} & \textbf{General Tasks} & \textbf{Code Tasks} & \textbf{Debug Tasks} & \textbf{MRR Bench} & \textbf{Compute Cost} \\
\hline
Flat Retrieval & 71.2\% & 68.3\% & 23.4\% & 31.7\% & 1.0x \\
HyDE & 82.1\% & 74.2\% & 31.2\% & 42.3\% & 2.1x \\
Self-RAG & 85.7\% & 78.9\% & 38.7\% & 48.1\% & 1.8x \\
FLARE & 83.9\% & 76.5\% & 35.2\% & 45.6\% & 2.3x \\
Graph RAG & 79.8\% & 81.2\% & 41.3\% & 51.7\% & 3.2x \\
\textbf{Chronos-1 AGR} & 88.3\% & \textbf{89.7\%} & \textbf{87.1\%} & \textbf{89.2\%} & 2.7x \\
\hline
\end{tabular}%
}
\caption{Comparison of RAG techniques across different task types. Advanced RAG methods show improvements over flat retrieval, with Chronos-1's AGR demonstrating the highest performance on debugging tasks.}
\label{tab:rag-comparison}
\end{table}

\begin{table}[H]
\centering
\resizebox{\columnwidth}{!}{%
\begin{tabular}{lccccc}
\hline
\textbf{Model} & \textbf{HumanEval} & \textbf{MBPP} & \textbf{Debug Success} & \textbf{Root Cause Acc.} & \textbf{Retrieval Prec.} \\
\hline
Claude 4.5 Opus & 94.1±0.6\%   & 92.8±0.5\%     & 17.1±1.5\%  & 22.4±1.3\% & 82±1.5\% \\
Gemini 3 Pro & 93.7±0.6\%   & 92.4±0.5\%     & 16.4±1.5\%  & 21.6±1.3\% & 80±1.6\% \\
GPT-4o       & 87.8±1.0\%   & 88.5±0.8\%     & 10.3±1.9\%  & 14.7±1.6\% & 72±2.0\% \\
GPT-4.1      & 91.2±0.8\%   & 90.7±0.7\%     & 13.8±1.7\%  & 18.2±1.5\% & 76±1.8\% \\
Claude 4.1 Opus & 92.8±0.7\%   & 91.3±0.6\%     & 14.2±1.6\%  & 19.1±1.4\% & 78±1.7\% \\
Claude 4.5 Sonnet & 92.1±0.8\%   & 90.9±0.7\%     & 13.6±1.7\%  & 18.5±1.5\% & 77±1.8\% \\
DeepSeek V3  & 90.5±0.9\%   & 89.8±0.8\%     & 12.1±1.8\%  & 16.3±1.6\% & 75±1.9\% \\
Qwen2.5-32B  & 89.7±1.0\%   & 88.9±0.9\%     & 11.5±1.9\%  & 15.7±1.7\% & 73±2.0\% \\
Gemini 2.5 Pro & 91.6±0.8\%   & 90.2±0.7\%     & 13.9±1.7\%  & 17.9±1.5\% & 76±1.8\% \\
\textbf{Chronos-1} & \textbf{90.2±0.6\%}NS & \textbf{88.9±0.5\%}NS & \textbf{65.3±1.4\%}*** & \textbf{78.4±1.2\%}*** & \textbf{91±0.8\%}*** \\
\hline
\multicolumn{6}{l}{\footnotesize NS: Not significant, *p $< 0.05$, **p $< 0.01$, ***p $< 0.001$ compared to best baseline (two-tailed t-test)}
\end{tabular}%
}
\caption{Performance across code synthesis and debugging tasks (mean ± std over 5 runs). Note that while Chronos-1 shows average performance on pure code generation tasks (HumanEval, MBPP), it dramatically outperforms all models on debugging-specific metrics.}
\label{tab:benchmark-results}
\end{table}

\subsection{Comprehensive Debugging Performance Analysis}

To provide a complete picture of debugging capabilities, Table~\ref{tab:comprehensive-debug} presents detailed metrics across all evaluation dimensions:

\begin{table}[H]
\centering
\caption{Comprehensive debugging performance comparison across all key metrics}
\label{tab:comprehensive-debug}
\resizebox{\columnwidth}{!}{%
\begin{tabular}{lccccccc}
\hline
\textbf{Model/System} & \textbf{Root Cause} & \textbf{Fix Valid} & \textbf{Cross-File} & \textbf{MRR} & \textbf{Regression} & \textbf{Debug} & \textbf{Time to} \\
 & \textbf{Precision (\%)} & \textbf{Accuracy (\%)} & \textbf{Hit Rate (\%)} & \textbf{Score} & \textbf{Avoid (\%)} & \textbf{Cycles} & \textbf{Fix (min)} \\
\hline
\multicolumn{8}{l}{\textit{General-Purpose LLMs}} \\
Claude 4.5 Opus & 38.2±2.4 & 17.1±1.4 & 52.3±3.5 & 0.36 & 73.4±3.8 & 6.8±1.7 & 39.2±10.4 \\
Gemini 3 Pro & 37.1±2.3 & 16.4±1.3 & 50.8±3.4 & 0.34 & 72.1±3.9 & 7.1±1.8 & 41.3±10.9 \\
GPT-4.1 & 31.2±2.1 & 13.8±1.2 & 42.3±3.1 & 0.28 & 67.2±4.2 & 8.3±2.1 & 47.2±12.3 \\
Claude 4.1 Opus & 34.7±2.3 & 14.2±1.3 & 45.8±3.3 & 0.31 & 69.8±4.0 & 7.9±1.9 & 44.6±11.8 \\
Gemini 2.5 Pro & 29.8±2.0 & 13.9±1.2 & 39.2±2.9 & 0.26 & 65.3±4.4 & 9.1±2.3 & 51.3±13.7 \\
DeepSeek V3 & 27.3±1.9 & 12.1±1.1 & 37.6±2.8 & 0.24 & 63.7±4.5 & 9.7±2.5 & 54.8±14.2 \\
\hline
\multicolumn{8}{l}{\textit{Enhanced RAG Systems}} \\
GPT-4.1 + HyDE & 38.9±2.5 & 22.3±1.8 & 51.2±3.5 & 0.36 & 74.2±3.8 & 6.2±1.5 & 38.7±9.4 \\
Claude 4.1 Opus + Self-RAG & 41.2±2.6 & 24.7±1.9 & 53.7±3.6 & 0.39 & 76.1±3.6 & 5.8±1.4 & 36.2±8.9 \\
Gemini + GraphRAG & 43.8±2.7 & 28.9±2.1 & 58.3±3.7 & 0.42 & 78.3±3.4 & 5.3±1.3 & 33.4±8.1 \\
\hline
\multicolumn{8}{l}{\textit{Specialized Debugging Tools (Open Source / Publicly Available)}} \\
SWE-agent$^\dagger$ & 47.3±2.8 & 22.9±1.7 & 61.2±3.8 & 0.44 & 81.2±3.2 & 4.7±1.1 & 31.2±7.6 \\
AutoCodeRover$^\dagger$ & 52.1±2.9 & 31.2±2.2 & 64.8±3.9 & 0.48 & 83.7±3.0 & 4.2±1.0 & 28.3±6.9 \\
OpenHands$^\dagger$ & 49.8±2.8 & 26.4±1.9 & 59.7±3.7 & 0.45 & 80.4±3.2 & 4.9±1.2 & 30.8±7.8 \\
Agentless$^\dagger$ & 45.2±2.6 & 21.7±1.6 & 54.3±3.5 & 0.40 & 78.9±3.4 & 5.3±1.3 & 33.1±8.2 \\
Moatless Tools$^\dagger$ & 51.4±2.9 & 28.3±2.0 & 62.1±3.8 & 0.47 & 82.1±3.1 & 4.4±1.0 & 29.4±7.2 \\
\hline
\textbf{Chronos-1} & \textbf{87.4±1.2} & \textbf{65.3±1.4} & \textbf{89.2±1.3} & \textbf{0.82} & \textbf{94.6±0.9} & \textbf{2.2±0.4} & \textbf{14.7±3.2} \\
\hline
\multicolumn{8}{l}{\footnotesize All metrics averaged over 5,000 debugging scenarios. $^\dagger$Open source projects. ± indicates standard error.}
\end{tabular}%
}
\end{table}

Key insights from comprehensive analysis:
\begin{itemize}
    \item \textbf{Root Cause Precision}: Chronos-1 achieves 87.4\% accuracy in identifying the true source of bugs, 2.5x better than the best baseline
    \item \textbf{Cross-File Retrieval}: 89.2\% hit rate demonstrates AGR's effectiveness at traversing dependencies
    \item \textbf{Regression Avoidance}: 94.6\% of Chronos-1 fixes don't introduce new bugs, critical for production deployment
    \item \textbf{Efficiency}: Average 2.2 debug cycles and 14.7 minutes to fix, showing rapid convergence
\end{itemize}

\subsection{Comparison with Agentic Code Tools}

While traditional LLMs struggle with debugging, a new generation of agentic code tools has emerged. We evaluate Chronos-1 against these systems on real-world debugging scenarios:

\begin{table}[H]
\centering
\caption{Comparison of Chronos-1 with modern agentic code tools (2025), sorted by debugging success rate.}
\label{tab:agentic-comparison}
\resizebox{\columnwidth}{!}{%
\begin{tabular}{llcccccc}
\toprule
\textbf{Category} & \textbf{Tool} & \textbf{Context} & \textbf{Memory} & \textbf{Debug Loop} & \textbf{Multi-File} & \textbf{CI/CD} & \textbf{Success Rate} \\
\midrule
\multirow{2}{*}{\textit{IDE-Integrated}} 
& Cursor & IDE + 32K & Session & Agent Mode & Yes & No & 4.2\% \\
& Windsurf & Cascade Tech & Session & Write Mode & Yes & No & 5.1\% \\
\midrule
\multirow{2}{*}{\textit{Code Assistants}}
& GitHub Copilot X & 16K tokens & None & No & Limited & No & 5.3\% \\
& Warp.dev & Terminal & Session & No & Limited & Yes & 7.3\%* \\
\midrule
\multirow{2}{*}{\textit{CLI Tools}}
& Claude Code CLI & 200K tokens & Session & No & Yes & No & 6.8\% \\
& Gemini CLI & 1.5M tokens & None & No & Yes & Limited & 9.7\% \\
\midrule
\textit{Debugging-First} & \textbf{Chronos-1} & \textbf{Unlimited**} & \textbf{Persistent} & \textbf{Yes} & \textbf{Yes} & \textbf{Yes} & \textbf{65.3\%} \\
\bottomrule
\end{tabular}%
}
\end{table}
{\footnotesize *Warp achieves 71\% on SWE-bench code generation with specialized setup but only 7.3\% on real-world debugging tasks (MRR benchmark). **Via AGR intelligent retrieval.}

Key differentiators of Chronos-1:
\begin{itemize}
    \item \textbf{Persistent Memory}: Unlike session-based tools, Chronos-1maintains cross-session knowledge of bugs, fixes, and patterns
    \item \textbf{True Debugging Loop}: Automated iteration through fix-test-refine cycles until validation succeeds
    \item \textbf{CI/CD Integration}: Native understanding of build systems, test frameworks, and deployment pipelines
    \item \textbf{Unlimited Context}: Smart retrieval enables reasoning over entire repositories without token limits
\end{itemize}

\subsubsection{Modern Development Environment Analysis}

The 2025 landscape of AI-powered development tools showcases significant innovation, yet reveals fundamental limitations in debugging capabilities:

\textbf{Cursor} introduced agent mode for multi-file generation but remains optimized for code completion speed over debugging accuracy. Its 4.2\% debug success rate reflects prioritization of real-time suggestions over deep reasoning.

\textbf{Windsurf} by Codeium features Cascade technology for mapping codebases "like a neural net," enabling impressive multi-file edits. However, its lack of persistent memory and debugging-specific training limits effectiveness on complex bugs.

\textbf{Claude Code CLI} and \textbf{Gemini CLI} represent the evolution of terminal-based AI assistants. Claude Code CLI can use either Opus 4 or Sonnet 4 models, offering flexibility between performance and cost. While Claude Code excels at multi-file reasoning and Gemini CLI offers 1M token context, neither implements the iterative debugging loops necessary for autonomous bug resolution.

\textbf{Warp.dev} demonstrates interesting potential, achieving 71\% on SWE-bench with specialized configuration, but this performance doesn't generalize to real-world debugging scenarios where it achieves only 7.3\% success.

\textbf{Agentic and Reasoning Models}: Despite recent advances, even sophisticated approaches fail at debugging. Claude 4 Opus's agent mode with chain-of-thought reasoning achieves only 16.8\% debug success, marginally better than its baseline. Reasoning-specific models like o1-preview (18.2\%) and o1-mini (15.7\%) show similar limitations. Web search integration provides modest improvements: Claude with web search reaches 19.3\%, while GPT-4.1 with browsing achieves 17.6\%. These results confirm that debugging requires more than enhanced reasoning or information access, it demands specialized architecture, persistent memory, and iterative refinement capabilities that current models lack.

\subsubsection{Comparative Debugging Loop Analysis}

We analyze the debugging approaches of various systems to understand why Chronos-1 achieves superior performance:

\begin{table}[H]
\centering
\caption{Debugging loop characteristics comparison across systems}
\label{tab:debug-loop-comparison}
\resizebox{\columnwidth}{!}{%
\begin{tabular}{lcccccc}
\hline
\textbf{System} & \textbf{Loop Type} & \textbf{Avg Iter} & \textbf{Test Exec} & \textbf{Memory} & \textbf{Backtrack} & \textbf{Success} \\
\hline
Chronos-1 & Autonomous & 7.8 & \checkmark & Persistent & \checkmark & 65.3\% \\
Claude 4.5 Opus & Single-shot & 1.4 & $\times$ & Session & $\times$ & 17.1\% \\
Gemini 3 Pro & Single-shot & 1.3 & $\times$ & Session & $\times$ & 16.4\% \\
Claude 4.1 Opus & Single-shot & 1.2 & $\times$ & Session & $\times$ & 14.2\% \\
GPT-4.1 & Limited retry & 2.1 & $\times$ & Session & $\times$ & 13.8\% \\
Cursor & User-driven & 3.5 & Manual & Session & $\times$ & 4.2\% \\
Windsurf & Multi-file & 2.8 & Manual & Session & $\times$ & 6.1\% \\
Amazon Q & Guided & 4.2 & \checkmark & Session & Limited & 49.0\% \\
GitHub Copilot & Suggestion & 1.0 & $\times$ & None & $\times$ & 8.7\% \\
Warp.dev & Config-based & 3.1 & \checkmark & Session & $\times$ & 7.3\% \\
\hline
\end{tabular}%
}
\end{table}

Key observations:
\begin{itemize}
    \item \textbf{Iteration Depth}: Chronos-1 averages 7.8 iterations vs 1-4 for others, enabling deeper exploration
    \item \textbf{Test Integration}: Only Chronos-1, Amazon Q, and Warp.dev execute tests automatically
    \item \textbf{Backtracking}: Chronos-1 uniquely supports full hypothesis backtracking when fixes fail
    \item \textbf{Memory Persistence}: All competitors use session-based memory, losing context between runs
\end{itemize}

\subsection{Quantitative Analysis: Debug Cycles and Convergence Rates}

To showcase Chronos-1's real-world debugging prowess, we conduct a qualitative study involving a set of regression bug scenarios drawn from a large open-source Python project. Metrics include:
\begin{itemize}
    \item Number of attempts to converge on a passing code/test cycle.
    \item Ability to document and explain root causes compared to human reviewers.
    \item Time-to-resolution and reduction in manual engineering effort.
\end{itemize}

\begin{figure}[H]
    \centering
    \begin{tikzpicture}
        \begin{axis}[
            ybar,
            bar width=25pt,
            width=0.9\linewidth,
            height=6cm,
            ylabel={Avg. Cycles to Fix},
            xlabel={Model},
            xtick=data,
            xticklabels={Claude 4.5 Opus, Gemini 3 Pro, GPT-4.1, Claude 4.1 Opus, DeepSeek V3, Chronos-1},
            xtick align=inside,
            x tick label style={rotate=45, anchor=east, font=\small},
            major x tick style = transparent,
            enlarge x limits=0.15,
            ymin=0, ymax=8,
            ymajorgrids=true,
            grid style=dashed,
            nodes near coords,
            nodes near coords align={vertical},
            every node near coord/.append style={font=\footnotesize, yshift=2pt},
        ]
        \addplot[fill=gray!40] coordinates {(0,4.1) (1,4.2) (2,4.8) (3,4.5) (4,4.2) (5,2.2)};
        \end{axis}
    \end{tikzpicture}
    \caption{Average code-to-fix cycles for Chronos-1 and baseline models on real-world bugs (lower is better).}
    \label{fig:cycles}
\end{figure}
\FloatBarrier

Chronos-1 demonstrates not only higher accuracy and retrieval precision, but also a dramatically reduced number of debug cycles, underscoring its benefit for continuous, automated codebase reliability.

\subsection{Real-World Debugging Scenarios: Detailed Case Analysis}

To illustrate Chronos-1's debugging capabilities, we present two real-world examples from our evaluation:

\subsubsection{Case Study 1: Cross-Module Null Pointer Exception}

\textbf{Bug Report:} "Application crashes with NullPointerException when processing user exports after recent authentication refactor"

\textbf{Technical Analysis:}
\begin{itemize}
    \item \textbf{Initial Context}: 3,247 tokens from stack trace, 2 affected files
    \item \textbf{AGR Retrieval}: k=1 found 3 auth commits, k=2 expanded to 7 related files, k=3 included test cases
    \item \textbf{Confidence Progression}: 32\% $\rightarrow$ 67\% $\rightarrow$ 91\% (terminated at k=3)
\end{itemize}

\textbf{Chronos-1 Reasoning Process:}
\begin{enumerate}
    \item \textbf{Semantic Search} (142ms): Retrieved commits with "authentication" keyword via Memory Engine
    \item \textbf{Data Flow Analysis} (287ms): Traced \texttt{AuthToken} propagation through 5 call sites
    \item \textbf{Pattern Recognition} (95ms): Identified missing null-safety pattern post-refactor
    \item \textbf{Impact Analysis} (178ms): Located 2 similar vulnerabilities in \texttt{UserService} and \texttt{AdminService}
\end{enumerate}

\textbf{Generated Fix:}

\textbf{Null Pointer Fix in Authentication Flow}
\begin{lstlisting}[language=Java, basicstyle=\tiny\ttfamily, breaklines=true, xleftmargin=5pt, xrightmargin=3pt, linewidth=0.92\columnwidth, aboveskip=0.5em, belowskip=0.5em]
// In AuthService.java
public AuthToken refreshToken(String userId) {
  AuthToken token = tokenCache.get(userId);
  if (token == null || token.isExpired()) {
    token = authProvider.generateToken(userId);
+   if (token != null) {
+     tokenCache.put(userId, token);
+   }
  }
  return token;
}
// In ExportService.java  
public void exportUserData(String userId) {
  AuthToken token = authService.refreshToken(userId);
+ if (token == null) {
+   throw new AuthenticationException(
+     "Failed to authenticate user for export");
+ }
  // ... rest of export logic
}
\end{lstlisting}

\textbf{Validation Results:} 
\begin{itemize}
    \item \textbf{Test Coverage}: 47/47 existing tests passed, 3 new edge case tests generated
    \item \textbf{Performance Impact}: No measurable latency increase ($<1$ms overhead)
    \item \textbf{Fix Iterations}: Converged in 2 cycles (initial fix $\rightarrow$ test generation $\rightarrow$ validation)
    \item \textbf{Total Time}: 1.7 seconds from bug report to validated fix
\end{itemize}

\subsubsection{Case Study 2: Async Race Condition in Message Queue}

\textbf{Bug Report:} "Intermittent message loss in high-load scenarios, approximately 0.1\% messages not processed"

\textbf{Technical Analysis:}
\begin{itemize}
    \item \textbf{Initial Context}: 8,432 tokens (logs, metrics, queue implementation)
    \item \textbf{AGR Retrieval}: k=4 required due to distributed system complexity
    \item \textbf{Pattern Matching}: Found similar fix in commit \texttt{a3f8b2c} via Memory Engine
\end{itemize}

\textbf{Chronos-1 Debugging Process:}
\begin{enumerate}
    \item \textbf{Log Analysis} (523ms): Correlated message IDs with acknowledgment timestamps
    \item \textbf{Concurrency Detection} (1,247ms): Identified non-atomic operation between \texttt{ack()} and \texttt{releaseConnection()}
    \item \textbf{Historical Search} (432ms): Retrieved 3 similar race conditions from Memory Engine
    \item \textbf{Solution Synthesis} (298ms): Adapted previous synchronization pattern to current architecture
\end{enumerate}

\textbf{Generated Fix:}

\textbf{Race Condition Fix in Message Queue}
\begin{lstlisting}[language=Java, basicstyle=\tiny\ttfamily, breaklines=true, xleftmargin=5pt, xrightmargin=3pt, linewidth=0.92\columnwidth, aboveskip=0.5em, belowskip=0.5em]
// MessageProcessor.java
private void processMessage(Message msg) {
  try {
    handler.process(msg);
-   connectionPool.returnConnection(conn);
-   msg.acknowledge();
+   // Fix: Acknowledge before returning connection
+   msg.acknowledge();
+   connectionPool.returnConnection(conn);
  } catch (Exception e) {
+   // Ensure connection returned even on error
+   connectionPool.returnConnection(conn);
    msg.nack();
    throw e;
  }
}
\end{lstlisting}

\textbf{Results:} Load test with 10M messages showed 0

\subsubsection{Complete Debugging Trace: Memory Persistence Example}

To demonstrate Chronos-1's full debugging workflow and memory retention, we present a complete trace from our evaluation:

\begin{tcolorbox}[
    colback=gray!5,
    colframe=black!75,
    boxrule=0.8pt,
    arc=2pt,
    left=3pt,
    right=3pt,
    top=3pt,
    bottom=3pt,
    title={\textbf{Complete Debugging Trace: API Rate Limiting Bug}},
    fonttitle=\bfseries\large,
    coltitle=black,
    colbacktitle=gray!20
]

\textbf{1. Original Bug Report:}
\begin{quote}
\textit{Bug Report \#3847:} API returns 429 errors even when under rate limit. User reports: "Getting rate limited at 50 req/min, limit is 100."
Environment: Production, Node.js 18, Redis 7.0
\end{quote}

\textbf{2. Chronos-1 Initial Context Retrieval (178ms):}
\begin{itemize}
    \item \texttt{RateLimiter.js} - Primary implementation
    \item \texttt{config/limits.json} - Rate limit configurations  
    \item \texttt{tests/rate-limiter.test.js} - Existing test cases
    \item Commit \texttt{8f3a2b1}: "Implement sliding window rate limiting"
    \item Previous fix \texttt{2c4d5e9}: Similar Redis timing issue (from PDM)
\end{itemize}

\textbf{3. AGR Expansion (k=1 $\rightarrow$ k=3):}
\begin{itemize}
    \item k=1: Found Redis connection pool (3 files, confidence: 42\%)
    \item k=2: Discovered timestamp precision issue (7 files, 71\%)
    \item k=3: Located distributed system clock skew (12 files, 89\%)
\end{itemize}

\hspace{0.5em}\textbf{4. Generated Fix Hypothesis:}
\begin{lstlisting}[language=Java, basicstyle=\tiny\ttfamily, breaklines=true, xleftmargin=5pt, xrightmargin=3pt, linewidth=0.92\columnwidth, aboveskip=0.5em, belowskip=0.5em, keywordstyle=\color{blue}\bfseries, commentstyle=\color{gray}, stringstyle=\color{green!60!black}]
// Issue: Redis TTL using seconds, JS timestamps in milliseconds
// Causing premature key expiration in sliding window
const key = `rate:${userId}:${Math.floor(Date.now() / 1000)}`;
// Changed to maintain millisecond precision:
const key = `rate:${userId}:${Date.now()}`;
const ttl = Math.ceil(windowSizeMs / 1000); // Proper TTL conversion
\end{lstlisting}

\vspace{0.5em}
\hspace{0.5em}\textbf{5. Test Loop Execution (3 iterations):}
\begin{itemize}
    \item Iteration 1: Fix failed - TTL still incorrect
    \item Iteration 2: Adjusted Redis EXPIRE command
    \item Iteration 3: All tests pass, including new edge cases
\end{itemize}

\vspace{0.5em}
\textbf{6. Memory Update (PDM Storage):}
\begin{lstlisting}[language=Java, basicstyle=\tiny\ttfamily, breaklines=true, xleftmargin=5pt, xrightmargin=3pt, linewidth=0.92\columnwidth, aboveskip=0.5em, belowskip=0.5em, keywordstyle=\color{blue}\bfseries, commentstyle=\color{gray}, stringstyle=\color{green!60!black}]
{
  "pattern": "Redis timestamp precision mismatch",
  "context": ["RateLimiter.js", "Redis.js", "config/limits.json"],
  "fix": "Convert JS ms to Redis seconds with Math.ceil",
  "tags": ["redis", "rate-limiting", "timestamp", "distributed"],
  "confidence": 0.92,
  "reuseCount": 0  // will increment on similar bugs
}
\end{lstlisting}

\vspace{0.5em}
\textbf{7. Later Bug (2 weeks after):}
\begin{quote}
\textit{Issue \#3912:} Session expiry happening too early\\
PDM Match: 87\% similarity to "Redis timestamp precision"\\
Chronos-1: Retrieved previous fix, adapted to session context\\
Time to Fix: 47 seconds (vs 3.2 minutes for original)
\end{quote}

This trace demonstrates how PDM enables Chronos-1 to learn from past debugging sessions, dramatically reducing fix time for similar issues.

\end{tcolorbox}

\subsection{Performance at Scale: Million-Token Context Evaluation}

Even models with extended context windows fail at real debugging tasks due to fundamental architectural limitations:

\begin{table}[H]
\centering
\resizebox{\columnwidth}{!}{%
\begin{tabular}{lccccc}
\hline
\textbf{Model} & \textbf{Context Size} & \textbf{Cross-File Bugs} & \textbf{Historical Bugs} & \textbf{Complex Traces} & \textbf{Avg. Success} \\
\hline
Claude 4.5 Opus & 200K tokens & 19.3±2.0\% & 10.4±1.4\% & 17.2±1.7\% & 15.6±1.6\% \\
Gemini 3 Pro & 1M tokens & 21.8±1.9\% & 11.7±1.3\% & 19.1±1.6\% & 17.5±1.5\% \\
GPT-4.1 & 128K tokens & 15.7±2.2\% & 7.8±1.6\% & 13.2±1.9\% & 12.2±1.8\% \\
Claude 4.1 Opus & 200K tokens & 16.2±2.1\% & 8.3±1.5\% & 14.1±1.9\% & 12.9±1.8\% \\
DeepSeek V3 & 128K tokens & 13.5±2.2\% & 6.9±1.6\% & 11.8±2.0\% & 10.7±1.9\% \\
Gemini 2.5 Pro & 1M tokens & 17.1±2.1\% & 9.2±1.5\% & 15.3±1.8\% & 13.9±1.7\% \\
\textbf{Chronos-1} & \textbf{Unlimited*} & \textbf{71.2±1.8\%}*** & \textbf{68.9±2.0\%}*** & \textbf{74.3±1.6\%}*** & \textbf{71.5±1.8\%}*** \\
\hline
\multicolumn{6}{l}{\footnotesize *Via AGR. ***p $< 0.001$ compared to Gemini 2.5 Pro (paired t-test, n=100 tasks per category)}
\end{tabular}%
}
\caption{Performance on debugging tasks requiring extensive context. Despite models having up to 1M tokens, intelligent retrieval and debug-specific training enable Chronos-1 to achieve 5x better performance.}
\label{tab:long-context}
\end{table}

The results demonstrate that raw context size alone cannot solve debugging. Chronos-1's intelligent retrieval, persistent memory, and debug-specific training enable it to outperform even million-token models by over 5x.

\subsection{Detailed Performance Analysis}

We further analyze Chronos-1's performance across different bug categories and complexity levels:

\begin{table}[H]
\centering
\resizebox{\columnwidth}{!}{%
\begin{tabular}{lcccccc}
\hline
\textbf{Bug Category} & \textbf{Syntax} & \textbf{Logic} & \textbf{Concurrency} & \textbf{Memory} & \textbf{API} & \textbf{Performance} \\
\hline
Claude 4.5 Opus & 93.8\% & 22.4\% & 8.1\% & 11.3\% & 31.2\% & 15.8\% \\
Gemini 3 Pro & 93.1\% & 21.6\% & 7.7\% & 10.8\% & 29.7\% & 14.9\% \\
GPT-4.1 & 88.7\% & 17.3\% & 5.8\% & 8.2\% & 24.6\% & 11.3\% \\
Claude 4.1 Opus & 91.2\% & 18.9\% & 6.3\% & 9.1\% & 26.8\% & 12.7\% \\
DeepSeek V3 & 87.9\% & 16.2\% & 4.9\% & 7.6\% & 23.1\% & 10.5\% \\
Qwen2.5-32B & 86.8\% & 15.7\% & 4.5\% & 7.2\% & 22.3\% & 9.8\% \\
Gemini 2.5 Pro & 89.3\% & 17.8\% & 5.7\% & 8.7\% & 25.3\% & 11.9\% \\
\textbf{Chronos-1} & \textbf{94.2\%} & \textbf{72.8\%} & \textbf{58.3\%} & \textbf{61.7\%} & \textbf{79.1\%} & \textbf{65.4\%} \\
\hline
\end{tabular}%
}
\caption{Success rates by bug category. While frontier models show incremental improvements over their predecessors, Chronos-1 demonstrates 3-10x better performance on complex bug types through debug-specific training.}
\label{tab:bug-categories}
\end{table}

\begin{table}[H]
\centering
\resizebox{\columnwidth}{!}{%
\begin{tabular}{lccccc}
\hline
\textbf{Repository Size} & \textbf{$<10K$ LOC} & \textbf{10K-100K} & \textbf{100K-1M} & \textbf{$>1M$ LOC} \\
\hline
Claude 4.5 Opus & 27.8\% & 21.2\% & 14.3\% & 6.8\% \\
Gemini 3 Pro & 29.3\% & 22.7\% & 15.8\% & 7.4\% \\
GPT-4.1 & 18.5\% & 13.2\% & 7.8\% & 3.1\% \\
Claude 4.1 Opus & 21.7\% & 15.8\% & 9.2\% & 4.3\% \\
DeepSeek V3 & 19.2\% & 13.9\% & 8.1\% & 3.5\% \\
Gemini 2.5 Pro & 24.1\% & 17.3\% & 11.5\% & 5.2\% \\
\textbf{Chronos-1} & \textbf{71.2\%} & \textbf{68.9\%} & \textbf{64.3\%} & \textbf{59.7\%} \\
\hline
\end{tabular}%
}
\caption{Debugging success rates by repository size, demonstrating Chronos-1's scalability.}
\label{tab:repo-scale}
\end{table}

\subsection{Multi-Code Association Retrieval Performance}

We evaluate Chronos-1's ability to retrieve and associate multiple code artifacts for debugging:

\begin{table}[H]
\centering
\resizebox{\columnwidth}{!}{%
\begin{tabular}{lccc}
\hline
\textbf{Retrieval Task} & \textbf{Precision} & \textbf{Recall} & \textbf{F1 Score} \\
\hline
Variable Tracing & 92.3±1.4\% & 89.7±1.6\% & 91.0±1.2\% \\
Cross-File Dependencies & 88.9±1.8\% & 91.2±1.5\% & 90.0±1.4\% \\
Historical Bug Patterns & 94.1±1.1\% & 87.3±2.0\% & 90.6±1.3\% \\
Test-Code Mapping & 91.7±1.3\% & 93.5±1.2\% & 92.6±1.0\% \\
Documentation Links & 85.4±2.1\% & 88.9±1.9\% & 87.1±1.7\% \\
\hline
\textbf{Average} & \textbf{90.5±0.8\%} & \textbf{90.1±0.9\%} & \textbf{90.3±0.7\%} \\
\hline
\end{tabular}%
}
\caption{Multi-code association retrieval performance across different debugging contexts.}
\label{tab:multi-code-retrieval}
\end{table}

\subsection{Efficiency Metrics: Cost, Latency, and Resource Usage}

A critical consideration for production deployment is computational efficiency. We analyze Chronos-1's performance characteristics compared to baselines and human debugging:

\begin{table}[H]
\centering
\resizebox{\columnwidth}{!}{%
\begin{tabular}{lccccc}
\hline
\textbf{Metric} & \textbf{GPT-4.1} & \textbf{Claude 4.1 Opus} & \textbf{Gemini 2.5 Pro} & \textbf{Chronos-1} & \textbf{Human Dev} \\
\hline
Avg. Time to Fix & 68.5s & 64.2s & 59.8s & 134.7s & 2.4 hours \\
Context Window & 128K tokens & 200K tokens & 1M tokens & Unlimited* & N/A \\
Cost per Bug & \$0.52 & \$0.58 & \$0.72 & \$0.89 & \$180 \\
Success Rate & 13.8\% & 14.2\% & 13.9\% & 65.3\% & 94.2\% \\
Effective Cost* & \$3.77 & \$4.08 & \$5.18 & \$1.36 & \$191 \\
\hline
\multicolumn{6}{l}{\footnotesize *Unlimited via dynamic retrieval; Effective cost = Cost per bug / Success rate}
\end{tabular}%
}
\caption{Computational efficiency and cost analysis. Despite higher per-attempt cost, Chronos-1's high success rate yields lowest effective cost.}
\label{tab:efficiency}
\end{table}

\subsubsection{Inference Time Breakdown}

Chronos-1's 134.7s average debugging time consists of:
\begin{itemize}
    \item Context Retrieval: 23.4s (17.4\%)
    \item Multi-round Reasoning: 67.8s (50.3\%)
    \item Test Execution: 31.2s (23.2\%)
    \item Memory Update: 12.3s (9.1\%)
\end{itemize}

\subsubsection{Return on Investment Analysis}

For a typical enterprise with 100 developers:
\begin{itemize}
    \item Annual debugging time: ~150,000 hours
    \item Chronos-1 automation potential: 65.3\% × 150,000 = 97,950 hours
    \item Cost savings: 97,950 × \$90/hour - deployment costs = \$8.1M annually
    \item ROI: 47:1 in first year, accounting for infrastructure and licensing
\end{itemize}

\subsubsection{Detailed Inference Cost and Latency Analysis}

We conducted comprehensive benchmarking of inference costs and latency across different bug complexity levels:

\begin{table}[H]
\centering
\caption{Inference cost and latency breakdown by bug complexity and system}
\label{tab:inference-cost-latency}
\resizebox{\columnwidth}{!}{%
\begin{tabular}{lcccccc}
\hline
\textbf{System} & \multicolumn{2}{c}{\textbf{Simple Bugs}} & \multicolumn{2}{c}{\textbf{Medium Bugs}} & \multicolumn{2}{c}{\textbf{Complex Bugs}} \\
 & Cost (\$) & Latency (s) & Cost (\$) & Latency (s) & Cost (\$) & Latency (s) \\
\hline
Claude 4.5 Opus & 0.24 & 13.8 & 0.63 & 61.2 & 1.52 & 192.7 \\
Gemini 3 Pro & 0.28 & 10.2 & 0.68 & 55.4 & 1.67 & 168.3 \\
GPT-4.1 (128K) & 0.18 & 12.3 & 0.52 & 68.5 & 1.24 & 187.2 \\
Claude 4.1 Opus & 0.21 & 14.7 & 0.58 & 64.2 & 1.41 & 201.3 \\
Gemini 2.0 Pro & 0.31 & 8.9 & 0.72 & 59.8 & 1.93 & 156.4 \\
Amazon Q Dev & 0.42 & 21.3 & 0.91 & 94.7 & 2.14 & 234.6 \\
\hline
Chronos-1 (Ours) & 0.34 & 31.2 & 0.89 & 134.7 & 1.78 & 298.4 \\
- Retrieval only & 0.08 & 7.8 & 0.19 & 23.4 & 0.41 & 52.3 \\
- LLM inference & 0.21 & 18.9 & 0.56 & 67.8 & 1.12 & 156.8 \\
- Test execution & 0.03 & 3.1 & 0.09 & 31.2 & 0.18 & 67.2 \\
- Memory update & 0.02 & 1.4 & 0.05 & 12.3 & 0.07 & 22.1 \\
\hline
\multicolumn{7}{l}{\textit{Cost-Effectiveness Analysis (Cost / Success Rate)}} \\
\hline
Claude 4.5 Opus & 1.40 & - & 3.68 & - & 8.89 & - \\
Gemini 3 Pro & 1.71 & - & 4.15 & - & 10.18 & - \\
GPT-4.1 & 1.30 & - & 3.77 & - & 8.99 & - \\
Claude 4.1 Opus & 1.48 & - & 4.08 & - & 9.93 & - \\
Gemini 2.0 Pro & 2.23 & - & 5.18 & - & 13.88 & - \\
Amazon Q Dev & 0.86 & - & 1.86 & - & 4.37 & - \\
Chronos-1 & \textbf{0.52} & - & \textbf{1.36} & - & \textbf{2.73} & - \\
\hline
\end{tabular}%
}
\end{table}

Key observations:
\begin{itemize}
    \item \textbf{Latency vs Accuracy Trade-off}: Chronos-1's higher latency (2-3x) is offset by 4-5x better success rates
    \item \textbf{Cost Scaling}: Complex bugs show sublinear cost increase (5.2x) despite exponential difficulty
    \item \textbf{Component Distribution}: LLM inference dominates cost (63\%) while retrieval dominates initial latency
    \item \textbf{Parallelization Opportunity}: Test execution (23\% of latency) can be parallelized for 1.3x speedup
\end{itemize}

\FloatBarrier
\section{ANALYSIS: WHY CHRONOS SUCCEEDS WHERE OTHERS FAIL}

\subsection{The Debugging Specialization Hypothesis: Evidence and Implications}

The stark performance gap between Chronos-1 (65.3\% debug success) and state-of-the-art general models ($\leq$15\%) reveals a fundamental truth: debugging is not simply an extension of code generation. While models like Claude 4.1 Opus (72.5\% on SWE-bench) and Claude 4.5 Sonnet (72.7\%) and GPT-4.1 (54.6\%) excel at writing new code, debugging demands distinct capabilities:

\begin{enumerate}
    \item \textbf{Temporal Reasoning}: Understanding how code evolved through commits, why certain patterns were introduced, and which changes correlate with bug emergence. General models lack this historical perspective.
    
    \item \textbf{Multi-Modal Signal Integration}: Debugging requires synthesizing error traces, logs, test failures, and code, a fundamentally different task than generating syntactically correct code from specifications.
    
    \item \textbf{Iterative Hypothesis Testing}: The debugging loop of propose$\rightarrow$test$\rightarrow$analyze$\rightarrow$refine cannot be learned through next-token prediction alone. Chronos-1's reinforcement learning from test execution feedback creates this capability.
    
    \item \textbf{Persistent Pattern Memory}: Bugs often recur in variations. Without persistent memory of past fixes and anti-patterns, models repeat mistakes, explaining why session-based tools achieve $<10\%$ success.
\end{enumerate}

\subsection{Key Architectural Insights: Memory, Iteration, and Context}

Our evaluation highlights several domains where Chronos-1 delivers outsized impact compared to prior systems:
\begin{itemize}
    \item \textbf{Holistic Bug Localization:} Chronos-1 traces complex error origins across modules, commits, and documentation with no manual guidance, routinely identifying root causes overlooked by token-limited models.
    \item \textbf{Autonomous Debugging Loops:} Chronos-1 adapts its retrieval and patching behavior over multiple test cycles, integrating failed test feedback and reviewer commentary to iteratively refine solutions.
    \item \textbf{Continuous Knowledge Incorporation:} By feeding CI/CD, reviewer, and test feedback into persistent memory, Chronos-1improves its project-specific performance over time, exhibiting lower repeated error rates and faster adaptation to new code patterns.
\end{itemize}

\begin{table}[H]
\centering
\resizebox{\columnwidth}{!}{%
\begin{tabular}{l|l|l|p{5.5cm}}
\hline
\textbf{Bug Scenario} & \textbf{GPT-4.1} & \textbf{Chronos-1} & \textbf{Chronos-1 Resolution Path} \\
\hline
Test Failure on \texttt{user\_auth} & Incorrect var patch & Full fix & Traced import drift $\rightarrow$ found stale config $\rightarrow$ auto-fix and doc update \\
\hline
API Deprecation & Missed call-site & Full fix & Multi-code association retrieved usage in 3 files, migrated all refs \\
\hline
Intermittent CI Error & Flaky retry logic & Full fix & Ingested CI logs, patched async boundary, added test case and explanation \\
\hline
\end{tabular}%
}
\caption{Qualitative examples where Chronos-1 successfully applies multi-code context to resolve debugging tasks beyond the reach of baseline LLMs.}
\label{tab:qualitative}
\end{table}

\subsection{Ablation Studies}

To isolate the contribution of core design features, we perform targeted ablations on our internal debugging benchmarks (MRR, DebugBench). \textbf{Important:} These ablations demonstrate component contributions under normal operation. For SWE-bench Lite evaluation, Chronos-1operated under strict compliance with the submission checklist:
\begin{itemize}
    \item \textbf{Pass@1:} Single attempt per instance (no retry logic, fallback planners, or refinement cycles)
    \item \textbf{No Test Knowledge:} PASS\_TO\_PASS and FAIL\_TO\_PASS signals disabled
    \item \textbf{No Hints:} Only raw issue descriptions used
    \item \textbf{No Web Browsing:} Fully offline sandbox with no external lookups
\end{itemize}

Ablation results on internal benchmarks:
\begin{itemize}
    \item \textbf{No Multi-Code Association:} When Chronos-1 is restricted to single-chunk retrieval, debug success falls by 45\% and retrieval precision drops sharply, mirroring the limitations of prior RAG pipelines.
    \item \textbf{Static Memory Only:} If the live feedback/memory update mechanism is ablated (i.e., only static embeddings used), adaptivity stagnates, and repeated bug classes recur more often.
    \item \textbf{No Orchestration Loop:} Disabling the validate-retrieve-update workflow reverts performance to basic code suggestion with higher error rate and longer time-to-fix.
\end{itemize}

\begin{table}[H]
\centering
\caption{Ablation study on internal debugging benchmarks showing contribution of each component. \checkmark{} indicates component enabled. \textit{Note: Test Loop was disabled for SWE-bench Lite to comply with pass@1 requirements.}}
\label{tab:ablation-components}
\resizebox{\columnwidth}{!}{%
\begin{tabular}{lccccc}
\hline
\textbf{Model Variant} & \textbf{Memory} & \textbf{AGR} & \textbf{Test Loop} & \textbf{Fix Rate↑} & \textbf{Bug Loc.↑} \\
\hline
Full Chronos-1 & \checkmark & \checkmark & \checkmark & \textbf{65.3\%} & \textbf{87.2\%} \\
\hline
w/o Memory & $\times$ & \checkmark & \checkmark & 48.1\% (-17.2) & 71.5\% (-15.7) \\
w/o AGR & \checkmark & $\times$ & \checkmark & 42.6\% (-22.7) & 63.8\% (-23.4) \\
w/o Test Loop & \checkmark & \checkmark & $\times$ & 51.7\% (-13.6) & 79.3\% (-7.9) \\
\hline
w/o Memory+AGR & $\times$ & $\times$ & \checkmark & 31.2\% (-34.1) & 52.1\% (-35.1) \\
w/o Memory+Test & $\times$ & \checkmark & $\times$ & 35.8\% (-29.5) & 58.7\% (-28.5) \\
w/o AGR+Test & \checkmark & $\times$ & $\times$ & 33.4\% (-31.9) & 55.2\% (-32.0) \\
\hline
Baseline (None) & $\times$ & $\times$ & $\times$ & 22.1\% (-43.2) & 41.3\% (-45.9) \\
\hline
\end{tabular}%
}
\end{table}

\begin{figure}[H]
    \centering
    \begin{tikzpicture}
        \begin{axis}[
            ybar,
            bar width=15pt,
            width=\columnwidth,
            height=5cm,
            ylabel={Debug Success (\%)},
            xlabel={Configuration},
            xtick=data,
            xticklabels={{Full},{No MCA},{Static},{No Orch}},
            x tick label style={font=\footnotesize},
            y tick label style={font=\footnotesize},
            ylabel style={font=\footnotesize},
            xlabel style={font=\footnotesize},
            xtick align=inside,
            major x tick style = transparent,
            enlarge x limits=0.25,
            ymin=0, ymax=100,
            ymajorgrids=true,
            grid style=dashed,
            nodes near coords,
            nodes near coords align={vertical},
            every node near coord/.append style={font=\tiny},
        ]
        \addplot[fill=blue!60] coordinates {(0,90) (1,49) (2,62) (3,55)};
        \end{axis}
    \end{tikzpicture}
    \caption{Ablation analysis on internal benchmarks: Debugging success rate with each Chronos-1 core component removed (lower is worse). For SWE-bench Lite, only single-pass inference was used per submission requirements.}
    \label{fig:ablation}
\end{figure}

These findings underscore the essential synergy between deep memory, multi-code contextualization, and autonomous workflow orchestration for effective debugging and adaptive code maintenance.

\subsection{Failure Mode Analysis: Understanding Chronos-1's Limitations}

Despite Chronos-1's strong performance, our analysis reveals specific failure modes and bug categories where the system struggles:

\subsubsection{Common Failure Modes}

\begin{enumerate}
    \item \textbf{Hardware-Dependent Bugs:} Chronos-1 achieves only 23.4\% success on bugs requiring hardware-specific knowledge (e.g., GPU memory alignment, embedded system timing). Example failure:
    \begin{quote}
    \small
    \textit{Bug:} CUDA kernel crashes with unaligned memory access on Tesla V100\\
    \textit{Chronos-1 Fix:} Added boundary checks (incorrect)\\
    \textit{Correct Fix:} Aligned memory allocation to 128-byte boundaries
    \end{quote}
    
    \item \textbf{Distributed System Race Conditions:} Complex timing-dependent bugs across multiple services show 31.2\% success rate. The model struggles to reason about non-deterministic execution orders across network boundaries.
    
    \item \textbf{Domain-Specific Logic Errors:} Bugs requiring deep domain knowledge (medical, financial regulations) succeed only 28.7\% of the time. Example:
    \begin{quote}
    \small
    \textit{Bug:} HIPAA compliance violation in patient data export\\
    \textit{Issue:} Chronos-1 lacks healthcare regulatory knowledge
    \end{quote}
\end{enumerate}

\subsubsection{Edge Cases and Limitations}

\begin{itemize}
    \item \textbf{Extremely Large Monorepos ($>10M$ LOC):} Performance degrades to 45.3\% success rate due to retrieval precision issues
    \item \textbf{Legacy Code with Poor Documentation:} Success drops to 38.9\% when code lacks comments and uses cryptic variable names
    \item \textbf{Multi-Language Polyglot Systems:} Cross-language bugs (e.g., Python calling Rust via FFI) show only 41.2\% success
    \item \textbf{UI/UX Bugs:} Visual rendering issues essentially unsolvable (8.3\% success) without screenshot analysis
\end{itemize}

\textbf{Concrete Failure Case Study:} Chronos-1 struggled with a TypeScript monorepo using Lerna, dynamic imports, and cross-package configuration resolution. Despite correctly retrieving all relevant files across packages, PDM lacked adequate context to resolve type scope ambiguity when multiple packages exported identically-named interfaces. The fix attempted to import from the wrong package, causing circular dependencies. This failure mode highlights the need for semantic-aware static linking and better understanding of module resolution strategies in complex build systems. Future work should incorporate build tool semantics directly into the memory layer.

\begin{table}[H]
\centering
\begin{tabular}{lcc}
\hline
\textbf{Bug Category} & \textbf{Success Rate} & \textbf{Primary Failure Reason} \\
\hline
Hardware-Specific & 23.4±3.2\% & Lacks hardware specs \\
Distributed Race & 31.2±2.8\% & Non-deterministic timing \\
Domain Logic & 28.7±3.1\% & Missing domain knowledge \\
Legacy Code & 38.9±2.9\% & Poor code quality \\
Cross-Language & 41.2±2.7\% & FFI complexity \\
UI/Visual & 8.3±1.9\% & No visual understanding \\
\hline
\end{tabular}
\caption{Chronos-1 performance on challenging bug categories.}
\label{tab:failure-analysis}
\end{table}

\FloatBarrier
\FloatBarrier

\subsection{Adversarial Evaluation}

To assess Chronos-1's robustness against malicious inputs and poisoned training data, we conduct comprehensive adversarial testing:

\subsubsection{Attack Methodology}

We evaluate three categories of adversarial attacks:

\textbf{1. Input Perturbation Attacks:}
\begin{itemize}
    \item \textbf{Prompt Injection:} Inserting malicious instructions in bug descriptions
    \item \textbf{Code Obfuscation:} Deliberately confusing variable names and control flow
    \item \textbf{Misleading Comments:} Comments that contradict actual code behavior
\end{itemize}

\textbf{2. Data Poisoning Attacks:}
\begin{itemize}
    \item \textbf{Backdoor Triggers:} Specific code patterns that trigger incorrect fixes
    \item \textbf{Label Flipping:} Training examples with intentionally wrong fixes
    \item \textbf{Gradient Attacks:} Adversarial examples crafted to maximize loss
\end{itemize}

\textbf{3. Retrieval Manipulation:}
\begin{itemize}
    \item \textbf{Context Stuffing:} Flooding retrieval with irrelevant but similar code
    \item \textbf{Dependency Confusion:} Creating fake dependencies to mislead AGR
    \item \textbf{Temporal Attacks:} Manipulating commit timestamps to affect retrieval
\end{itemize}

\subsubsection{Robustness Results}

\begin{table}[H]
\centering
\caption{Adversarial robustness evaluation on 1,000 attack samples per category}
\label{tab:adversarial}
\resizebox{\columnwidth}{!}{%
\begin{tabular}{llccc}
\hline
\textbf{Category} & \textbf{Attack Type} & \textbf{Success} & \textbf{Detection} & \textbf{Mitigation} \\
& & \textbf{Rate} & \textbf{Rate} & \\
\hline
\textit{Input} & Prompt Injection & 12.3\% & 87.7\% & Input sanitization \\
\textit{Perturbation} & Code Obfuscation & 23.1\% & 76.9\% & AST normalization \\
& Misleading Comments & 8.7\% & 91.3\% & Code-comment align \\
\hline
\textit{Data} & Backdoor Triggers & 5.2\% & 94.8\% & Anomaly detection \\
\textit{Poisoning} & Label Flipping & 15.6\% & 84.4\% & Consistency checks \\
& Gradient Attacks & 19.3\% & 80.7\% & Gradient clipping \\
\hline
\textit{Retrieval} & Context Stuffing & 31.4\% & 68.6\% & Relevance filtering \\
\textit{Manipulation} & Dependency Confusion & 27.8\% & 72.2\% & Graph validation \\
& Temporal Attacks & 11.2\% & 88.8\% & Timestamp verify \\
\hline
\end{tabular}%
}
\end{table}

\textbf{Key Findings:}
\begin{itemize}
    \item AGR's graph structure provides natural defense against context stuffing (68.6\% detection)
    \item PDM's pattern matching detects 94.8\% of backdoor triggers through anomaly scores
    \item Temporal attacks are largely ineffective due to multi-signal validation
\end{itemize}

\subsection{Scalability Analysis}

We evaluate Chronos-1's performance across codebases ranging from 1K to 100M lines of code:

\subsubsection{Performance Metrics by Scale}

\begin{table}[H]
\centering
\caption{Scalability analysis across different codebase sizes}
\label{tab:scalability}
\resizebox{\columnwidth}{!}{%
\begin{tabular}{@{}lccccc@{}}
\toprule
\textbf{Codebase Size} & \textbf{Debug} & \textbf{Avg Response} & \textbf{Memory} & \textbf{Graph Build} & \textbf{Index} \\
\textbf{(LOC)} & \textbf{Success} & \textbf{Time (s)} & \textbf{Usage (GB)} & \textbf{Time (min)} & \textbf{Size (GB)} \\
\midrule
1K-10K & 91.2\% & 2.3 & 0.5 & 0.1 & 0.01 \\
10K-100K & 89.7\% & 4.7 & 2.1 & 1.2 & 0.08 \\
100K-1M & 87.1\% & 8.9 & 8.7 & 12.4 & 0.92 \\
1M-10M & 82.3\% & 18.2 & 31.5 & 87.3 & 9.7 \\
10M-100M & 73.4\% & 45.7 & 128.3 & 512.8 & 98.2 \\
\bottomrule
\end{tabular}%
}
\end{table}

\subsubsection{Optimization Strategies at Scale}

\begin{enumerate}
\item \textbf{Incremental Graph Updates:} For codebases $>1M$ LOC, we implement differential graph updates that process only changed files, reducing update time by 87\%.

\item \textbf{Hierarchical Indexing:} Multi-level indexes with module-level summaries enable sub-linear retrieval complexity O(log n) for repositories $>10M$ LOC.

\item \textbf{Distributed Processing:} Graph construction parallelizes across 32 cores, achieving near-linear speedup for codebases up to 50M LOC.

\item \textbf{Memory-Mapped Storage:} PDM uses memory-mapped files for patterns exceeding RAM, maintaining $<100$ms access latency.
\end{enumerate}

\subsubsection{Detailed Memory Consumption Analysis}

We profile memory usage across different components and operations:

\begin{table}[H]
\centering
\caption{Memory consumption breakdown by component (1M LOC codebase)}
\label{tab:memory-breakdown}
\resizebox{0.9\columnwidth}{!}{%
\begin{tabular}{lccc}
\hline
\textbf{Component} & \textbf{Peak Memory} & \textbf{\% of Total} & \textbf{Growth Rate} \\
& \textbf{(GB)} & & \\
\hline
AGR Graph Structure & 2.8 & 32.2\% & O(n log n) \\
PDM Pattern Storage & 1.9 & 21.8\% & O(m) \\
Embedding Cache & 1.6 & 18.4\% & O(n) \\
AST Representations & 1.2 & 13.8\% & O(n) \\
Temporal Indexes & 0.7 & 8.0\% & O(n log t) \\
Runtime Buffers & 0.5 & 5.8\% & O(1) \\
\hline
\textbf{Total} & \textbf{8.7} & \textbf{100\%} & - \\
\hline
\end{tabular}%
}
\end{table}

\textbf{Memory Optimization Strategies:}
\begin{itemize}
    \item \textbf{Lazy Loading:} Graph nodes loaded on-demand, reducing baseline memory by 65\%
    \item \textbf{Pattern Compression:} PDM patterns compressed using AST-aware encoding (3.2x compression ratio)
    \item \textbf{Cache Eviction:} LRU policy for embeddings with 90\% hit rate at 20\% memory cost
    \item \textbf{Incremental GC:} Custom garbage collection for graph traversal reduces peak memory by 40\%
\end{itemize}

\subsection{Human Evaluation Study}

We conducted a controlled study with N=50 professional developers to assess Chronos-1's real-world effectiveness:

\subsubsection{Study Design}

\textbf{Participants:} 50 developers (5-15 years experience) from 12 companies
\begin{itemize}
    \item 20 backend engineers (Java/Python)
    \item 15 full-stack developers (JavaScript/TypeScript)
    \item 10 infrastructure engineers (Go/Rust)
    \item 5 ML engineers (Python/C++)
\end{itemize}

\textbf{Tasks:} Each developer debugged 10 real production bugs:
\begin{itemize}
    \item 5 bugs with Chronos-1 assistance
    \item 5 bugs with their preferred tools (baseline)
    \item Randomized order to prevent learning effects
    \item Bugs selected from actual production incidents
\end{itemize}

\subsubsection{Quantitative Results}

\begin{table}[H]
\centering
\caption{Human evaluation results (N=50 developers, 500 debugging sessions)}
\label{tab:human-eval}
\resizebox{0.95\columnwidth}{!}{%
\begin{tabular}{@{}lcccc@{}}
\toprule
\textbf{Metric} & \textbf{With Chronos-1} & \textbf{Baseline} & \textbf{Improvement} & \textbf{p-value} \\
\midrule
Fix Success Rate & 87.2\% & 62.4\% & +39.7\% & $<0.001$ \\
Time to Fix (min) & 18.3±7.2 & 43.7±19.8 & -58.1\% & $<0.001$ \\
Code Quality Score & 8.7/10 & 7.2/10 & +20.8\% & $<0.01$ \\
Confidence Rating & 8.9/10 & 6.8/10 & +30.9\% & $<0.001$ \\
Would Use Again & 92\% & - & - & - \\
\bottomrule
\end{tabular}%
}
\end{table}

\subsubsection{Qualitative Feedback}

\textbf{Positive Themes:}
\begin{itemize}
    \item "AGR found cross-file dependencies I would have missed" (31/50 developers)
    \item "PDM patterns saved hours on recurring issues" (27/50)
    \item "Explanations helped me understand the fix" (42/50)
\end{itemize}

\textbf{Areas for Improvement:}
\begin{itemize}
    \item "Needs better IDE integration" (18/50)
    \item "Sometimes retrieves too much context" (12/50)
    \item "UI lag on very large codebases" (8/50)
\end{itemize}

\subsection{Fine-grained Ablations}

We conduct detailed ablations on individual AGR components using internal benchmarks (MRR, DebugBench) to understand their contributions. These results reflect full-capability operation on internal benchmarks.

\textbf{SWE-bench Lite Evaluation Constraints:} Chronos-1's 80.33\% score on SWE-bench Lite was achieved under strict submission requirements that disabled several capabilities shown in these ablations:
\begin{itemize}
    \item \textbf{Pass@1:} Single attempt per instance (no retry logic, fallback planners, or refinement cycles)
    \item \textbf{No Autonomous Loop:} Fix-test-refine iteration disabled; single-pass patch generation only
    \item \textbf{No PDM Cross-Session:} Persistent Debug Memory limited to static patterns (no learning from current session)
    \item \textbf{No Test Knowledge:} PASS\_TO\_PASS and FAIL\_TO\_PASS signals disabled
    \item \textbf{No Hints:} Only raw issue descriptions used
    \item \textbf{No Web Browsing:} Fully offline sandbox with no external lookups
\end{itemize}
The ablation results below show component contributions under full operation; SWE-bench Lite performance reflects constrained single-pass inference without the iterative refinement that provides an additional 15-20\% improvement in internal benchmarks.

\subsubsection{Component-wise Ablation Results}

\begin{table}[H]
\centering
\caption{Fine-grained ablation study on AGR components (MRR benchmark). \textit{Note: SWE-bench Lite's 80.33\% was achieved with test loops and refinement disabled per pass@1 requirements.}}
\label{tab:fine-ablation}
\footnotesize
\setlength{\tabcolsep}{3pt}
\begin{tabular}{@{}p{1.5cm}p{3.8cm}ccc@{}}
\hline
\textbf{Category} & \textbf{Configuration} & \textbf{Debug} & \textbf{Prec@3} & \textbf{Hops} \\
\hline
\textit{Baseline} & Full AGR System & \textbf{87.1\%} & \textbf{92.8\%} & 2.3 \\
\hline
\textit{Graph} & Remove AST edges & 71.2\% & 84.1\% & 3.1 \\
\textit{Construction} & Remove log flow edges & 78.9\% & 87.3\% & 2.7 \\
& Remove test trace edges & 73.4\% & 81.2\% & 3.8 \\
& Remove temporal weights & 82.3\% & 89.7\% & 2.5 \\
\hline
\textit{Traversal} & Fixed k=3 (no adaptive) & 69.8\% & 78.4\% & 3.0 \\
\textit{Strategy} & No confidence weighting & 74.1\% & 83.2\% & 4.2 \\
& BFS instead of guided & 61.3\% & 71.8\% & 5.7 \\
& No backtracking & 79.2\% & 88.1\% & 2.8 \\
\hline
\textit{Memory} & No PDM patterns & 72.8\% & 85.3\% & 3.2 \\
\textit{Integration} & No pattern ranking & 81.2\% & 90.1\% & 2.6 \\
& No temporal decay & 83.7\% & 91.2\% & 2.4 \\
\hline
\end{tabular}
\end{table}

\noindent\textbf{Critical Findings:}

\begin{enumerate}
    \item \textbf{AST edges are most critical:} Removing them causes 18.3\% performance drop
    \item \textbf{Adaptive depth is essential:} Fixed k=3 reduces success by 19.9\%
    \item \textbf{PDM integration provides 16.4\% boost:} Pattern memory significantly improves debugging
    \item \textbf{Guided traversal beats BFS by 29.6\%:} Intelligence in path selection matters
\end{enumerate}

\section{DESIGN RATIONALE AND THEORETICAL MOTIVATION FOR AGR}

The Adaptive Graph-Guided Retrieval (AGR) system addresses three fundamental limitations of vector-based retrieval in debugging contexts. First, real codebases exhibit structural dependencies that span multiple files, a bug in \texttt{AuthService.java} may stem from a configuration change in \texttt{config.yml} that affects \texttt{DatabasePool.java}, which then impacts authentication. Traditional top-k vector search fails here because these files share minimal textual similarity. AGR's multi-hop traversal follows actual code dependencies (imports, function calls, data flow), recovering the complete causal chain regardless of embedding distances. Second, local neighborhood propagation in AGR exploits the insight that code proximity (in the dependency graph) correlates strongly with debugging relevance. When a test fails in \texttt{PaymentTest.java}, the bug is exponentially more likely to be within 1-3 hops in the dependency graph than in a random file with high textual similarity. This locality principle, combined with confidence-weighted traversal, enables AGR to achieve 89\% precision where vector search achieves only 31\%. Third, temporal anchoring prevents the retrieval of stale code patterns in actively maintained repositories. By weighting edges based on commit recency and co-modification frequency, AGR naturally prioritizes current architectural patterns over deprecated ones, reducing false positives by 67\% in repositories with $>1000$ commits/month.

\begin{figure}[h]
\centering
\resizebox{0.9\columnwidth}{!}{%
\begin{tikzpicture}[
    node distance=1.5cm,
    file/.style={rectangle, draw, minimum width=2.5cm, minimum height=0.8cm, font=\small},
    bug/.style={rectangle, draw, rounded corners, minimum width=2.5cm, minimum height=0.8cm, font=\small},
    result/.style={rectangle, draw, rounded corners, minimum width=2.5cm, minimum height=0.8cm, font=\small}
]
    \node[bug, fill=red!20] (query1) at (0,0) {Bug: NPE in Export};
    \node[file, fill=gray!20] (vec1) at (0,-1.8) {ExportService.java};
    \node[file, fill=gray!20] (vec2) at (-2.2,-3.8) {ExportUtils.java};
    \node[file, fill=gray!20] (vec3) at (2.2,-3.8) {ExportConfig.java};
    \node[result, fill=red!40] (miss1) at (0,-5.4) {$\times$ Misses Auth};
    
    \draw[->, thick] (query1) -- (vec1) node[midway, right, font=\scriptsize] {sim=0.89};
    \draw[->, dashed] (vec1) -- (vec2) node[midway, left, font=\scriptsize] {sim=0.72};
    \draw[->, dashed] (vec1) -- (vec3) node[midway, right, font=\scriptsize] {sim=0.68};
    
    \node[font=\small\bfseries] at (0,-6.4) {Vector RAG: Top-3 Similar};
    
    \node[bug, fill=green!20] (query2) at (8.5,0) {Bug: NPE in Export};
    \node[file, fill=gray!20] (agr1) at (8.5,-1.8) {ExportService.java};
    \node[file, fill=yellow!30] (agr2) at (6.2,-3.8) {AuthService.java};
    \node[file, fill=yellow!30] (agr3) at (10.8,-3.8) {TokenCache.java};
    \node[result, fill=green!40] (found) at (8.5,-5.4) {\checkmark{} Found Root Cause};
    
    \draw[->, thick] (query2) -- (agr1) node[midway, right, font=\scriptsize] {hop=0};
    \draw[->, thick, blue] (agr1) -- (agr2) node[midway, above left, font=\scriptsize, xshift=-2pt] {calls};
    \draw[->, thick, blue] (agr2) -- (agr3) node[midway, below, font=\scriptsize, yshift=-2pt] {uses};
    \draw[->, thick, orange] (agr1) to[bend right=15] (agr3) node[pos=0.65, above, font=\scriptsize, yshift=12pt] {\shortstack{recent\\commit}};
    
    \node[font=\small\bfseries] at (8.5,-6.4) {AGR: Dependency Traversal};
    
    \draw[->, ultra thick, red] (3.2,-2.8) -- (5.2,-2.8) node[midway, above, font=\small\bfseries] {3x better};
\end{tikzpicture}%
}
\caption{AGR vs Vector RAG: While vector search retrieves textually similar files that miss the root cause, AGR follows actual code dependencies and temporal signals to locate the authentication bug causing the export failure.}
\label{fig:agr-vs-vector}
\end{figure}

The superiority of AGR demonstrated in Figure~\ref{fig:agr-vs-vector} is not merely empirical, it stems from fundamental theoretical properties that we now formalize.

\section{THEORETICAL GUARANTEES}

Having demonstrated AGR's practical superiority through extensive evaluation, we now provide theoretical foundations that explain \textit{why} these improvements are fundamental rather than incidental.

\subsection{Why AGR Converges Efficiently}

\begin{tcolorbox}[colback=blue!5,colframe=blue!40!black,title=Key Insight]
AGR achieves rapid convergence by exploiting code locality: bugs are typically within 3-5 hops of the error location in the dependency graph. This locality principle enables AGR to find relevant context in O(k log d) time with $>98\%$ probability.
\end{tcolorbox}

\textbf{Convergence Guarantee:} Given a code graph with maximum degree $d$, AGR converges to the optimal debugging context within $k$ iterations (typically k=5) with probability exceeding 98\%. 

\textbf{Practical Implications:}
\begin{itemize}
    \item \textbf{Bounded search:} AGR examines at most $O(d^k)$ nodes, preventing exponential blowup
    \item \textbf{Early termination:} 73\% of bugs are resolved within 3 hops
    \item \textbf{Predictable performance:} Response time scales logarithmically with codebase size
\end{itemize}

\subsection{Fix Correctness with PDM}

The Persistent Debug Memory provides statistical guarantees on fix quality:

\textbf{Pattern Matching Accuracy:} With $m$ stored bug patterns, the probability of generating a correct fix increases as:
$$P(\text{correct fix}) \geq 1 - e^{-m/1000}$$

This means with just 3,000 patterns, Chronos-1 achieves $>95\%$ fix accuracy on similar bugs.

\subsection{Why This Matters for Production Debugging}

Unlike theoretical debugging models, AGR's guarantees translate directly to production benefits:

\begin{table}[H]
\centering
\caption{Theoretical guarantees vs. real-world impact}
\begin{tabular}{lcc}
\hline
\textbf{Theoretical Property} & \textbf{Guarantee} & \textbf{Production Impact} \\
\hline
Convergence time & O(k log d) & $<10$s for 1M LOC \\
Fix accuracy & $>95\%$ with PDM & 3x fewer rollbacks \\
Memory usage & O(n log n) & Fits in 32GB RAM \\
False positive rate & $<5\%$ & Trustworthy fixes \\
\hline
\end{tabular}
\end{table}

\section{FAILURE ANALYSIS: DEEP DIVE INTO CHALLENGING CASES}

While our theoretical guarantees and empirical results demonstrate Chronos-1's effectiveness, understanding failure modes is crucial for both users and future improvements. We conduct systematic failure analysis on 2,500 debugging sessions where Chronos-1 failed to produce correct fixes:

\subsection{Failure Taxonomy}

\begin{table}[H]
\centering
\caption{Detailed failure analysis across 2,500 failed debugging sessions}
\label{tab:failure-taxonomy}
\resizebox{\columnwidth}{!}{%
\begin{tabular}{@{}llccl@{}}
\toprule
\textbf{Failure Category} & \textbf{Subcategory} & \textbf{Count} & \textbf{\% of Failures} & \textbf{Root Cause} \\
\midrule
\multirow{3}{*}{\parbox{3cm}{\textit{Retrieval Failures\\(42.3\% of total)}}} 
& Missing Context & 423 & 16.9\% & AGR depth limit reached \\
& Wrong Context & 312 & 12.5\% & Semantic aliasing confusion \\
& Stale Context & 323 & 12.9\% & Outdated cached patterns \\
\midrule
\multirow{3}{*}{\parbox{3cm}{\textit{Understanding Failures\\(31.2\% of total)}}} 
& Complex Logic & 287 & 11.5\% & Multi-step reasoning limit \\
& Domain Specific & 234 & 9.4\% & Missing specialized knowledge \\
& Implicit Behavior & 256 & 10.2\% & Undocumented assumptions \\
\midrule
\multirow{3}{*}{\parbox{3cm}{\textit{Generation Failures\\(26.5\% of total)}}} 
& Incorrect Fix & 298 & 11.9\% & Wrong root cause identified \\
& Partial Fix & 189 & 7.6\% & Incomplete solution \\
& Breaking Changes & 176 & 7.0\% & Side effects not considered \\
\bottomrule
\end{tabular}%
}
\end{table}

\subsection{Representative Failure Cases}

\subsubsection{Case 1: Distributed Consensus Bug}

\textbf{Bug Description:} Raft consensus implementation fails during leader election with 5+ nodes

\textbf{Failure Mode:} Chronos-1 retrieved all relevant consensus code but failed to identify the subtle race condition in vote counting across network partitions.

\begin{minipage}{\columnwidth}
\begin{lstlisting}[language=Java, basicstyle=\tiny\ttfamily, breaklines=true, xleftmargin=5pt, xrightmargin=3pt, linewidth=0.92\columnwidth, aboveskip=0.5em, belowskip=0.5em, caption={Distributed consensus bug: Chronos-1's incorrect fix}, label=lst:consensus-bug]
// Chronos-1 attempted fix (incorrect)
if (voteCount > clusterSize / 2) {
    becomeLeader(); // Missing: partition check
}

// Correct fix required understanding of:
// 1. Network partition detection
// 2. Split-brain prevention
// 3. Distributed state consistency
\end{lstlisting}
\end{minipage}

\textbf{Root Cause:} Lacks theoretical understanding of distributed systems safety properties and FLP impossibility theorem implications.

\subsubsection{Case 2: Memory Corruption in C++ Template Metaprogramming}

\textbf{Bug Description:} Segfault in recursive template instantiation with variadic parameters

\textbf{Failure Mode:} AGR correctly identified template definitions but PDM had no patterns for template instantiation depth issues.

\begin{minipage}{\columnwidth}
\begin{lstlisting}[caption={C++ template metaprogramming failure}, label={lst:template-bug}]
  // Complex template causing stack overflow
  template<typename... Args>
  struct TypeList {
      template<template<typename...> class F>
      using apply = F<Args...>;
  };
  // Chronos-1 missed: recursive instantiation depth
  // Required: SFINAE and constexpr evaluation
\end{lstlisting}
\end{minipage}

\textbf{Root Cause:} Template metaprogramming requires compile-time reasoning that exceeds current model capabilities.

\subsubsection{Case 3: Machine Learning Pipeline Data Leakage}

\textbf{Bug Description:} Validation accuracy drops 40\% in production despite 95\% test accuracy

\textbf{Failure Mode:} Chronos-1 identified standard preprocessing steps but missed subtle data leakage through feature normalization before train/test split.

\textbf{Root Cause:} Requires understanding of statistical independence and temporal data dependencies in ML pipelines.

\subsection{Failure Pattern Analysis}

\begin{figure}[H]
\centering
\resizebox{0.85\columnwidth}{!}{%
\begin{tikzpicture}
\begin{axis}[
    ybar,
    bar width=16pt,
    width=10cm,
    height=6.5cm,
    xlabel={Bug Complexity (1-10 scale)},
    ylabel={Failure Rate (\%)},
    ymin=0, ymax=105,
    xmin=0.5, xmax=10.5,
    xtick={1,2,3,4,5,6,7,8,9,10},
    xticklabel style={font=\footnotesize},
    yticklabel style={font=\footnotesize},
    xlabel style={font=\small},
    ylabel style={font=\small},
    ymajorgrids=true,
    grid style={dashed, gray!30},
    axis lines=left,
    clip=false,
    nodes near coords,
    nodes near coords style={font=\scriptsize},
    nodes near coords align={vertical}
]
\addplot[fill=blue!70, draw=blue!80] coordinates {
    (1,2.3) (2,5.1) (3,8.7) (4,15.2) (5,23.6) 
    (6,38.9) (7,52.1) (8,71.3) (9,84.7) (10,93.2)
};
\end{axis}
\end{tikzpicture}%
}
\caption{Failure rate increases exponentially with bug complexity. The system struggles with highly complex bugs (complexity 8+) that require cross-module understanding.}
\label{fig:failure-complexity}
\end{figure}

\subsection{Mitigation Strategies and Future Directions}

\textbf{1. Enhanced Retrieval for Edge Cases:}
\begin{itemize}
    \item Implement fallback to exhaustive search when confidence $< 0.3$
    \item Add cross-repository pattern matching for rare bugs
    \item Develop specialized retrievers for distributed/concurrent code
\end{itemize}

\textbf{2. Improved Understanding through Neuro-Symbolic Integration:}
\begin{itemize}
    \item Integrate theorem provers for correctness verification
    \item Add symbolic execution for path exploration
    \item Incorporate domain-specific reasoning modules
\end{itemize}

\textbf{3. Generation Safety Mechanisms:}
\begin{itemize}
    \item Mandatory test generation before fix application
    \item Rollback mechanisms for breaking changes
    \item Human-in-the-loop validation for critical systems
\end{itemize}

\begin{tcolorbox}[colback=green!5,colframe=green!40!black,title=The Bottom Line]
Despite these edge cases, Chronos-1 achieves \textbf{87.1\% debugging success} compared to 22.9\% for the next best system. Even in failure cases, Chronos-1provides actionable insights that accelerate manual debugging by 2.3x on average.
\end{tcolorbox}

\section{LIMITATIONS AND FUTURE WORK}

\subsection{Technical Limitations}

While Chronos-1 advances autonomous debugging capabilities, several technical constraints remain:

\begin{table}[H]
\centering
\caption{Summary of Technical Limitations}
\label{tab:limitations-summary}
\resizebox{\columnwidth}{!}{%
\begin{tabular}{llcc}
\toprule
\textbf{Limitation} & \textbf{Impact} & \textbf{Affected Scenarios} & \textbf{Mitigation} \\
\midrule
Extreme-scale latency & 5s+ retrieval time & Repos $>10M$ LOC & Parallel expansion \\
Memory cold start & 23\% lower success & New projects $<1K$ commits & Transfer learning \\
Non-determinism & 18.3\% variance & Distributed systems & Deterministic replay \\
Reasoning opacity & 10GB traces/bug & All debugging sessions & Selective logging \\
Memory coherence & 2.1\% conflicts & Concurrent instances & Eventual consistency \\
Dynamic languages & 41.2\% accuracy & Python/Ruby/JS & Runtime instrumentation \\
External dependencies & 31\% lower success & API/DB bugs & Mock generation \\
\bottomrule
\end{tabular}%
}
\end{table}

\textbf{Key Technical Constraints:}
\begin{itemize}
    \item \textbf{Extreme-Scale Context Latency:} O(k²n) complexity causes 5s+ retrieval times in 10M+ LOC repositories.
    
    \item \textbf{Memory Cold Start:} 23\% lower success on projects with $<1K$ commits until memory builds over 2-3 weeks.
    
    \item \textbf{Non-Determinism:} 18.3\% variance in distributed systems due to timing and resource contention.
    
    \item \textbf{Reasoning Transparency:} 10GB trace data per bug makes full interpretability infeasible.
    
    \item \textbf{Dynamic Languages:} 41.2\% accuracy on Python/Ruby/JS due to runtime metaprogramming.
    
    \item \textbf{External Dependencies:} 31\% lower success on API/database bugs without visibility into external states.
\end{itemize}

\subsubsection{Failure Mode Analysis}

\begin{table}[H]
\centering
\caption{Common failure modes, frequencies, and mitigation strategies.}
\label{tab:failure-modes}
\resizebox{\columnwidth}{!}{%
\begin{tabular}{llcc}
\toprule
\textbf{Failure Mode} & \textbf{Root Cause} & \textbf{Frequency} & \textbf{Mitigation Strategy} \\
\midrule
Incomplete fixes & Insufficient test coverage & 12.3\% & Automated test generation \\
Over-engineering & Historical pattern bias & 8.7\% & Confidence thresholding \\
Context overflow & Repository complexity & 6.2\% & Hierarchical retrieval \\
False positives & Ambiguous error messages & 4.8\% & Multi-source validation \\
Regression introduction & Side effect blindness & 3.1\% & Impact analysis expansion \\
\bottomrule
\end{tabular}%
}
\end{table}

\subsection{Research Directions: Extending the Debugging Paradigm}

Key research directions:
\begin{itemize}
    \item \textbf{Algorithmic Optimization:} Sub-quadratic retrieval for 10M+ LOC repositories
    \item \textbf{Visual Understanding:} Screenshot analysis for UI/UX debugging
    \item \textbf{Federated Learning:} Cross-organization bug patterns with privacy
    \item \textbf{Human-AI Collaboration:} Interactive debugging with feedback loops
    \item \textbf{Security Hardening:} Defense against adversarial attacks
\end{itemize}

\subsection{Deployment Architecture and Integration}

The proposed architecture extends beyond isolated debugging to comprehensive autonomous maintenance through a multi-tiered system design. The integration framework comprises:

\begin{itemize}
    \item \textbf{Continuous Monitoring Layer:} Real-time analysis of code quality metrics, security vulnerability patterns, and performance degradation indicators using static and dynamic analysis techniques
    \item \textbf{Automated Dependency Resolution:} Graph-based impact analysis for dependency updates with probabilistic risk assessment and automated rollback mechanisms
    \item \textbf{Self-Healing Pipeline Integration:} Event-driven architecture for autonomous incident response, incorporating validated patches into existing CI/CD workflows
    \item \textbf{Knowledge Synthesis Module:} Automated extraction and formalization of implicit domain knowledge through documentation generation and code pattern analysis
\end{itemize}

Our empirical studies indicate that such integrated deployment can reduce mean time to resolution (MTTR) by 67\% while maintaining a false positive rate below 3\%. Field trials with industry partners are ongoing to validate these findings at scale.

\subsection{Broader Impact}

By enabling persistently self-healing and context-aware code maintenance, Chronos-1aims to shift an industry paradigm: reducing human toil and repetitive bug resolution, freeing engineers to focus on architecture, innovation, and user demands. As we scale deployment, it is crucial to steward responsible AI governance, data privacy, and an inclusive transition for developer workforces worldwide.


\FloatBarrier

\begin{table}[H]
\centering
\caption{End-to-end debugging pipeline efficiency metrics}
\label{tab:pipeline-efficiency}
\resizebox{\columnwidth}{!}{%
\begin{tabular}{lccccc}
\hline
\textbf{System} & \textbf{Avg. Time to Fix} & \textbf{Iterations} & \textbf{Token Usage} & \textbf{Cost per Bug} & \textbf{Success Rate} \\
\hline
Chronos-1 & \textbf{4.2 min} & \textbf{2.2} & \textbf{31.2K} & \textbf{\$0.18} & \textbf{65.3\%} \\
Claude 4.5 Opus + RAG & 16.4 min & 5.2 & 128.7K & \$2.57 & 17.1\% \\
Gemini 3 Pro + RAG & 15.8 min & 4.9 & 121.3K & \$1.82 & 16.4\% \\
Claude 4.1 Opus + RAG & 18.7 min & 5.8 & 142.3K & \$2.84 & 14.2\% \\
GPT-4.1 + LangChain & 21.3 min & 6.2 & 156.7K & \$3.13 & 13.8\% \\
Gemini 2.0 Pro & 19.5 min & 5.5 & 134.8K & \$1.35 & 13.4\% \\
Human Developer & 35.8 min & 3.4 & N/A & \$29.83$^*$ & 87.2\% \\
\hline
\multicolumn{6}{l}{\footnotesize $^*$Based on average developer hourly rate of \$50}
\end{tabular}%
}
\end{table}

\FloatBarrier
\FloatBarrier
\section{CONCLUSION}

We have presented Chronos-1, a novel debugging-specific language model that addresses fundamental limitations in existing code understanding systems. Through specialized training on debugging workflows and a purpose-built architecture incorporating persistent memory and intelligent retrieval, Chronos-1 demonstrates significant improvements over general-purpose language models in automated debugging tasks.

Our comprehensive evaluation reveals that Chronos-1 achieves a 65.3\% success rate on real-world debugging benchmarks and establishes state-of-the-art performance on SWE-bench Lite with 80.33\% (241/300 instances resolved)—a 20 percentage point absolute lead over the next best system (ExpeRepair-v1.0 + Claude 4.5 Sonnet: 60.33\%) and 4-5x improvement over general-purpose models including Claude 4.5 Opus (17.1\%), Gemini 3 Pro (16.4\%), Claude 4.1 Opus (14.2\%), Claude 4.5 Sonnet (~14\%), GPT-4.1 (13.8\%), and Gemini 2.0 Pro (13.4\%). This performance gain persists despite these models achieving state-of-the-art results on code generation tasks (Claude 4.5 Opus: 74.40\%, Gemini 3 Pro: 76.2\%, Claude 4.1 Opus: 72.5\%, Claude 4.5 Sonnet: 72.7\% on SWE-bench Verified). The disparity confirms our hypothesis: debugging is fundamentally different from code generation and requires specialized architectures. On SWE-bench Lite, Chronos-1 demonstrates exceptional domain-specific performance with 96.1\% success on symbolic mathematics (sympy), 90.4\% on web frameworks (django), and 93.8\% on documentation systems (sphinx).

Key technical contributions enabling this breakthrough include: (1) domain-specific pre-training on 15 million debugging instances including stack traces, fix commits, and CI/CD logs, (2) Adaptive Graph-Guided Retrieval (AGR) that outperforms advanced RAG techniques like HyDE, Self-RAG, and FLARE by 2-3x on debugging tasks, (3) a persistent memory architecture that maintains cross-session knowledge, a capability absent in modern IDEs like Cursor and Windsurf, and (4) an autonomous debugging loop with iterative refinement based on test execution feedback.

The implications of this work extend beyond immediate debugging applications. By demonstrating that specialized architectures and training regimes can dramatically improve performance on complex software engineering tasks, we provide evidence for the viability of task-specific language models in technical domains. While Chronos-1 is built for debugging, its architecture naturally extends to long-term codebase governance ,  enabling intelligent memory of design decisions, bug patterns, and recovery heuristics. Future research directions include extending this approach to other software engineering workflows, investigating transfer learning between debugging domains, and exploring human-AI collaborative debugging frameworks.

The transition toward autonomous debugging systems raises important considerations regarding software quality assurance, developer skill evolution, and the changing nature of software maintenance. As these systems mature, careful attention must be paid to maintaining human oversight, ensuring explainability of automated fixes, and preserving the creative and architectural aspects of software development that remain fundamentally human endeavors.

The Chronos-1 model will be available in Q4 of 2025 and deploy on Kodezi~\cite{kodezi2025} OS Q1 2026. This timeline allows for additional safety testing, enterprise integration development, and establishment of responsible deployment guidelines.

\FloatBarrier
\section*{ACKNOWLEDGMENTS}

This work benefited from the feedback and real-world challenges shared by early-access engineering partners, enterprise pilot users, and the broader Kodezi community. The maintainers of open-source repositories and tooling enabled large-scale benchmarking and inspired several retrieval and memory innovations described in this paper. Insights from researchers and practitioners in the software engineering and AI communities helped refine both methodology and experimental design. Support from Kodezi’s investors enabled the sustained research and development necessary to realize Chronos-1. Continued engagement and rigorous testing from the developer community have driven Chronos-1 toward greater reliability and practical impact.

\addtolength{\textheight}{-12cm}
\IEEEtriggeratref{12}

\end{document}